\documentclass[a4paper,twoside,notitlepage,makeidx]{book}
\usepackage{a4}
\usepackage[english]{babel}
\usepackage{color,graphicx}
\usepackage{subfigure}
\usepackage{tikz-cd}
\usepackage{amsmath,amssymb,amsthm,amsfonts}
\usepackage{hyperref}
\usepackage{mathrsfs}
\usepackage{verbatim}
\usepackage{tabularx}
\usepackage{wasysym}
\usepackage{soul}
\usepackage{makeidx}

\newcommand{\norm}[1]{\left\lVert #1\right\rVert}

\newcommand{\txt}[1]{\quad\textnormal{#1}\quad}
\newcommand{\dau}[0]{\partial}
\newcommand{\eps}[0]{\epsilon}
\newcommand{\Li}[0]{\mathcal{L}}
\newcommand{\ip}[2]{\langle\: #1\, ,  #2 \:\rangle}
\newcommand{\pis}[0]{\tilde{\pi}}

\newcommand{\g}[0]{\mathfrak{g}}

\newcommand{\R}[0]{\mathbb{R}}
\newcommand{\Z}[0]{\mathbb{Z}}
\newcommand{\C}[0]{\mathbb{C}\;}
\newcommand{\N}[0]{\mathbb{N}}
\newcommand{\GL}[0]{GL}
\newcommand{\com}[0]{\circlearrowleft}
\newcommand{\hor}[0]{\mathrm{hor}}
\newcommand{\F}[0]{\mathcal{F}}

\long\def\symbolfootnote[#1]#2{\begingroup%
\def\thefootnote{\fnsymbol{footnote}}\footnote[#1]{#2}\endgroup}
\newcommand{\exeq}[0]{\stackrel{!}{=}}
\newcommand{\ra}[1]{\stackrel{#1}{\longrightarrow}}
\newcommand{\la}[1]{\stackrel{#1}{\longleftarrow}}

\newtheorem{theorem}{Theorem}[section]
\newtheorem*{theorem*}{Theorem}

\newtheorem*{lemma*}{Lemma}

\newtheorem*{claim*}{Claim}
\newtheorem{claim}[theorem]{Claim}
\newtheorem{corollary}[theorem]{Corollary}
\theoremstyle{definition}
\newtheorem*{definition*}{Definition}
\newtheorem{notation}[theorem]{Notation}
\newtheorem{remark}[theorem]{Remark}

\makeindex 

\begin{document}
\title{\vspace{80pt}
\large Bachelor's Thesis\\
\vspace{20pt} \huge \textbf{Principal Bundles and Gauge Theories}}
\author{Matthijs V\'ak\'ar\\
      Student number 3367428 \\
        Universiteit Utrecht}
\date{Utrecht, June 21, 2011}
\maketitle
\vspace{240pt}
\Large\textbf{\qquad\qquad Supervisor: Erik van den Ban}\normalsize
\thispagestyle{empty}
\clearpage

\section*{Preface}
Principal bundles are of great mathematical importance. It will be argued that, in some sense, they are the best fibre bundles for a given structure group, from which all others can be constructed. For instance, one can use one principal bundle to understand all tensor bundles of a vector bundle or one principal bundle to understand all (Dirac and non-Dirac) spinor bundles. Moreover, they encode the structure of certain nice (i.e. smooth, free, and proper) group actions on a manifold. It is clear that these objects could be studied in their own right, without any physical applications. 

On the other hand, they are arguably \emph{the} most important spaces in modern physics. They arise when we want to resolve the issue that we are trying to give an intrinsic description of nature, while all measurements made with respect to the reference frame of some observer are non-intrinsic -- one might even argue that they say as much about the observer as about the observed. They are the spaces on which certain modern field theories describing elementary particle interactions, called \emph{gauge theories}, are formulated, even if the physicists themselves are not very explicit about it.

Mathematicians often equip principal bundles with some extra geometrical structure called a principal connection. These are the natural generalisations of the affine connections we know from Riemannian geometry. They have been an object of study in pure mathematics ever since Cartan. Following the famous paper by Yang and Mills in 1954, a theory of elementary particle physics was built using precisely this kind of mathematics (without being aware of it). However, as is often the case in physics and mathematics, communication was rather sporadic and it was realised only later that the physicists had been studying precisely what the mathematicians had already worked out years before. \cite{marmat}

It should be clear that the subject of principal bundles with connections, sometimes called gauge theory, is of interest from a mathematical perspective as well as from a physical one. This is also clear from the amount of material on the subject published in both disciplines. However, little has been written that would satisfy both mathematicians and physicists - as far as I am concerned, only \cite{blegau} comes close to doing this. Although I have no illusions of filling this gap, I will try to give an introduction to the subject from both points of view. Although most results in this thesis are not original - probably with the exception of the category-theoretic ones, especially those in section \ref{catfib} - I presume that the presentation and the derivations are less conventional.\\
\\
I assume that the reader has some familiarity with manifold geometry and the classical field theories of general relativity and electromagnetism. Naturally, I first discuss the mathematics of the subject and only then proceed to the physical side of the story, of which I unfortunately cover only the tip of the iceberg, restricting myself to a brief introduction to non-second-quantised gauge theories of the Yang-Mills kind. I have tried to keep both accounts as separate as possible, although the chapter on physics rests heavily on the mathematical preliminaries. Almost all results that are proven in the chapter on mathematics have some application in physics, even if this application is not entirely clear from this thesis.

\clearpage
\tableofcontents
\clearpage
\addcontentsline{toc}{chapter}{Key Insights}
\chapter*{Key Insights}
While writing this thesis, I have had a few insights that now form the guiding thread of the text. The reader might want to bear these in mind.
\begin{enumerate}
\item Mathematical insight: Principal bundles are bundles of frames of their associated bundles. We can formulate this principle as an equivalence of categories between principal $G$-bundles and $(G,\lambda)$-fibre bundles for an effective action $\lambda$.
\item Mathematical insight: It should be possible to construct categories of $(G,\lambda)$-fibre bundles with appropriate connections. (In the case of $\GL(\Bbbk^k)$ these would be affine connections.) Then the equivalence of item 1 should extend to an equivalence of categories between principal bundles with principal connections and these $(G,\lambda)$-bundles with $(G,\lambda)$-connections. In the case of $\GL(\Bbbk^k)$ it is clear from this thesis that this is true and we should easily be able to prove the general claim in an analogous way.
\item Didactic insight: Using this equivalence, we can view principal bundles with principal connections (which arise in particle physics) as a generalisation of vector bundles with affine connections (which arise in general relativity), the generalising step being the replacement of $\GL(\Bbbk^k)$ by an arbitrary Lie group $G$.
\item Physical insight: Because of item 1, principal bundles are the natural context for many physical gauge theories. Principal connections naturally arise when we try to give a geometric description of gauge-invariant interactions. Principal bundles with connections may serve as a unifying language in physics. Not only do they provide a mathematically rigorous framework for first-quantised gauge field theories, but they also give geometric interpretations for the physicists' calculations, helping to create an overview of this vast subject. \end{enumerate}
\clearpage
\qquad \vspace{50pt}\\
\\
\noindent \Huge \textbf{Notation and Conventions} \normalsize
\addcontentsline{toc}{chapter}{Notation and Conventions}\vspace{40pt}\\
Throughout the text, maps and actions will be assumed to be smooth, i.e. $C^\infty$, unless stated otherwise.\\
\\
At times, we use the notation $m\in M$, $g\in G$, $s\in S$, $f\in F$, $p\in P$, $v\in V$, without further comment. Moreover, principal connection $1$-forms are denoted by $\omega$ and equivariant maps in $C^\infty(P,W)$ from a principal $G$-bundle to a vector space with a $G$-action are denoted by $\psi$.

If $F\ra{\pi}M$ is a fibre bundle and $m\in M$, $F_m$ denotes the fibre of $\pi$ over $m$: $\pi^{-1}(m)$. Moreover, we use the conventional abuse of notation, sometimes writing $F$ for a fibre bundle instead of $F\ra{\pi}M$. A second abuse of notation is the dot notation for a group action if there is no ambiguity in which action we mean.\\
\\
A glossary of notation can be found at the end of this thesis.

\chapter{Fibre Bundles and Connections}

\section{Categories of Fibre Bundles}
\label{catfib}
I know from personal experience that the notion of a fibre bundle can be a very confusing one, if it is not introduced properly. This is because there are many different but related variants, and their definitions include many subtleties. Moreover, it is very hard to find a good account of these in the literature. I will therefore begin this thesis with some very strict categorical definitions in this domain, most of which I have had to formulate myself, since I could not find them in the literature\footnote{The book that comes closest is the classic reference \cite{stetop} by Steenrod.}.

These definitions allow us to interpret the associated bundle construction as defining an equivalence of categories (for an effective action). This result sets the tone for this entire thesis. If one reads between the lines in many books, it seems that some form of this result must be known, although it has not been (as far as I know) stated explicitly. Starting from this abstract point of view, I will gradually specialise to more specific variants of fibre bundles.

\begin{definition*}[Fibre bundle]A \emph{fibre bundle\index{fibre bundle} with standard fibre $S$} is the combination of
\begin{enumerate}
\item a smooth\footnote{Most of this theory can be developed in equal generality for the case of topological spaces, instead of smooth manifolds. However, from now on, we shall assume that we are working with the latter, since the smooth structure is essential to formulate the equations of many physical theories.} manifold $S$, called the \emph{fibre},
\item a map of manifolds $F\ra{\pi}M$, called the \emph{bundle projection} from the \emph{total space} $F$ onto the \emph{base space} $M$,
\end{enumerate}
such that for each $m\in M$, there exists a \emph{local trivialisation}, i.e. a neighbourhood $U$ of $m$ in $M$, such that $\pi^{-1}(U)\cong U \times S$ via a fibre-respecting diffeomorphism:
\begin{center}
\begin{tikzcd}[column sep=large,row sep=large]
\pi^{-1}(U) \arrow[r,"\psi"] \arrow[d,"{\pi|_{\pi^{-1}(U)}}"'] \arrow[dr,phantom,"\com",description] & U\times S \arrow[d,"\pi_1"]\\
U \arrow[r,equal] & U.
\end{tikzcd}
\end{center}
We define a \emph{map of fibre bundles} $\pi\ra{k}\pi'$ to be a tuple of maps $F\ra{k}F'$ and $M\ra{k_M}M'$ such that

\begin{center}
\begin{tikzcd}[column sep=large,row sep=large]
F \arrow[r,"k"] \arrow[d,"\pi"'] \arrow[dr,phantom,"\com",description] & F' \arrow[d,"{\pi'}"]\\
M \arrow[r,"{k_M}"'] & M'.
\end{tikzcd}
\end{center}

By a \emph{map of fibre bundles over $M$}, we mean a map $\pi\ra{k}\pi'$ between fibre bundles for which $M=M'$. Note that for some authors $k_M$ might still be nontrivial; here I mean the more restrictive case in which $k_M=\mathrm{id}_M$.

This gives a category $FB$ of fibre bundles and a category $FB_M$ of fibre bundles over $M$, where composition is defined as one would expect.
\end{definition*}
One often has preferred local trivialisations $(U_\alpha,\psi_\alpha)$ on a fibre bundle. A set of such local trivialisations such that $(U_\alpha)$ covers $M$ is called a \emph{fibre bundle atlas}\index{fibre bundle atlas}. Then, for each $\alpha,\beta$ there exists a unique function $(U_\alpha\cap U_\beta)\times S\ra{\phi_{\alpha\beta}}S$, such that $\psi_\alpha\circ\psi_\beta^{-1}(m,s)=(m,\phi_{\alpha\beta}(m,s))$. Equivalently, we find \emph{transition functions} $U_\alpha\cap U_\beta\ra{\phi_{\alpha\beta}}\mathrm{Diff}\; S$, with $\phi_{\alpha\beta}(m)(s)=\phi_{\alpha\beta}(m,s)$. Obviously, these maps $(\phi_{\alpha\beta})$ satisfy the \emph{\v Cech cocycle condition}\index{\v Cech cocycle condition} $\phi_{\alpha\beta}(m)\cdot \phi_{\beta\gamma}(m)=\phi_{\alpha\gamma}(m)$. A fibre bundle together with such a fibre bundle atlas is called a fibre bundle with coordinates or, briefly, a \emph{coordinate fibre bundle}\index{coordinate fibre bundle}.

\subsection{Coordinate $(G,\lambda)$-Bundles}
\label{cocyc}
Unfortunately, the group $\mathrm{Diff}\; S$ is difficult to handle. For instance, we run into trouble when we try to endow it with a smooth structure. It turns out to have a natural structure of an infinite-dimensional manifold (a Fr\'echet manifold if $S$ is compact), in which, indeed, the group operations are smooth maps. \cite{micman} This would lead to many technicalities. Fortunately, for many important examples of fibre bundles these transition functions only take values in a small part of $\mathrm{Diff}\; S$. To be precise, often, we can find a Lie group $G$, together with an action $G\times S\ra{\lambda} S$, such that the $\phi_{\alpha\beta}$ factor through $\lambda$. That is, we can find maps $U_\alpha\cap U_\beta\ra{g_{\alpha\beta}}G$, such that $\lambda\circ g_{\alpha\beta}=\phi_{\alpha\beta}$:
\begin{center}
\begin{tikzcd}[column sep=large,row sep=large]
U_\alpha\cap U_\beta \arrow[rrrr,"{\phi_{\alpha\beta}}"] \arrow[ddrr,"{g_{\alpha\beta}}"'] & & & & \mathrm{Diff}\; S\\
& & \com & &\\
& & G \arrow[uurr,"\lambda"'] & &.
\end{tikzcd}
\end{center} 
The special case of an \emph{effective} action\index{effective group action} $\lambda$, i.e. where we can interpret $\lambda$ as an (algebraic) embedding $G\ra{\lambda} \mathrm{Diff}\; S$, will turn out to be of special importance since the fibre bundles corresponding to these actions can be characterised in a coordinate-free way (by a maximal atlas). This leads to the following definitions.
\begin{definition*}[$G$-cocycle]Let $M$ be a manifold and $G$ a Lie group. We define a $G$-\emph{cocycle of transition functions} on it to be the combination of
\begin{enumerate}
\item an open cover $(U_\alpha)$,
\item a family of smooth mappings $U_\alpha\cap U_\beta\ra{g_{\alpha\beta}}G$, which satisfies the cocycle condition $g_{\alpha\beta}(m)\cdot g_{\beta\gamma}(m)=g_{\alpha\gamma}(m)$, for all $m\in U_\alpha\cap U_\beta\cap U_\gamma$ and $g_{\alpha\alpha}(m)=e$.
\end{enumerate}
We understand a \emph{map of cocycles} to consist of
\begin{enumerate}
\item a map $M\ra{h}M'$,
\item a tuple of cocycles $(U_\alpha,g_{\alpha\beta})$ on $M$ and $(U'_\alpha,g'_{\alpha\beta})$ on $M'$,
\item maps $U_\beta\cap h^{-1}(U'_\alpha)\ra{h_{\alpha\beta}}G$,
\end{enumerate}
 such that $h_{\alpha\beta}(m)\cdot g_{\beta\gamma}(m)=h_{\alpha\gamma}(m)$, for $m\in U_\gamma\cap U_\beta \cap h^{-1}(U'_\alpha)$ and $g'_{\alpha\beta}(h(m))\cdot h_{\beta\gamma}(m)=h_{\alpha\gamma}(m)$, for $m\in U_\gamma\cap h^{-1}(U'_\beta\cap U'_\alpha)$.\\
\\
It is easily seen that the $G$-cocycles on smooth manifolds form a category $CC^G$, where the composition of $(h,(U_\alpha,g_{\alpha\beta}),(U'_\alpha,g'_{\alpha\beta}),h_{\alpha\beta})$ and $(i,(U'_\alpha,g'_{\alpha\beta}),(U''_\alpha,g''_{\alpha\beta}),i_{\alpha\beta})$ is defined to be
$(i\circ h,(U_\alpha,g_{\alpha\beta}),(U''_\alpha,g''_{\alpha\beta}),(i\circ h)_{\alpha\beta})$, with $(i\circ h)_{\alpha\gamma}(m):=i_{\alpha\beta}(h(m))\cdot h_{\beta\gamma}(m)$\footnote{This is immediately seen to be independent of the choice of $\beta$.}.
\end{definition*}
\begin{definition*}[Coordinate $(G,\lambda)$-bundle structure] Suppose $F\ra{\pi}M$ is a fibre bundle, with standard fibre $S$, $G$ a Lie group and $G\times S\ra{\lambda} S$ is a left action of $G$ on $S$. A $(G,\lambda)$-fibre bundle with coordinates or, briefly, a \emph{coordinate $(G,\lambda)$-bundle structure}\index{coordinate $(G,\lambda)$-bundle} on $\pi$ for the action $\lambda$ consists of
\begin{enumerate}
\item a $G$-cocycle $(U_\alpha,g_{\alpha\beta})$ on $M$,
\item a fibre bundle atlas $(U_\alpha,\psi_\alpha)$,
\end{enumerate}
such that $\psi_\alpha\circ\psi_\beta^{-1}(m,s)=(m,\lambda(g_{\alpha\beta}(m),s))$. A coordinate $(G,\lambda)$-bundle is also called a coordinate $G$-bundle if there is no ambiguity about the action $\lambda$.\\
\\
We understand a \emph{map of coordinate $G$-bundles} to consist of
\begin{enumerate}
\item a map of cocycles $(h,(U_\alpha,g_{\alpha\beta}),(U'_\alpha,g'_{\alpha\beta}),h_{\alpha\beta})$,
\item a map of fibre bundles:
\begin{center}
\begin{tikzcd}[column sep=large,row sep=large]
F \arrow[r,"k"] \arrow[d,"\pi"'] & F' \arrow[d,"{\pi'}"]\\
M \arrow[r,"{k_M}"'] & M',
\end{tikzcd}
\end{center}
\end{enumerate}
such that $k_M=h$ and $\psi'_\alpha \circ k \circ \psi_\beta^{-1}(m,s)=(k_M(m),\lambda(h_{\alpha\beta}(m),s))$.\\
\\
Again, we can form a category, which we shall call $CFB^\lambda$, where the composition is understood to be the simultaneous composition of maps of cocycles and of fibre bundles.
\end{definition*}
Note that, in both categories $CC^G$ and $CFB^\lambda$, the isomorphisms are precisely those maps that are diffeomorphisms on the base space.

One can wonder if it is possible to construct a fibre bundle from an arbitrary $(G,\lambda)$-cocycle. The (affirmative) answer is given by the well-known fibre bundle construction theorem. Here, I formulate a categorical version of this theorem.
\begin{theorem}[Fibre bundle construction theorem]\index{fibre bundle construction theorem}\label{cfbcon} Let $G\times S\ra{\lambda} S$ be a smooth group action. We then have an equivalence of categories
$$CC^G\mathop{\leftrightarrows}^{CU^\lambda}_{CB^\lambda} CFB^\lambda.$$
\end{theorem}
\begin{proof}We define the functor $CU^\lambda$ as the obvious forgetful functor that sends a coordinate $(G,\lambda)$-bundle $(\pi,U_\alpha,\psi_\alpha,g_{\alpha\beta})$ to the cocycle $(U_\alpha,g_{\alpha\beta})$ and a map of coordinate bundles to its map of cocycles. Functoriality holds by definition\footnote{Moreover, I defined the composition of maps of cocycles in this way, so that functoriality would hold.} of the composition in $CFB^\lambda$.\\
\\
The definition of the functor $CB^\lambda$ is the non-trivial part of the equivalence. We proceed by first defining it on objects. Suppose we are given a cocycle $(U_\alpha,g_{\alpha\beta})\in CC^G$ on $M$. Write $\Sigma$ for the topological coproduct $\bigsqcup_\alpha U_\alpha\times S$. To be explicit, we understand the underlying set to be $\bigcup_\alpha \{\alpha\}\times U_\alpha\times S$. On this topological space, we define an equivalence relation:
$$\{\alpha\}\times U_\alpha\times S\ni (\alpha,m,s)\sim (\beta,m',s')\in \{\beta\}\times U_\beta\times S\txt{iff} (m,s)=(m',g_{\alpha\beta}(m)\cdot s'). $$
Note that the fact that this is an equivalence relation follows directly from the cocycle condition. We can divide out this relation and give $F:=\Sigma/\sim$ the quotient topology. Note that there exists a natural projection $\Sigma\ra{\pi'}M$, which we obtain by the fact that we have a projection $U_\alpha\times S\ra{\pi_1}U_\alpha\ra{i_\alpha}M$, where $\pi_1$ is the projection onto $U_\alpha$ and $i_\alpha$ is the obvious inclusion, for each term in the coproduct. Since points that are related by $\sim$ project along $\pi'$ to the same element, this projection descends to the quotient: $F\ra{\pi}M$.

The trivialisations $(\psi_\alpha)$ of the bundle are then defined as one would expect: $\psi_\alpha([\alpha,m,s]):=(m,s)$, if we write $[\alpha,m,s]\in F$ for the equivalence class of $(\alpha,m,s)\in \Sigma$. This map is well-defined since $(\alpha,m,s)\sim (\alpha,m',s')$ iff $(m,s)=(m',s')$. It is a homeomorphism since it is an inverse to the restriction of the quotient map. 

To see that $F$ is a Hausdorff space, suppose that $f,f'\in F$ are two distinct points. If $\pi(f)\neq \pi(f')$, we can separate them in $M$, which is Hausdorff, and take the inverse image under $\pi$ to obtain distinct neighbourhoods in $F$. Otherwise, we can separate $f$ and $f'$ in one chart $\psi_\alpha$, since we know $S$ to be Hausdorff. To see that $F$ is second countable, note that $M$ is so. Therefore, we can find a countable refinement $(U'_i)$ of $(U_\alpha)$. Then, each $U'_i$ lies in the domain of a local trivialisation $\psi_{\alpha_i}$. Since $S$ is second countable, we find a countable basis $(V_j)$ for its topology. We conclude that $\cup_{i,j} (\psi_{\alpha_i}^{-1}(U'_i\times V_j))$ is a basis for the topology of $F$. It is clearly countable.

We conclude that we have constructed a topological fibre bundle $F\ra{\pi}M$. We want to show that $\pi$ has structure group $(G,\lambda)$. For this, suppose that $U_\alpha\cap U_\beta\neq \emptyset$. Then, $m\in U_\alpha\cap U_\beta$ implies that $\psi_\beta^{-1}(m,s)\in \pi^{-1}(U_\alpha\cap U_\beta)=\pi^{-1}(U_\alpha)\cap\pi^{-1}(U_\beta)$. Since $\psi_\beta^{-1}(m,s)\in \pi^{-1}(U_\alpha)$, we can write it as $[\alpha,m',s']$ for some $m'\in U_\alpha,\; s'\in S$. Then, however, $m=m'$ and $s'=g_{\alpha\beta}(m)\cdot s$, by definition of $\sim$. We conclude that $\psi_\beta^{-1}(m,s)=[\alpha,m,g_{\alpha\beta}(m)\cdot s]$ and therefore $\psi_\alpha\circ\psi_\beta^{-1}(m,s)=(m,g_{\alpha\beta}(m)\cdot s)$. Now, the smooth structure on $F$ is obtained from that on $S$ by the maps $(\psi_\alpha)$. This gives a consistent notion of smoothness precisely because the transition functions act by the (smooth) action of a Lie group. Therefore $(\pi,U_\alpha,\psi_\alpha,g_{\alpha\beta})$ is a $(G,\lambda)$-coordinate bundle, which we denote by $CB^\lambda(U_\alpha,g_{\alpha\beta})$.

The next thing to consider is the definition of $CB^\lambda$ on arrows. Suppose we are given a map of cocycles $(h,(U_\alpha,g_{\alpha\beta}),(\tilde{U}_\alpha,\tilde{g}_{\alpha\beta}),h_{\alpha\beta})$
in $CC^G$. Write $F\ra{\pi}M$ and $\tilde{F}\ra{\pis}\tilde{M}$ respectively for the fibre bundles corresponding to these cocycles by $CB^\lambda$, as described above, and adopt similar notations $(\psi_\alpha)$ and $(\tilde{\psi}_\alpha)$ for their trivialisations. Suppose $f\in \pi^{-1}(U_\beta\cap h^{-1}(\tilde{U}_\alpha))$ and $f=\psi_\beta^{-1}(m,s)$. Then we can define the smooth function $\pi^{-1}(U_\beta\cap h^{-1}(\tilde{U}_\alpha))\ra{k_{\alpha\beta}}\tilde{F}$ by $\tilde{\psi}_\alpha \circ k_{\alpha\beta}(f)=\tilde{\psi}_\alpha \circ k_{\alpha\beta}\circ \psi_\beta^{-1}(m,s):=(h(m),h_{\alpha\beta}(m)\cdot s)$. Note that this defines a smooth function since $\tilde{\psi}_\alpha$ and $\psi_\beta$ are diffeomorphisms. (Explicitly, $k_{\alpha\beta}(f)=\tilde{\psi}_\alpha^{-1}(h(m),h_{\alpha\beta}(m)\cdot \pi_S(\psi_\beta(f))),$ if $m=\pi(f)$.)

Now, the sets $(\pi^{-1}(U_\beta\cap h^{-1}(\tilde{U}_\alpha)))_{\alpha,\beta}$ form an open cover of $F$. If we verify that the functions $k_{\alpha\beta}$ agree on overlaps of their domains, the sheaf property of smooth functions will give us a unique smooth function $F\ra{k}\tilde{F}$, such that $k|_{\pi^{-1}(U_\beta\cap h^{-1}(\tilde{U}_\alpha))}=k_{\alpha\beta}$. Then
\begin{align*}k_{\alpha\beta}(f)&=\tilde{\psi}_\alpha^{-1}(h(m),h_{\alpha\beta}(m)\cdot g_{\beta\gamma}(m)\cdot \pi_S(\psi_\gamma(f)))\\
&=\tilde{\psi}_\alpha^{-1}(h(m),h_{\alpha\gamma}(m)\cdot \pi_S(\psi_\gamma(f)))\\
&=k_{\alpha\gamma}(f)\\
&=\tilde{\psi}_\delta^{-1}(h(m),\tilde{g}_{\delta\alpha}(h(m))\cdot h_{\alpha\gamma}(m)\cdot \pi_S(\psi_\gamma(f)))\\
&=\tilde{\psi}_\delta^{-1}(h(m),h_{\delta\gamma}(m)\cdot \pi_S(\psi_\gamma(f)))\\
&=k_{\delta\gamma}(f),
\end{align*}
if $f\in \pi^{-1}(U_\gamma\cap U_\beta \cap h^{-1}(\tilde{U}_\alpha\cap\tilde{U}_\delta))=\pi^{-1}(U_\gamma\cap h^{-1}(\tilde{U}_\delta))\cap \pi^{-1}(U_\beta\cap h^{-1}(\tilde{U}_\alpha))$, $m=\pi(f)$ and $\pi_S$ denotes the projection onto $S$. (Here we have used the compatibility of $h_{\alpha\beta},\; g_{\alpha\beta}$ and $\tilde{g}_{\alpha\beta}$, which holds because we started with a map of cocycles.)
This gives us our smooth map $F\ra{k}\tilde{F}$. Obviously, $\pis\circ k=h\circ\pi=k_M\circ \pi$ and $h=k_M$. We conclude that $k$ gives us a map of coordinate $(G,\lambda)$-bundles. This will be our definition for $CB^\lambda$ on arrows: \begin{align*}
& CB^\lambda(h,(U_\alpha,g_{\alpha\beta}),(\tilde{U}_\alpha,\tilde{g}_{\alpha\beta}),h_{\alpha\beta}):=(h,(U_\alpha,\psi_\alpha,g_{\alpha\beta}),(\tilde{U}_\alpha,\tilde{\psi}_\alpha,\tilde{g}_{\alpha\beta}),h_{\alpha\beta},k).
\end{align*}
Again, functoriality is a simple consequence of the definition of composition. \\
\\
Finally, we check that these two functors indeed define an equivalence of categories. Obviously, $CU^\lambda\circ CB^\lambda=\mathrm{id}_{CC^G}$. Thus, what remains to verify is that we can find a natural isomorphism $CB^\lambda\circ CU^\lambda\stackrel{\eps}{\Longrightarrow} \mathrm{id}_{CFB^\lambda}$.

Let $\F=(U_\alpha,\psi_\alpha,g_{\alpha\beta},F\ra{\pi}M)\in CFB^\lambda$. Write
$\F'=(U_\alpha,\psi_\alpha',g_{\alpha\beta},F'\ra{\pi'}M)$ and $\kappa '$ for the images of objects $\F$ and arrows $\kappa$, respectively, of $CFB^\lambda$ under the functor $CB^\lambda\circ CU^\lambda$. We define the components of $\eps$ as follows. Let them be the identity map of cocycles together with the following maps $F'\ra{\delta_\F}F$, for each $\F$. We first define the diffeomorphism $\delta_{\F_\alpha}:=\psi_\alpha^{-1}\circ\psi_\alpha':\pi'^{-1}(U_\alpha)\ra{}\pi^{-1}(U_\alpha)$. Note that for $f'\in \pi'^{-1}(U_\alpha\cap U_\beta)$, if $m'=\pi'(f')$ and $s=\pi_S(\psi_\beta'(f'))$, then
\begin{align*}
\delta_{\F_\alpha}(f')&=\psi_\alpha^{-1}(\psi_\alpha'(f'))\\
&=\psi_\alpha^{-1}(m',g_{\alpha\beta}(m')\cdot s),
\end{align*}
while
\begin{align*}
\psi_\alpha(\delta_{\F_\beta}(f'))&=\psi_\alpha\circ\psi_\beta^{-1}(m',s)\\
&=(m',g_{\alpha\beta}(m')\cdot s).
\end{align*}
Since $\psi_\alpha$ is injective, $\delta_{\F_\alpha}(f')=\delta_{\F_\beta}(f')$ on overlaps.
Note that the open sets $(\pi'^{-1}(U_\alpha))$ cover $F'$. If we consider $\delta_{\F_\alpha}$ as a map $\pi'^{-1}(U_\alpha)\ra{}F$, the sheaf property of smooth maps gives us a unique map $F'\ra{\delta_\F}F$, such that $\delta_\F|_{\pi'^{-1}(U_\alpha)}=\delta_{\F_\alpha}$. This map is obviously a bijection and since each $\delta_{\F_\alpha}$ is a diffeomorphism, so is $\delta_\F$. Moreover, we easily see that $\pi\circ\delta_\F=\pi'$. So $\delta_\F$ is an isomorphism of fibre bundles. We conclude that the combination $\epsilon_\F$ of the identity map of cocycles and $\delta_\F$ is an isomorphism of coordinate $(G,\lambda)$-bundles.

Naturality of $\epsilon$ is also obvious. Since $\epsilon$ is the identity on cocycles, we only verify that $\delta_\F$ is natural in $\F$. Let $\kappa:=(h,\F:=(U_\alpha,\psi_\alpha,g_{\alpha\beta},F\ra{\pi}M),\tilde{\F}:=(\tilde{U}_\alpha,\tilde{\psi}_\alpha,\tilde{g}_{\alpha\beta},\tilde{F}\ra{\pis}\tilde{M}),h_{\alpha\beta},F\ra{k}\tilde{F})$ be an arrow in $CFB^\lambda$. We check that $\delta_{\tilde{\F}}\circ k'=k\circ \delta_\F$:
\begin{center}
\begin{tikzcd}[column sep=large,row sep=large]
F' \arrow[r,"{k'}"] \arrow[d,"{\delta_\F}"'] \arrow[dr,phantom,"\com",description] & \tilde{F}' \arrow[d,"{\delta_{\tilde{\F}}}"]\\
F \arrow[r,"k"'] & \tilde{F}.
\end{tikzcd}
\end{center}
Let $f'\in F'$, such that $m'=\pi'(f')\in U_\beta$. Then
\begin{align*}\tilde{\psi}_\alpha\circ\delta_{\tilde{\F}}\circ k'(f')&=\tilde{\psi}_\alpha\circ\tilde{\psi}_\gamma^{-1}\circ\tilde{\psi}_\gamma'\circ k'\circ \psi_\beta'^{-1}(m',\pi_S(\psi_\beta'(f')))\\
&=\tilde{\psi}_\alpha\circ\tilde{\psi}_\gamma^{-1}(h(m'),h_{\gamma\beta}(m')\cdot \pi_S(\psi_\beta'(f')))\\
&=(h(m'),\tilde{g}_{\alpha\gamma}(h(m'))\cdot h_{\gamma\beta}(m')\cdot \pi_S(\psi_\beta'(f')))\\
&=(h(m'),h_{\alpha\beta}(m')\cdot \pi_S(\psi_\beta'(f')))\\
&=\tilde{\psi}_\alpha\circ k\circ \psi_\beta^{-1}(m', \pi_S(\psi_\beta'(f')))\\
&=\tilde{\psi}_\alpha\circ k\circ \psi_\beta^{-1}\circ\psi_\beta'(f')\\
&=\tilde{\psi}_\alpha\circ k\circ \delta_\F(f').\\
\end{align*}
Since $\tilde{\psi}_\alpha$ is a diffeomorphism, naturality follows.\end{proof}
This is a remarkable result. It means that any two categories $CFB^\lambda$ and $CFB^\mu$ are equivalent for any two $G$-actions $\mu$ and $\lambda$ on a manifold. We may even be dealing with a situation in which $\mu$ and $\lambda$ are actions on two different manifolds. In this sense, the cocycle structure actually encodes precisely all relevant information about the coordinate $G$-bundle. 

For each Lie group $G$, there is one particularly easy canonical choice for a smooth $G$-action, namely the action $L$ of $G$ on itself by left multiplication. We can take the category $CFB^L$ as a representative of the equivalence class of categories equivalent to $CC^G$. The objects of this category are known as \emph{coordinate principal $G$-bundles}. (This may explain the name.) Therefore we also write $CPB^G$ for $CFB^L$.

\subsection{$(G,\lambda)$-Bundles}
Since, in modern differential geometry, we would like our theory to be coordinate-independent, we say that some coordinate $(G,\lambda)$-bundle structures are equivalent and describe the same $(G,\lambda)$-fibre bundle. Compare this with manifold theory, where a differentiable manifold is defined as a topological space, together with an equivalence class of coordinate atlases. Similarly, we can also define an equivalence relation on the coordinate $(G,\lambda)$-bundle structures for a given $G\times S\ra{\lambda}S$. We say that two coordinate $(G,\lambda)$-bundle structures $(U_\alpha,\psi_\alpha,g_{\alpha\beta})$ and $(\tilde{U}_\alpha,\tilde{\psi}_\alpha,\tilde{g}_{\alpha\beta})$ for the fibre bundle $F\ra{\pi}M$ are \emph{strictly equivalent} (in the terminology of \cite{stetop}) if there exist smooth maps $U_\beta\cap \tilde{U}_\alpha\ra{t_{\alpha\beta}}G$, such that $\tilde{\psi}_\alpha\circ\psi_\beta^{-1}(m,s)=(m,t_{\alpha\beta}(m)\cdot s)$, for all $s\in S$. Such an equivalence class is called a $(G,\lambda)$-fibre bundle. One can compare these with the equivalence classes of atlases that correspond to manifold structures on a topological space. Again, one more often speaks simply of $G$-bundles and considers the action $\lambda$ to be understood.

These objects are spaces that are ubiquitous in theoretical physics. They are of special importance in studying classical field theories. Therefore, it is useful to define a category $FB^\lambda$ which has these $(G,\lambda)$-fibre bundles as objects. A natural definition for the arrows of this category is that they are equivalence classes of arrows in $CFB^\lambda$ corresponding to the same map $F\ra{k}F'$ of fibre bundles. Of course, again, the isomorphisms in $FB^\lambda$ are precisely those maps that are diffeomorphisms on the base space. Note that this gives us an obvious quotient functor $CFB^\lambda\ra{Q^\lambda}FB^\lambda$ which `forgets' the particular cocycle structure, i.e. it sends each object and arrow to its equivalence class in the sense described above. Note also that we have the obvious functorial embedding $FB^\lambda\ra{}FB$. Of course, this also gives us a category $PB^G:=FB^L$ of \emph{principal $G$-bundles}\index{principal bundle}.

\subsection{Principal Bundles}
Note that, since we also have a right $G$-action $R$ on $G$ that commutes with $L$, we can define a global right $G$-action $\rho$ on $P$, if $P\ra{\pis}M$ is a principal $G$-bundle. This is the action that looks like right multiplication in a local trivialisation. Obviously, this action, called the \emph{principal right action}\index{principal right action} restricts to a transitive, free action on each fibre. Therefore, $G$ acts freely on $P$. Moreover, this action on $P$ is easily seen to be proper. These properties turn out precisely to characterise the fact that a fibre bundle $P\ra{\pis}M$ is a principal $G$-bundle. \cite{banlie}
\begin{theorem}[Quotient manifold theorem]\index{quotient manifold theorem} Let $\rho$ be a right action of a Lie group $G$ on a manifold $P$. Then the following conditions are equivalent:
\begin{enumerate}
\item The action is proper and free. 
\item The quotient map $P\ra{\pis}P/G$ is a principal fibre bundle\footnote{In particular, since the action is proper and free, the quotient space $P/G$ has a unique smooth structure such that the quotient map $P\ra{\pis}P/G$ is a smooth submersion.} with $\rho$ as a principal right action.
\end{enumerate}
\end{theorem}
It is easier to see what maps of principal bundles are. Of course, they are a special kind of maps of fibre bundles between principal bundles. The property that precisely captures the fact that they are maps of principal bundles is that they intertwine both principal right actions.\footnote{It should be noted that, although this is a common definition of a map of principal bundles and it is used in, for example, the standard reference on fibre bundles, \cite{husfib}, the standard reference on principal bundles, \cite{kobfou}, uses a broader class of maps. They define a map of principal bundles to be the combination of a map $\pis\ra{k}\pis'$ of fibre bundles from the principal $G$-bundle $\pis$ to the principal $G'$-bundle $\pis'$ and a homomorphism $G\ra{f}G'$ of groups, such that $k\circ \rho_{g}=\rho'_{f(g)}\circ k$, for all $g\in G$. (Here $\rho$ and $\rho'$ denote both principal right actions.) Let us denote the category of all principal fibre bundles with this larger collection of maps by $PB$. 

This category has some advantages as it can be used to describe more phenomena, like maps between principal bundles with different structure groups. As a consequence, the restriction of $PB$ to some fixed base space $M$ on which the arrows are demanded to be the identity, $PB_M$ has finite products. Indeed, these products arise in Yang-Mills theories. 

At this point, however, I see no possibility to adjust the definition of a map of $(G,\lambda)$-bundles to fit this definition. For instance, I cannot see how the associated bundle functor, which can be defined for a $G'$-action $\lambda$ and a map $(P,G)\ra{(k,f)}(P',G')$ to be the obvious map $P[\lambda\circ f]\ra{k[\lambda]}P'[\lambda]$, would extend to an equivalence from this category.} To see this, note that for a map of principal $G$-bundles 
\begin{center}
\begin{tikzcd}[column sep=large,row sep=large]
P \arrow[r,"k"] \arrow[d,"\pis"'] & P' \arrow[d,"{\pis'}"]\\
M \arrow[r,"{k_M}"'] & M',
\end{tikzcd}
\end{center} the following holds if we choose corresponding principal $G$-bundle atlases $(U_\alpha,\psi_\alpha,g_{\alpha\beta})$ and $(U'_\alpha,\psi'_\alpha,g'_{\alpha\beta})$:
\begin{align*}\psi_\alpha'(k(\rho(\psi_\beta^{-1}(m,g_1),g_2)))
&=\psi_\alpha'\circ k\circ \psi_\beta^{-1}(m,g_1\cdot g_2)\\ 
&=(k_M(m),h_{\alpha\beta}(m)\cdot (g_1\cdot g_2))\\
&=(k_M(m),(h_{\alpha\beta}(m)\cdot g_1)\cdot g_2)\\
&=\psi_\alpha'(\rho(\psi_\alpha'^{-1}(\psi_\alpha'(k(\psi_\beta^{-1}(m,g_1)))),g_2))\\
&=\psi_\alpha'(\rho(k(\psi_\beta^{-1}(m,g_1)),g_2)),
\end{align*}
so $k$ intertwines the principal right actions (which I both denoted by $\rho$). For the converse statement, note that the above also tells us how we should define $h_{\alpha\beta}$ in terms of $k$. It is easily checked that $h_{\alpha\beta}$ is a map of cocycles if $k$ intertwines the actions.

\subsection{A Coordinate-Free Equivalence}
Note that, in the case of an effective action $G\times S\ra{\lambda}S$, the strict equivalence of $(G,\lambda)$-bundle structures can also be formulated as the statement that their union is again a $(G,\lambda)$-bundle structure with respect to the $G$-valued transition functions which are uniquely determined since $G$ acts effectively. Therefore we can associate its maximal element to each equivalence class (given by the total union). We can also lift a map of $(G,\lambda)$-fibre bundles in a canonical way to a map between the maximal representatives of the equivalence classes. Indeed, define $(m,h_{\alpha\beta}(m)\cdot s):=\tilde{\psi}_\alpha \circ k \circ \psi_\beta^{-1}(m,s)$, if $k$ is the given fibre bundle map and $(U_\alpha,\psi_\alpha)$ and $(\tilde{U}_\alpha,\tilde{\psi}_\alpha)$ are the trivialisations of the maximal atlases on both $G$-bundles. (Note that this defines $h_{\alpha\beta}$ uniquely, since $\lambda$ is effective.) Therefore, we obtain a canonical section $FB^\lambda\ra{\Sigma^\lambda}CFB^\lambda$ of the quotient functor $Q^\lambda$. This will lead to a natural coordinate-free treatment of $(G,\lambda)$-fibre bundles when $\lambda$ is effective. 

In light of this discussion on a coordinate-free notion of a $(G,\lambda)$-fibre bundle, it is natural to wonder if we can find an atlas-independent form of theorem \ref{cfbcon}. We can hope to succeed since the construction of $\Sigma^\lambda$ is canonical. It turns out that we are at least able to obtain such a form for the equivalence between $CFB^\lambda$ and $CFB^\mu$ for $\lambda,\mu$ effective $G$-actions. That is, we have the following.
\begin{theorem}\label{fbcon} Let $G\times S\ra{\lambda} S$ and $G\times S'\ra{\mu}S'$ be effective actions, i.e. $\lambda,\mu$ define (algebraic) embeddings of groups $G\hookrightarrow \mathrm{Diff}\; S,\; \mathrm{Diff}\; S'$. We then have an equivalence of categories
$$FB^\mu\mathop{\leftrightarrows}^{C^{\lambda\mu}}_{C^{\mu\lambda}}  FB^\lambda.$$
\end{theorem}
\begin{proof}The idea of the proof will be to use the functors of theorem \ref{cfbcon} to construct the equivalence: 
\begin{center}
\begin{tikzcd}[column sep=large,row sep=large]
CFB^\mu \arrow[r,leftarrow,shift left=0.7ex,"{CB^\mu}"] \arrow[r,shift right=0.7ex,"{CU^\mu}"'] \arrow[d,two heads,"{Q^\mu}"'] & CC^G \arrow[r,leftarrow,shift left=0.7ex,"{CU^\lambda}"] \arrow[r,shift right=0.7ex,"{CB^\lambda}"'] & CFB^\lambda \arrow[d,two heads,"{Q^\lambda}"]\\
FB^\mu \arrow[rr,leftarrow,dashed,shift left=0.7ex,"{C^{\lambda\mu}}"] \arrow[rr,dashed,shift right=0.7ex,"{C^{\mu\lambda}}"'] & & FB^\lambda.
\end{tikzcd}
\end{center}
Put differently, it gives a coordinate-independent characterisation of $G$-bundles. To realise this, we use the canonical sections $\Sigma^\lambda$ and $\Sigma^\mu$ for $Q^\lambda$ and $Q^\mu$ respectively, to construct the functors $C^{\lambda\mu}$ and $C^{\mu\lambda}$ exactly as one would expect. ($\Sigma^\lambda\circ Q^\lambda$ extends the atlas to the maximal compatible one. The resulting coordinate $(G,\lambda)$-fibre bundle is by definition strictly equivalent to the one we started with.)

The proof will rest on the following observation: $CU^\lambda\circ \Sigma^\lambda\circ Q^\lambda\circ CB^\lambda\cong \mathrm{id}_{CC^G}$. Therefore, we first verify this. Let $\mathcal{C}:=(U_\alpha,g_{\alpha\beta})$ be a $G$-cocycle on $M$. Then $CB^\lambda(U_\alpha,g_{\alpha\beta})=(\pi,U_\alpha,\psi_\alpha,g_{\alpha\beta})$ is the corresponding coordinate $(G,\lambda)$-bundle. Now, $\Sigma^\lambda\circ Q^\lambda\circ CB^\lambda(U_\alpha,g_{\alpha\beta})=(\pi,\tilde{U}_\alpha,\tilde \psi_\alpha,\tilde g_{\alpha\beta})$ is the compatible maximal coordinate $(G,\lambda)$-bundle. Therefore, there exist smooth maps $U_\beta\cap \tilde{U}_\alpha\ra{t_{\alpha\beta}}G$, such that $\tilde{\psi}_\alpha\circ\psi_\beta^{-1}(m,s)=(m,t_{\alpha\beta}(m)\cdot s)$, for all $s\in S$. Now, this allows us to define a map of $G$-cocycles $\epsilon_{\mathcal{C}}:=(\mathrm{id}_M, t_{\alpha\beta})$ from $\mathcal{C}=(U_\alpha,g_{\alpha\beta})$ to $\tilde{\mathcal{C}}:=CU^\lambda\circ \Sigma^\lambda\circ Q^\lambda\circ CB^\lambda(\mathcal{C})=(\tilde U_\alpha,\tilde{g}_{\alpha\beta})$. It is easily verified that this is indeed a map of cocycles. Indeed,
\begin{align*}
(m,t_{\alpha\beta}(m)\cdot g_{\beta\gamma}(m)\cdot s)&=\tilde{\psi}_\alpha\circ \psi_\beta^{-1}\circ\psi_\beta\circ\psi_\gamma^{-1}(m,s)\\
&=\tilde \psi_\alpha\circ\psi_\gamma^{-1}(m,s)\\
&=(m,t_{\alpha\gamma}(m)\cdot s)
\end{align*}
and therefore $t_{\alpha\beta}(m)\cdot g_{\beta\gamma}(m)=t_{\alpha\gamma}(m)$, since $\lambda$ is effective. Similarly $\tilde{g}_{\alpha\beta}(m)\cdot t_{\beta\gamma}(m)=t_{\alpha\gamma}(m)$. Moreover, $\epsilon_\mathcal{C}$ is immediately seen to be an isomorphism. Indeed, $(\mathrm{id}_M,t'_{\alpha\beta})$, with $t'_{\alpha\beta}(m):=t_{\beta\alpha}(m)^{-1}$ on $\tilde U_\beta\cap U_\alpha$, is its inverse. 

Moreover, naturality of this isomorphism in $\mathcal{C}$ follows immediately. Indeed, let $\mathcal{C}\ra{(h,h_{\alpha\beta})}\mathcal{C}'$ be a map of cocycles. Then $\epsilon_{\mathcal{C}'}\circ (h,h_{\alpha\beta})=(h,(t'\circ h)_{\alpha\beta})$, while $(h,\tilde h_{\alpha\beta})\circ\epsilon_{\mathcal{C}}=(h,(\tilde h\circ t)_{\alpha\beta})$, where, by $(\tilde h\circ t)_{\alpha\beta}$, I mean the composition by the natural pointwise multiplication (as described in paragraph \ref{cocyc}). What we have to verify, therefore, is that $t'_{\alpha\beta}(h(m))\cdot h_{\beta\gamma}(m)=\tilde h_{\alpha\beta}(m)\cdot t_{\beta\gamma}(m)$, for $m\in U_\gamma\cap \tilde U_\beta\cap h^{-1}(\tilde U '_\alpha)\cap h^{-1}(U'_\beta)$. This is almost immediate, if we write $k$ for the unique map of fibre bundles that corresponds to $(h,h_{\alpha\beta})$:
\begin{align*}(h(m),t'_{\alpha\beta}(h(m))\cdot h_{\beta\gamma}(m)\cdot s)&=\tilde\psi '_\alpha\circ{\psi '}_\beta^{-1}(h(m),h_{\beta\gamma}(m)\cdot s)\\
&=\tilde\psi '_{\alpha}\circ{\psi '}_\beta^{-1}\circ \psi '_\beta\circ k\circ\psi_\gamma^{-1}(m, s)\\
&={\tilde{\psi}'}_{\alpha}\circ k\circ{\tilde{\psi} '}_\beta{}^{-1}\circ {\tilde{\psi} '}_\beta\circ\psi_\gamma^{-1}(m, s)\\
&=\tilde{\psi} '_{\alpha}\circ k\circ{\tilde {\psi} '}_\beta{}^{-1}\circ (h(m), t_{\beta\gamma}(m)\cdot s)\\
&=(h(m),\tilde{h}_{\alpha\beta}(m)\cdot t_{\beta\gamma}(m)\cdot s).
\end{align*}
Effectiveness of the action does the rest. So we conclude that $CU^\lambda\circ \Sigma^\lambda\circ Q^\lambda\circ CB^\lambda\cong \mathrm{id}_{CC^G}$. Of course, if we interchange $\mu$ and $\lambda$, we also find that $CU^\mu\circ \Sigma^\mu\circ Q^\mu\circ CB^\mu\cong \mathrm{id}_{CC^G}$.

We define $C^{\lambda\mu}:=Q^\mu\circ CB^\mu\circ CU^\lambda\circ \Sigma^\lambda$ and $C^{\mu\lambda}$ with $\lambda$ and $\mu$ interchanged. Then 
\begin{align*}C^{\lambda\mu}\circ C^{\mu\lambda}&=(Q^\mu\circ CB^\mu\circ CU^\lambda\circ \Sigma^\lambda)\circ(Q^\lambda\circ CB^\lambda\circ CU^\mu\circ \Sigma^\mu)\\
&=Q^\mu\circ CB^\mu\circ (CU^\lambda\circ \Sigma^\lambda\circ Q^\lambda\circ CB^\lambda)\circ CU^\mu\circ \Sigma^\mu\\
&\cong Q^\mu\circ CB^\mu\circ \mathrm{id}_{CC^G}\circ CU^\mu\circ \Sigma^\mu\txt{(because of the above)}\\
&= Q^\mu\circ \Sigma^\mu\txt{(because of theorem \ref{cfbcon})}\\
&=\mathrm{id}_{FB^\mu}.
\end{align*}
Again, we can interchange the roles of $\mu$ and $\lambda$ to obtain the second part of the equivalence of categories. This proves the theorem.
\end{proof}

\begin{remark} Perhaps the reader was expecting a direct analogue of theorem \ref{cfbcon} in the form of an equivalence of categories between $FB^\lambda$ and a category $C^G$ that is obtained from $CC^G$ by dividing out an equivalence relation on the objects and arrows. Indeed, I initially expected that such an equivalence would be an easy stepping stone to obtain the final result of theorem \ref{cfbcon}, which I really was trying to work towards. My candidate for the category $C^G$ was one of cohomology classes of $G$-cocycles as objects and some corresponding equivalence classes of arrows. This seemed the natural thing to do. Indeed, since my initial attempt I have found a result in \cite{husfib} that establishes that there exists a (not necessarily natural) bijective correspondence between isomorphism classes of $(G,\lambda)$-fibre bundles and cohomology classes of $G$-cocycles. However, my attempt failed because of the difficulty of constructing a canonical section of the quotient functor that sends a cocycle to its cohomology class. To obtain the categorical analogue of the result in Husem\"oller, one would like to have such a section to obtain one of the natural isomorphisms in the suspected equivalence of categories between $FB^\lambda$ and $C^G$. However, it seems plausible that the result can still be proven without such a canonical section. One would then have to pick an arbitrary representative of every cohomology class and show that the final result does not depend on the choice that was made.
\end{remark}

\subsection{Associated Bundles and Generalised Frame Bundles} \label{sec:asbun}
We have seen that, given a $G$-cocycle on $M$, there is one preferred $G$-bundle that we can always construct: the principal $G$-bundle. Therefore, given a $(G,\lambda)$-fibre bundle, we can construct the corresponding principal bundle from it, using theorem \ref{fbcon}. The question arises what the relation is between our original fibre bundle and the principal bundle we constructed from it. 

The answer is that the principal bundle can be given the interpretation of a \emph{bundle of generalised frames}\index{bundle of generalised frames} for the original fibre bundle or, reasoning in the other direction, that the fibre bundle is simply an \emph{associated bundle}\index{associated bundle} to the principal bundle. Both constructions can be understood to be functors. If we write $L$ for the action of $G$ on itself by left multiplication and let $\lambda$ be an effective left $G$ action on $S$, the generalised frame bundle functor $F^\lambda$ is easily seen to be (naturally isomorphic to) $C^{\lambda L}$, while we will see that the associated bundle functor $-[\lambda]$ corresponds to $C^{L\lambda}$. Theorem \ref{fbcon} states that these two functors are each other's pseudoinverse. I will discuss the explicit realisation of both constructions, starting with that of the so-called associated bundles.\\
\\
Suppose we are given a principal bundle $P\ra{\pis}M$ with structure group $G$. Then, for each (not necessarily effective) left action $G\times S\ra{\lambda}S$ on a manifold $S$, we can construct a right action $P\times S\times G\ra{R}P\times S$ by $(p,s,g)\mapsto (p\cdot g,g^{-1} \cdot s)$. This action is obviously free, since the principal action is. Properness can also easily be verified; however, this is somewhat technical, and we will therefore not write it out. Moreover, the action respects the fibre, and hence the quotient manifold theorem tells us that we can construct a so-called associated fibre bundle $P\times S/R=: P[\lambda]\ra{\pi}M$. (Here, $\pi$ is induced by the $G$-invariant map $P\times S \ra{\pis\circ\pi_1}M$.) However, following \cite{micnat}, we can also obtain the smooth structure on $P[\lambda]$ in a different way, which ties in directly with the rest of our discussion.

\begin{claim}\label{assbund} $P[\lambda]\ra{\pi}M$ is a fibre bundle over $M$ with standard fibre $S$, in such a way that for each $p\in P$, $\{p\}\times S \ra{q_p} P[\lambda]_{\pis(p)}$, that is, the restriction of the quotient map $P\times S\ra{q} P[\lambda]$, is a diffeomorphism to the fibre of $P[\lambda]$ above $\pis(p)$ and the equation $q(p,s)=v$ has a unique smooth solution $P\times_M P[\lambda]\ra{\tau}S$.
\end{claim}
\begin{proof} Suppose $(U_\alpha,\phi_\alpha)$ is the maximal principal bundle atlas for $P$. (So $\pis^{-1}(U_\alpha)\ra{\phi_\alpha} U_\alpha\times G$.) From this we will construct a $(G,\lambda)$-fibre bundle atlas for $P[\lambda]$. Let $\pi^{-1}(U_\alpha)\ra{\psi_\alpha} U_\alpha\times S$ be defined by $\psi_\alpha^{-1}(m,s):=q(\phi_\alpha^{-1}(m,e),s)$, where $e$ denotes the identity element in $G$ and $q$ is the quotient map $P\times S\ra{q}P[\lambda]$. (We will show that $\psi_\alpha^{-1}$ is indeed invertible in a minute.) Note that, defined this way, $\psi_\alpha^{-1}$ respects the fibre and is smooth.

For $v\in\pi^{-1}(U_\alpha)$, choose a representative $v=q(p,s)$ with $\pis(p)=m\in U_\alpha$. Since $P_m$ is a free and transitive right $G$-space, there is a unique $g\in G$ such that $p=\phi_\alpha^{-1}(m,e)\cdot g$. Hence $v=q(\phi_\alpha^{-1}(m,e),g\cdot s)$, and this representative is unique. Therefore $\psi_\alpha^{-1}$ is indeed bijective, so our notation was appropriate.

Now, if we agree to write $g_{\alpha\beta}$ for the transition function, i.e. $\phi_\alpha\circ\phi_\beta^{-1}(m,g)=(m,g_{\alpha\beta}(m)\cdot g)$, then 
\begin{align*}\psi_\beta^{-1}(m,s)&=q(\phi_\beta^{-1}(m,e),s)\txt{(*)}\\
&=q(\phi_\alpha^{-1}(m,g_{\alpha\beta}(m)),s)\\
&=q(\phi_\alpha^{-1}(m,e)\cdot g_{\alpha\beta}(m),s)\\
&=q(\phi_\alpha^{-1}(m,e),g_{\alpha\beta}(m)\cdot s)\\
&=\psi_\alpha^{-1}(m,g_{\alpha\beta}(m)\cdot s).
\end{align*}
We conclude that $\psi_\alpha\circ\psi_\beta^{-1}(m,s)=(m,g_{\alpha\beta}(m)\cdot s)$. Since we know $g_{\alpha\beta}$ to satisfy the cocycle condition ($P$ is a fibre bundle) we conclude that, by the fibre bundle construction theorem, $P[\lambda]$ has a unique structure as a fibre bundle with the $G$-atlas $(U_\alpha,\psi_\alpha)$. 
Note that $(*)$ implies that $q$ is a smooth submersion with respect to the smooth structure that $P[\lambda]$ obtains from this fibre bundle atlas. Therefore this smooth structure coincides with the one we would obtain from an application of the quotient manifold theorem.

For the last assertion, we observe that by definition of $\psi_\alpha$
\begin{center}
\begin{tikzcd}[column sep=large,row sep=large]
P|_{U_\alpha}\times S \arrow[r,"{\phi_\alpha \times \mathrm{id}_{S}}"] \arrow[d,"q"'] \arrow[dr,phantom,"\com",description] & U_\alpha\times G\times S \arrow[d,"{\mathrm{id}_{U_\alpha} \times \mathrm{ev}}"]\\
P[\lambda]|_{U_\alpha} \arrow[r,"{\psi_\alpha}"'] & U_\alpha\times S,
\end{tikzcd}
\end{center}
where $G\times S \ra{\mathrm{ev}} S$ is given by $(g,s)\mapsto g(s)$.
Moreover, the horizontals in the diagram are diffeomorphisms, so we see that $q_p$ is indeed a diffeomorphism.

Now, consider the equation $q(p,s)=v$. Note that $\det D_s q(p,s)\neq 0$, since $q_p$ is a diffeomorphism for each $p$. (By $D_s$ we mean the square matrix of partial derivatives with respect to coordinates corresponding to $s$.)  The implicit function theorem (in local coordinates) then tells us that a unique smooth solution $P\times_M P[\lambda]\ra{\tau}S$ exists such that $q(p,\tau(p,v))=v$.
\end{proof}
Again, we would like to extend this definition on objects to a functor $PB^G\ra{-[\lambda]} FB^\lambda$. Of course, there is only one natural way to do this. We send a morphism $P\ra{k}P'$ of principal bundles to the map $k[\lambda]=(k\times \mathrm{id}_S)/G$, i.e. $(k[\lambda])([p,s])=[k(p),s]$\footnote{Here $[p,s]$ denotes the equivalence class of $(p,s)$.}. This map is well-defined. Indeed, $(k[\lambda])([(p,s)\cdot g])=[k(p\cdot g),g^{-1}\cdot s]=[k(p)\cdot g,g^{-1}\cdot s]=[k(p),s]=(k[\lambda])([p,s])$, where the second identity holds precisely because $k$ is a map of principal bundles, i.e. $G$-equivariant. Functoriality is clear, i.e. $(k' \circ k)[\lambda]=k'[\lambda]\circ k[\lambda]$.

Note that, by definition of $k[\lambda]$, the following square commutes
\begin{center}
\begin{tikzcd}[column sep=large,row sep=large]
P\times S \arrow[r,"{k\times \mathrm{id}_S}"] \arrow[d,"q"'] & P'\times S \arrow[d,"{q'}"]\\
P[\lambda] \arrow[r,dashed,"{k[\lambda]}"'] & P'[\lambda].
\end{tikzcd}
\end{center}
Therefore $k[\lambda]\circ q$ is a smooth map, since $q' \circ (k \times \mathrm{id}_S)$ is so. Since $q$ is a surjective submersion, we conclude that $k[\lambda]$ is smooth.

To conclude our discussion of the associated bundle, we check that $k[\lambda]$ is a morphism of $(G,\lambda)$-bundles. Denote by $\pis$ and $\pis'$, respectively, the projections from $P[\lambda]$ and $P'[\lambda]$ to $M$ and $M'$. Note that $k\times \mathrm{id}_S$ is a map of fibre bundles:
\begin{center}
\begin{tikzcd}[column sep=large,row sep=large]
P\times S \arrow[r,"{k\times \mathrm{id}_S}"] \arrow[d,"{\pis \circ q}"'] \arrow[dr,phantom,"\com",description] & P'\times S \arrow[d,"{\pis'\circ q'}"]\\
M \arrow[r,"{k_M}"'] & M'.
\end{tikzcd}
\end{center}
Therefore $k_M\circ \pis \circ q=\pis'\circ q'\circ (k\times \mathrm{id}_S)=\pis' \circ k[\lambda] \circ q$. Now, since $q$ is a surjection, we conclude that $k_M\circ \pis=\pis'\circ k[\lambda]$, so $k[\lambda]$ is a map of fibre bundles. 

The last thing we want to do is to verify that this map acts by $\lambda$. Let $p\in P_m$ and choose $p'\in P'_{k_M(m)}$. For this paragraph, write $\bar q_p:S\to P[\lambda]_{m}$ and $\bar q'_{p'}:S\to P'[\lambda]_{k_M(m)}$ for the maps $s\mapsto [p,s]$ and $s\mapsto [p',s]$. Since $G$ acts freely and transitively on the fibres of $P'$, there is a unique $g\in G$ such that $k(p)=p'\cdot g$. Then
$$k[\lambda](\bar q_p(s))=[k(p),s]=[p'\cdot g,s]=[p',g\cdot s]=\bar q'_{p'}(\lambda_g(s)).$$
Thus, in the fibre coordinates determined by $p$ and $p'$, $k[\lambda]$ is given by the action map $\lambda_g:S\to S$.
\begin{corollary}In the case of an effective action $\lambda$, the associated bundle construction can equivalently be described as
$-[\lambda]\cong C^{L\lambda}$.
\end{corollary}
\begin{proof}The statement that $P[\lambda]\cong C^{L\lambda}(P)$ follows immediately by comparison of the conclusion that $\psi_\alpha\circ\psi_\beta^{-1}(m,s)=(m,g_{\alpha\beta}(m)\cdot s)$ in the proof of claim \ref{assbund} with the construction in the proof of theorem \ref{cfbcon}. The naturality of the isomorphism should be clear from the discussion above.
\end{proof}
Therefore, there is also a general (pseudo)inverse construction to that of the associated bundle corresponding to an effective action. Indeed, suppose we are given a fibre bundle $F\ra{\pi}M$, with fibre $S$ and structure group $G$ that acts by an action $\lambda$. Then, as a consequence of theorem \ref{fbcon}, $C^{\lambda L}$ is a pseudoinverse to $C^{L \lambda}$ and therefore also to $-[\lambda]$.

We can also find an interpretation of this inverse construction. Note that we can consider the elements of $C^{\lambda L}(F)$ as generalised frames\footnote{That is, a (linear) frame being a linear isomorphism (and thus a diffeomorphism) $\Bbbk^k\ra{}V_m$, if $V$ is a vector bundle and either $\Bbbk=\R$ or $\Bbbk=\C$.} for $F$. We identify a $p\in C^{\lambda L}(F)$ with $S\cong \{p\}\times S\ra{q_p} ( C^{\lambda L}(F))[\lambda]_{\pis(p)}\cong F_{\pis(p)}$, which is a diffeomorphism according to claim \ref{assbund}. This enables us to see $ C^{\lambda L}(F)_m$ as a subset $F^\lambda(F)_m\subset \mathrm{Diff}(S,F_m)$. Note that $F^\lambda(F)$ contains precisely the frames corresponding to the fibre components of arrows between $S\times U$ and $F|_U$ in $FB_U^\lambda$, if $U$ is some open in $M$. In particular, sections of $F^\lambda(F)$ correspond to trivialisations of $F$. Also note that the action of $C^{\lambda L}$ on arrows corresponds in a natural way to postcomposition in the identification $C^{\lambda L}(F)\cong F^\lambda(F)$, i.e.
$$\label{eq:asfrid}q(C^{\lambda L}(k)(p),-)=q_{C^{\lambda L} (k)(p)}=k\circ q_p=k\circ q(p,-),\txt{(*)}$$
if $F\ra{k}F'$ is a map of $(G,\lambda)$-fibre bundles. We have obtained the following result.
\begin{corollary}
We have a natural isomorphism $C^{\lambda L}\cong F^\lambda$, where $F^\lambda$ is defined as described above.
\end{corollary}

Summarising, we have made the equivalence of categories
\begin{center}
\begin{tikzcd}[column sep=huge]
PB^G \arrow[r,leftarrow,shift left=0.7ex,"{C^{\lambda L}}"] \arrow[r,shift right=0.7ex,"{C^{L\lambda}}"'] & FB^\lambda
\end{tikzcd}
\end{center}
(which we had established for an effective action $G\times S\ra{\lambda}S$) concrete in the form of the associated bundle construction $-[\lambda]\cong C^{L\lambda}$ and that of the generalised frame bundle $F^\lambda\cong C^{\lambda L}$. Naturally, these constructions are also each other's pseudoinverse:
\begin{center}
\begin{tikzcd}[column sep=huge]
PB^G \arrow[r,leftarrow,shift left=0.7ex,"{F^\lambda}"] \arrow[r,shift right=0.7ex,"{-[\lambda]}"'] & FB^\lambda.
\end{tikzcd}
\end{center}
Therefore, whenever we encounter a $(G,\lambda)$-fibre bundle over a manifold $M$, we can interpret it as an associated bundle of some principal $G$-bundle over $M$, which then takes on the role of a bundle of generalised frames for the original fibre bundle. This interpretation is of importance when trying to understand why principal bundles arise in physics.

\begin{remark}The reader might wonder why we have to demand that $\lambda$ be effective to obtain the result that the associated bundle functor defines an equivalence of categories between $PB^G$ and $FB^\lambda$. Indeed, we have seen that the associated bundle functor can be defined for arbitrary $G$-actions. To see what can go wrong, consider the following (somewhat extreme) example. Let $G\times S \ra{\lambda}S$ be the trivial action of a non-trivial Lie group $G$ on a one-point manifold $S=\{p\}$. Then we can construct a category $FB^\lambda\cong \mathrm{Mfd}$, which is, of course, just isomorphic to the category of smooth manifolds. Moreover, we have an associated bundle functor $PB^G\ra{-[\lambda]}FB^\lambda\cong \mathrm{Mfd}$, which just sends each principal bundle to its base space. Obviously, this functor is not faithful\footnote{As is easily checked, it is full and essentially surjective on objects, though.}, so it cannot possibly define an equivalence of categories. 
\end{remark}
\subsection{Vector Bundles and Principal $\GL(\Bbbk^k)$-Bundles}\label{vbequivalence}
An important special case of this equivalence of categories is that of the standard action $\lambda$ of $\GL(\Bbbk^k)$ on the vector space $\Bbbk^k$, where either $\Bbbk=\R$ or $\Bbbk=\C$. Note that the objects of $FB^\lambda$ are then just the $k$-vector bundles over the field $\Bbbk$, while the arrows are morphisms of vector bundles (morphisms of fibre bundles that are linear on the fibre) that are invertible on each fibre. We therefore also denote this category by $VB^k_\Bbbk$.

Let $V\ra{\pi}M$ be a $k$-vector bundle. The corresponding $\GL(\Bbbk^k)$-principal bundle of linear frames $F^\lambda(V)\ra{\pis}M$ is often simply called \emph{the frame bundle}. This can be regarded as the open subbundle of the $\mathrm{Hom}$-bundle $\mathrm{Hom}(\Bbbk^k\times M,V)$ defined by the condition that the determinant is non-zero. Observe that the existence of a canonical basis for $\Bbbk^k$ tells us that a (local) frame for a $k$-vector bundle is the same as a set of $k$ linearly independent (local) sections of the bundle. The inverse construction $PB^{\GL(\Bbbk^k)}\ra{-[\lambda]}VB^k_\Bbbk$ is sometimes called \emph{the associated vector bundle construction}.

\begin{remark}
One could argue that $VB_\Bbbk^k$ is not the right category to study when discussing vector bundles, since we would also like to consider maps that are not invertible on the fibre. One could advocate that the statement of the equivalence of categories between $VB_\Bbbk^k$ and $PB^{\GL(\Bbbk^k)}$\, could be restricted to the statement that we have a bijective correspondence on objects. However, this equivalence of categories tells us more. It tells us that this bijective correspondence is natural with respect to a set of morphisms of vector bundles (and principal bundles), which includes the isomorphisms. Note that in the case in which we only consider the corresponding categories of bundles over a manifold $M$ with morphisms that are the identity (or some automorphism) on $M$, this set of arrows consists precisely of the isomorphisms. 
\end{remark}

\subsection{Orthogonal Frame Bundles}\label{sec:ortheq}
We can also start with a vector bundle equipped with a metric, i.e. a smooth fibrewise symmetric bilinear form (or Hermitian form in the complex case) $g$ that is non-degenerate in the sense that it identifies each fibre with its dual (antilinearly in the Hermitian case). This is a situation that we encounter all the time in physics. $g$ can either be a Euclidean, Hermitian, or Lorentzian metric. The examples are numerous. In these physical examples, we are usually not so much interested in the category of all vector bundles, but in the category $MVB_{\Bbbk,(p,q)}^k$ of metrised $k$-dimensional vector bundles over the field $\Bbbk$ with signature $(p,q)$. Of course, its objects are just $k$-vector bundles together with a specified $(p,q)$-metric while its arrows are vector bundle morphisms which are also isometries. (Note that isometries between $k$-vector spaces are automatically invertible, so these morphisms are indeed fibrewise invertible.) We can equivalently describe these metrised vector bundles in the following way.
\begin{claim}\label{orth}
Let $O(\Bbbk,(p,q))\times\Bbbk^k\ra{\lambda}\Bbbk^k$ be the standard action of the orthogonal group corresponding to the standard $(p,q)$-inner product $\ip{\cdot}{\cdot}$ on $\Bbbk^k$. Then we have an isomorphism of categories
\begin{center}
\begin{tikzcd}[column sep=huge]
MVB_{\Bbbk,(p,q)}^k \arrow[r,leftarrow,shift left=0.7ex,"{M^k}"] \arrow[r,shift right=0.7ex,"{D^k}"'] & FB^\lambda.
\end{tikzcd}
\end{center}
\end{claim}
\begin{proof}We first construct a metrisation functor $M^k$. Let $\F:=(\pi,[U_\alpha,\psi_\alpha,g_{\alpha\beta}])\in FB^\lambda$ be an $(O(\Bbbk,(p,q)),\lambda)$-fibre bundle (which is of course a vector bundle). Pick the corresponding maximal atlas structure, i.e. $\Sigma^\lambda(\F)=(\pi,(U_\alpha,\psi_\alpha,g_{\alpha\beta}))$. Then each $\psi_\alpha$ gives us a (pseudo) inner product $g_\alpha$ on $\pi^{-1}(U_\alpha)$ (from the standard inner product on $\Bbbk^k$). Moreover, these inner products agree on chart overlaps, since the transition functions act by the standard action of the orthogonal group. The sheaf property of smooth functions gives us a unique inner product $g$ on $\pi$. We write $M^k(\F):=(\F,g)$ for the metrised vector bundle we have constructed. Note that arrows in $FB^\lambda$ are indeed isometries. Indeed, in charts, they are described by an action of the orthogonal group, so we can define $M^k$ to be the identity on maps.

The demetrisation functor $D^k$ will just be the functor that forgets the metric and is the identity on arrows. However, we have to show that a metrised vector bundle has an atlas with transition functions that only take values in $O(\Bbbk,(p,q))\subset \GL(\Bbbk^k)$. This is obvious when we see that we can construct local orthonormal frames for the bundle.

Indeed, let $(V\ra{\pi}M,g)$ be a metrised vector bundle. Let $(U_\alpha,\psi_\alpha,g_{\alpha\beta})$ be its maximal vector bundle atlas (so $(g_{\alpha\beta})$ have values in $\GL(\Bbbk^k)$). Let $(f_{\alpha,j})$, $U_\alpha\ra{f_{\alpha,j}}V|_{U_\alpha}$ be a local frame for $V\ra{\pi}M$. After refining the cover if necessary, the pointwise application of the Gram-Schmidt procedure gives us an orthonormal frame $(f'_{\alpha,j})$, i.e. a frame such that $g(f'_{\alpha,j}(m),f'_{\alpha,k}(m))=0$ if $j\neq k$ and $g(f'_{\alpha,j}(m),f'_{\alpha,j}(m))=\pm 1$, where the $-$-sign occurs when $j\leq p$ and the $+$-sign holds for the other $q$ indices. Of course, these orthonormal frames define a strictly equivalent atlas $(U_\alpha,\psi'_\alpha,g'_{\alpha\beta})$, by $\psi'^{-1}_\alpha(m,(v^1,\ldots,v ^k)):=\sum_{j=1}^k v^j f_{\alpha,j}'(m)$. Then, clearly, $\ip{v}{w}=g(\psi'_\alpha {}^{-1}(m,v),\psi'_\alpha {}^{-1}(m,w))$, where $\ip{\cdot}{\cdot}$ is the $(p,q)$-inner product on $\Bbbk^k$. It follows that
\begin{align*}\ip{v}{w}&=g(\psi'_\beta {}^{-1}(m,v),\psi'_\beta {}^{-1} (m,w))\\
&=g(\psi'_\alpha {}^{-1}(m,g'_{\alpha\beta}(m)\cdot v),\psi'_\alpha {}^{-1} (m,g'_{\alpha\beta}(m)\cdot w))\\
&=\ip{g'_{\alpha\beta}(m)\cdot v}{g'_{\alpha\beta}(m)\cdot w}.
\end{align*}
We conclude that the image of $g'_{\alpha\beta}$ lies in $O(\Bbbk,(p,q))$. This proves the claim, since $M^k$ and $D^k$ are clearly each other's inverse.
\end{proof}
\begin{remark}
Since we can construct a Riemannian metric on each vector bundle by using partitions of unity, we see that we can always specify a $k$-dimensional vector bundle over $\Bbbk$ by transition functions with values in $O(\Bbbk,(k,0))$. Put differently, by choosing a metric, we can interpret it as a $(O(\Bbbk,(k,0)),\lambda)$-fibre bundle. Note that this interpretation is highly non-canonical, depending on the choice of a metric.
\end{remark}
\begin{corollary}\label{metrequi}Now, theorem \ref{fbcon} tells us that we obtain an equivalence of categories $$MVB_{\Bbbk,(p,q)}^k\mathop{\leftrightarrows}^{}_{}PB^{O(\Bbbk,(p,q))}.$$ Similarly, we obtain an equivalence of categories
$$OMVB_{\Bbbk,(p,q)}^k\mathop{\leftrightarrows}^{}_{}PB^{SO(\Bbbk,(p,q))}$$  between the categories $OMVB_{\Bbbk,(p,q)}^k$ of oriented $(p,q)$-metrised $k$-vector bundles and that of $SO(\Bbbk,(p,q))$-principal bundles\footnote{Here $SO(\Bbbk,(p,q))$ is the subgroup of orientation-preserving elements of $O(\Bbbk,(p,q))$ (those with determinant $+1$). Of course, in this somewhat clumsy notation, $U(n)=O(\C\, ,(n,0))$ and $O(p,q)=O(\R\, , (p,q))$ and similarly for their orientation-preserving counterparts.}  
\end{corollary}
An important advantage of using an orthonormal frame bundle instead of a general frame bundle when describing physics is that its structure group $G$ is in many cases (in the case of a positive definite metric) compact. This means, for example, that, using the Haar measure, we can integrate over it to average quantities, making them invariant under the principal right action. Since the principal right action represents a change of basis, or gauge transformation, in the original vector bundle, this can be an important demand for quantities that are of physical interest. Moreover, in a very similar way, it will allow us to define invariant inner products on vector spaces on which we have a linear $G$-action (in particular an $\mathrm{Ad}$-invariant inner product on $\g$), which will be crucial ingredients in Yang-Mills theories.

This turns out to be a very useful way of looking at these bundles. Indeed, modern classical field theory is not formulated in the language of metrised vector bundles but in that of principal bundles with orthogonal structure groups. For example, the electromagnetic force is dealt with on a principal $U(1)$-bundle, the weak force on a principal $SU(2)$-bundle and the strong force on a principal $SU(3)$-bundle. \cite{frageo} Moreover, the spacetime metric in general relativity gives rise to an $O(1,3)$-principal bundle that has the Lorentz frames as elements. Using this principal bundle, we can hope to define spinors on a curved spacetime, a construction that would be incomprehensible without the language of principal bundles.

\begin{remark}
Of course, we never used the symmetry of the inner product in the proof of claim \ref{orth}. A similar result therefore also holds, for example, for symplectic forms: the category of $2k$-dimensional symplectic vector bundles and the category of principal bundles with $Sp(2k,\Bbbk)$ as a structure group are equivalent. One encounters an example of this fact when one constructs the (two-dimensional) spinor bundle corresponding to a spin structure on a manifold. Since this is an associated bundle to a $SL(2,\C)$-principal bundle and $SL(2,\C)= Sp(2,\C)$, we naturally obtain a symplectic form on the spinor bundle, from the canonical symplectic form on $\C^2$. This symplectic form is used intensively by physicists, since it is closely related to the spacetime metric in general relativity, when we consider spin structures on a spacetime manifold. \end{remark}

\subsection{How do sections transform under the equivalence?}
In field theories, e.g. classical electromagnetism, physical quantities are described by sections of certain vector bundles over spacetime. While doing physics, we like to do calculations in concrete situations to derive concrete results. To do this, it is very convenient to drop the abstract notation we have been using, to pick a basis for each fibre and do the calculations in coordinates corresponding to this basis. Of course, we want this basis to vary smoothly with the base point. So, briefly, we pick a local frame for our vector bundle when we want to do computations.

This is where the frame bundle comes in. A local frame is, of course, just a local section of the frame bundle. So we have a natural idea of how sections of principal $\GL(\Bbbk^k)$-bundles should transform under the equivalence of categories: they correspond to frames of the canonical associated bundle.

The converse question also naturally arises. The quantities of physical interest are sections of a vector bundle. If we want to describe the physics in terms of the corresponding frame bundle, we would like to know what the counterparts of these sections are in this context. The answer (in a slightly more general context) turns out to be as follows.

\begin{claim}\label{seccoo} Let $(P,\pi,M,G)$ be a principal fibre bundle and let $G\times S\ra{\lambda}S$ be an effective action. Denote by $ C^\infty(P,S)^G$ the space of $G$-equivariant mappings in the sense that $\Phi(p\cdot g)=g^{-1}\cdot \Phi(p)$ for all $g\in G$, $p\in P$ and $\Phi\in  C^\infty(P,S)^G$. Then, sections $\sigma \in \Gamma(P[\lambda])$ of the associated bundle are in one-to-one correspondence with $G$-equivariant mappings $\Phi\in  C^\infty(P,S)^G$.
\end{claim}
\begin{proof} Let $\Phi\in C^\infty(P,S)^G$. Then we define a section $\sigma_\Phi$ of the associated bundle as follows. Obviously $P\ra{\mathrm{id}_P\times \Phi} P\times S$ is $G$-equivariant. Taking the quotient by the $G$-action we obtain a smooth section $\sigma_\Phi=(\mathrm{id}_P\times \Phi)/G: M\cong P/G\ra{}(P\times S)/G=P[\lambda]$. More explicitly, $\sigma_\Phi(m):= [p,\Phi(p)]$, where $p\in \pis^{-1}(m)$. This expression is well-defined, for if $\pis(p)=\pis(q) \in M$, then there exists a $g\in G$ such that $q=p\cdot g$. Then $[q,\Phi(q)]=[p\cdot g, \Phi(p\cdot g)]=[p\cdot g, g^{-1}\Phi(p)]=[p,\Phi(p)]$, where the second equality holds since $\Phi$ is $G$-equivariant. Moreover, we see that $\sigma_\Phi$ is indeed a section of $P[\lambda]$. Further, $\sigma_\Phi\circ \pis=q\circ (\mathrm{id}_P\times \Phi)$ and thus smooth, where $P\times S\ra{q} P[\lambda]$ is the quotient map. Since $\pis$ is a surjective submersion, $\sigma_\Phi$ is smooth as well.

Conversely, given $\sigma\in \Gamma(P[\lambda])$, we define $\Phi_\sigma\in  C^\infty(P,S)^G$ by $\Phi_\sigma(p):=q_p^{-1}(\sigma(\pis(p)))$, $q_p$ is the restriction of $q$, as in claim \ref{assbund}. It is easy to see that this gives a $G$-equivariant map. Indeed, $q_p(s)=[p,s]=[p\cdot g,g^{-1}\cdot s]=q_{p\cdot g}(g^{-1}\cdot s)$ and therefore $\Phi_\sigma(p\cdot g)=q_{p\cdot g}^{-1}(\sigma(m))=q_{p\cdot g}^{-1}(q_p(\Phi_\sigma(p)))=g^{-1}\cdot \Phi_\sigma(p)$. Smoothness of $\Phi_\sigma$ is also apparent, if we note that $\Phi_\sigma(p)=q_p^{-1}(\sigma(\pis(p)))=\tau(p,\sigma(\pis(p)))$, where $\tau$ is the smooth map defined in claim \ref{assbund}.

Finally, these constructions are each other's inverse, since
$$\Phi_{\sigma_\Phi}(p)=q_p^{-1}(\sigma_\Phi(\pis(p)))=q_p^{-1}([p,\Phi(p)])=\Phi(p)\txt{and}$$
$$\sigma_{\Phi_\sigma}(m)=[p,\Phi_\sigma(p)]=[p,q_p^{-1}(\sigma(m))]=q_p(q_p^{-1}(\sigma(m)))=\sigma(m).$$
\end{proof}

This is of course what physicists do all the time: they choose a local frame $e\in \Gamma(F^\lambda(V)|_U)$ for a vector bundle $V\ra{}M$ and identify a section $\sigma\in \Gamma(V)$ with its coordinates $M\supset U\ra{\Phi_\sigma\circ e}\Bbbk^k$ with respect to the frame $e$. Claim \ref{seccoo} tells us that this identification is valid.

This approach has the advantage that it is very concrete: a section of $V$ becomes just a function from $M$ that has some numbers as output, which enables us to do computations. However, we must always remember that the quantity $\Phi_\sigma\circ e$ is not an intrinsic characterisation of $\sigma$ ($\Phi_\sigma$ is though!), since it depends on the frame we choose to describe $\sigma$.

\begin{remark}
In the theory of relativity, this principle has some remarkable manifestations. For example, the electromagnetic field is represented mathematically by a $2$-form $M\ra{F}\bigwedge^2 T^*M$ on the spacetime manifold $M$. An observer has an associated coordinate frame $e\in F^\lambda(V)|_U$, with $e_{\mu\nu}=d x^\mu\wedge d x^\nu$, if $(x^\mu)$ are geodesic normal coordinates in which the observer is at rest. With respect to this frame, we can describe $F$ in coordinates $M\ra{\Phi_F\circ e=(F_{\mu \nu})}\bigwedge^2(\R^4)^*$. However, another observer will have their own coordinate frame $e'$ and with respect to this frame $F$ will have different coordinates $(F_{\mu\nu}')$. Components of $F$ that one observer would describe as magnetic fields will be electric fields for the other and vice versa.
\end{remark}

\subsection{Categories of Principal Bundles Revisited}
Principal bundles have turned out to play a fundamental role in modern physics. It is therefore important to stop and consider their properties.

There is some ambiguity when dealing with principal bundles from a categorical point of view. Of course, everyone would agree on what we would like to call a principal bundle. However, what maps should we consider? We have two choices to make:
\begin{enumerate}
\item Do we allow different base spaces $M$ (in which case it seems natural to allow the maps to be non-trivial on the base)?
\item Do we allow different structure groups $G$ (in which case it seems natural to allow the maps to be non-trivial on the group)?
\end{enumerate}
This gives us our four categories of principal bundles, which turn out to be very useful: of course there is the biggest one, $PB$, where we allow everything to happen, then there is $PB^G$ (the one of \cite{husfib}), where the answers to the questions are, respectively, yes and no, $PB_M$ (the one of \cite{kobfou}), where the answers are the other way around, and finally $PB_M^G$, where we answer no to both questions. The first category is essential when we want to emphasise the role of a principal bundle as a special case of a $(G,\lambda)$-fibre bundle. On the other hand, the second two are more workable when we want to do physics. In particular, we have the following results that will be important for gauge theory.
\begin{claim}
$PB_M$ has finite products.  To be precise, if $P_i\ra{\pis_i}M$ is a principal $G_i$-bundle, then the fibred product $P_1\times_MP_2\ra{\pis}M$ is a principal $G_1\times G_2$-bundle. This is the categorical product.
\end{claim}
\begin{proof}
We take the projection from $\pis$ to $\pis_i$ to be the map of principal bundles that is the obvious projection onto the $i$-th component as a map of fibre bundles and the projection $G_1\times G_2\ra{}G_i$ on the structure group. The universal property of the product is easily checked.
\end{proof}
It should be noted that physicists usually call $\pis$ the \emph{spliced principal bundle}\index{spliced principal bundle} of $\pis_1$ and $\pis_2$. \footnote{Although the importance of this construction might not be evident from this thesis, it is one of the more important tools in gauge theory. When we deal with particles that interact with multiple gauge fields or with particles with spin, the construction is essential to obtain the correct principal bundle on which the physics is described.}

The following is a definition that is of the most fundamental importance in gauge theory.
\begin{definition*}Let $P\ra{\pis}M\in PB^G_M$ be a principal $G$-bundle. We call $\mathrm{Aut}(\pis)$, the group of automorphisms of $\pis$ in $PB^G_M$, the group of \emph{gauge transformations}\index{gauge transformation} of $\pis$ and denote it by $\mathrm{Gau}(\pis)$. To be explicit, this group consists of diffeomorphisms $P\ra{k}P$, such that $\pis\circ k=\pis$ and such that $\rho_g\circ k=k\circ \rho_g$, for all $g\in G$. Note that any arrow in $PB^G_M$ is an isomorphism, so $\mathrm{Gau}\, \pis$ consists precisely of the endomorphisms of $\pis$ in $PB^G_M$.
\end{definition*}
Just like $\mathrm{Diff}\, S$, $\mathrm{Gau}\, \pis$ has the structure of an infinite-dimensional manifold on which the group operations are smooth. \cite{micman} In the line of thought of the previous paragraph, we have the following result.
\begin{claim}\label{isomgauge}
There is an isomorphism of groups $C^\infty(P,G)^G\cong \mathrm{Gau}\, \pis$, where we take $G$ to act on itself on the left by conjugation. Claim \ref{seccoo} now tells us that we can see gauge transformations as sections of the obvious associated $(G,C)$-fibre bundle.
\end{claim}
\begin{proof}Let $\Phi\in C^\infty(P,G)^G$. Then we define a map $P\ra{k_\Phi}P$ by $k_\Phi(p):=p\cdot \Phi(p)$. Obviously, this is a map of fibre bundles that is the identity on $M$. Moreover $k_\Phi(p\cdot g)=(p\cdot g)\cdot \Phi(p\cdot g)=(p\cdot g)\cdot g^{-1}\cdot \Phi(p)\cdot g=p\cdot \Phi(p)\cdot g=k_\Phi(p)\cdot g$, so $k_\Phi$ is a map of principal bundles that is the identity on $G$. We conclude that $k_\Phi\in \mathrm{Gau}\, \pis$.

Conversely, let $k\in \mathrm{Gau}\, \pis$ and define $P\ra{\Phi_k}G$ by the implicit relation $k(p)=p\cdot \Phi_k(p)$. (Note that this is a good definition since the principal right action is transitive and free on the fibre. Moreover, smoothness is immediate from the smoothness of the group multiplication, if we write out the definition in a chart.) Now, $p\cdot g\cdot \Phi_k(p\cdot g)=k(p\cdot g)=k(p)\cdot g=p\cdot \Phi_k(p)\cdot g$. Therefore $\Phi_k(p\cdot g)=g^{-1}\cdot \Phi_k(p)\cdot g$. We conclude that $\Phi_k\in C^\infty(P,G)^G$.

Now, obviously, these constructions are each other's inverse. Finally, let $\Phi,\Phi'\in C^\infty(P,G)^G$. Then $k_\Phi\circ k_{\Phi'}(p)=p\cdot \Phi'(p)\cdot\Phi(p\cdot \Phi'(p))=p\cdot\Phi'(p)\cdot\Phi'(p)^{-1}\cdot\Phi(p)\cdot \Phi'(p)=p\cdot \Phi(p)\cdot \Phi'(p)=k_{\Phi\cdot \Phi'}(p)$.
\end{proof}

\begin{remark}
These gauge transformations are a generalisation to arbitrary spacetimes of what physicists call gauge transformations. Indeed, every principal bundle over Minkowski spacetime $M=\R^4$ is trivial, since $\R^4$ is contractible. \cite{stetop} Therefore, in the case of a principal $G$-bundle $\pis$ over $M=\R^4$ (which we know to be trivial), we can construct a special class of gauge transformations from maps $M\ra{\phi}G$, by defining $f_\phi\in\mathrm{Gau}\, \pis$ by $f_\phi(m,g)=(m,\phi(m)\cdot g)$. Indeed, for a trivial principal bundle, we easily see that every gauge transformation is of this form. These transformations are what physicists call a local gauge transformation, in contrast to the case in which they demand that $\phi$ be constant, in which case they talk about global gauge transformations. It is important to see, however, that in general, such a difference cannot be made and we can only talk about plain gauge transformations.\end{remark}

\section{Connections on Fibre Bundles}
In this section, I introduce the notion of a connection on a fibre bundle, which is central both in the general theory of relativity, where it determines the geodesics along which a point particle moves through spacetime (and in some sense is a potential for the gravitational field), and in elementary particle physics, where it describes the gauge potentials, the mathematical abstractions of the fundamental interactions of matter.

Many different notions of a connection are used in the literature. I will use a particularly general geometric definition, in accordance with \cite{micnat}, and will treat a few special cases separately. I hope that this abstract approach will help the reader to see the connections between the mathematics used in the different theories of physics, which are insufficiently emphasised in many references.

\subsection{A Geometric Definition}
\label{geocon}
Let $F\ra{\pi}M$ be a fibre bundle. Note that at each point $f\in F$ the kernel $\ker T_f \pi$ is a linear subspace of $T_f F$. This is easily seen to define a subbundle $\ker T\pi \subset TF$ over $F$, called the \emph{vertical bundle}\index{vertical bundle} and denoted $VF$.

This bundle can also be given another interpretation. To see this, note that
$$0\ra{}T_f F_m\ra{T_f i_m}T_f F\ra{T_f \pi} T_m M\ra{}0$$
is a short exact sequence. Here $m=\pi(f)$ and $F_m \ra{i_m} F$ denotes the inclusion. (Indeed $\pi\circ i_m$ has constant value m, so $\mathrm{im}\; T_f i_m \subset VF_f$. Moreover their dimensions are equal.) \footnote{Furthermore, if $F$ is a vector bundle, we can identify $T_f F_m\cong F_m$, so, in that case, we can see the vertical subspace $VF_f$ at a point $f$ as a copy of $F_m$.} The short exact sequence above also tells us that we can identify the quotient space $T_fF/VF_f$ with $T_mM$. This gives us the following identity of vector bundles over $F$:
$$TF/VF\cong \pi^* TM,$$
where $\pi^*TM$ denotes the pullback bundle of $TM$ along $\pi$. Of course, this calls for an explicit realisation of $TF$ as a direct sum of $VF$ and $\pi^* TM$. This is given by a connection.
\begin{definition*}[Connection]
Let $F\ra{\pi}M$ be a fibre bundle. Let $\pi^* TM\ra{\iota} TF$ be an embedding of vector bundles over $F$ such that $T\pi\circ\iota(f,X)=X$, and denote its image by $HF$. We call $HF$ a \emph{connection}\index{connection} if \; $TF=HF\oplus VF$ and $HF_f$ is called the \emph{horizontal subspace}\index{horizontal bundle} at $f\in F$. The projections $HF\la{}TF\ra{}VF$ are denoted by $P_H$ and $P_V$ respectively. \cite{banrie}
\end{definition*}
\begin{remark}Although I will interchangeably use the vertical projection $P_V$ and the horizontal bundle $HF$ to denote a connection, what I mean by a connection is the whole product diagram $HF\la{P_H}TF\ra{P_V}VF$  in the category $FB_F$ of fibre bundles over $F$.
\end{remark}
\subsection{Curvature}
\label{sec:curv}
Note that we can see $HF$ and $VF$ as distributions on $F$. We can of course ask whether these distributions are integrable. Recall Frobenius' theorem, which gives a useful criterion for integrability \cite{wesdif}:

\begin{theorem}[Frobenius]\index{Frobenius' theorem}\label{thm:frob} A $l$-dimensional distribution $H$ on a manifold $F$ is integrable, i.e. there exists a $l$-dimensional foliation $\mathcal{F}$ such that for all $f\in F:$ $H_f$ is the tangent space $T_f\mathcal{F}_f$ to the leaf $\mathcal{F}_f$ through $f$, if and only if it is involutive, i.e. for each pair of vector fields $X,Y\in \Gamma(H)\subset \mathcal{X}(F)$, $[X,Y]\in \Gamma(H)$.
\end{theorem}
Obviously, $VF$ is integrable. 

\begin{claim}\label{cl:vertb}
$VF$ is integrable.
\end{claim}
\begin{proof}Suppose $X,Y\in\Gamma(VF)$. For the flow of $X$, we have
$$\frac{d}{dt}(\pi\circ\exp(tX))(f)=T_{\exp(tX)(f)}\pi\bigl(X(\exp(tX)(f))\bigr)=0,$$
so $\pi\circ\exp(tX)=\pi$ for all times for which the flow is defined. Therefore
$$T\pi([X,Y])=T\pi\frac{d}{dt}T(\exp(-tX))\circ Y\circ\exp(tX) |_{t=0}=\frac{d}{dt}T(\pi \circ\exp(-tX))\circ Y\circ\exp(tX) |_{t=0}=\frac{d}{dt}T\pi\circ Y\circ\exp(tX) |_{t=0}=0,$$
where, in the last equality, we used that $T\pi\circ Y=0$.
\end{proof}
The question turns out to be more subtle for $HF$. We begin by noting the following. (This proof is based on the discussion in \cite{duipri}.)
\begin{claim}\label{bramod}
Let $H$ be a distribution on $F$. For $X,Y\in\Gamma(H)$, $[X,Y](f)\;\mathrm{mod}\; H_f\in T_f F/H_f$ depends only on the values of $X$ and $Y$ at $f$. In particular, the question of involutivity can be answered pointwise.
\end{claim}
\begin{proof}
Let $r=\mathrm{rank}\,H$ and $n=\dim F$. Choose a neighbourhood $U$ of $f$ and vector fields $E_1,\ldots,E_r$ spanning $H|_U$. Extend them to a local frame $E_1,\ldots,E_n$ of $TF|_U$, and let $\alpha^{r+1},\ldots,\alpha^n$ be the corresponding dual coframe elements annihilating $H|_U$. For $X,Y\in\Gamma(H)$, $\alpha^j(X)=\alpha^j(Y)=0$, hence
$$\alpha^j([X,Y])=-d\alpha^j(X,Y)\qquad (j=r+1,\ldots,n).$$
The right-hand side at $f$ depends only on $X(f)$ and $Y(f)$. The values of the $\alpha^j$ determine the class of a tangent vector modulo $H_f$, so $[X,Y](f)\;\mathrm{mod}\; H_f$ depends only on $X(f)$ and $Y(f)$.
\end{proof}
Let us introduce the following $\Bbbk$-bilinear map $\mathcal{X}(F)\times\mathcal{X}(F)\ra{R}\Gamma(VF)$;
$$-R(X,Y):=P_V([P_H X,P_H Y]).$$
Because of claim \ref{bramod}, the following is true. 
\begin{corollary}
For $X,Y\in\mathcal{X}(F)$, $R(X,Y)(f)$ depends only on the values of $X$ and $Y$ at $f$. We can therefore view $R$ as a map $TF\wedge_F TF\ra{R}VF$ of vector bundles over $F$.\end{corollary}
\begin{proof}
$P_V([P_H X, P_H Y])$ is the unique representative of $[P_H X,P_H Y]\;\mathrm{mod} \; HF$ in $\mathrm{ker}\; P_H$. Since $P_H X,P_H Y\in \Gamma(HF)$, claim \ref{bramod} now tells us that we can see $R$ as a map $HF\otimes HF\ra{R}VF$ of vector bundles over $F$. To complete the proof, make the observation that $R(X,Y)(f)$ depends anti-symmetrically on $X$ and $Y$.
\end{proof}
This map is known as the \emph{curvature}\index{curvature} of the connection. Of course, by the theorem of Frobenius, $R$ measures whether $HF$ is integrable. This is the case precisely if $R=0$.

\subsection{Horizontal Lifts}\label{sec:hor}
It follows from the definition of a connection that the restriction of $T_f\pi$ to $HF_f$ defines a linear isomorphism to $T_{\pi(f)}M$. This allows us to lift a vector $v\in T_{\pi(f)}M$ to a unique horizontal vector $v_\hor\in HF_f$. Let $X\in\mathcal{X}(M)$ be a vector field. Then we can also construct its pointwise \emph{horizontal lift}\index{horizontal lift} $F\ra{X_\hor}HF\subset TF$. This is again a smooth vector field since it is nothing other than
\begin{center}
\begin{tikzcd}[column sep=normal]
F \arrow[r,"{\mathrm{id}_F\times \pi}"] & F\times M \arrow[r,"{\mathrm{id}_F\times X}"] & \pi^* TM \arrow[r,"\iota"] & HF\subset TF,
\end{tikzcd}
\end{center}
where we understand $\pi^* TM$ to be embedded in $F\times TM$ and where $TM\ra{\pi_M}M$ is the bundle projection. Summarising, we have obtained a unique horizontal lift:
\begin{center}
\begin{tikzcd}[column sep=large,row sep=large]
F \arrow[r,"{X_\hor}"] \arrow[d,"\pi"'] \arrow[drr,phantom,"\com",description] & HF \hspace{4pt} \subset & TF \arrow[d,"{T\pi}"]\\
M \arrow[rr,"X"'] & & TM.
\end{tikzcd}
\end{center}

How are the flows of these vector fields related? If $\gamma$ is a solution curve for the equation $df/dt=X_\hor(f)$, then $\pi\circ \gamma$ satisfies $dm/dt=X(m)$, so we see that
$$\pi\circ e^{t X_\hor}=e^{tX}\circ \pi.$$
Therefore, for any two $X,Y\in\mathcal{X}(M)$:
$$T_f\pi([X_\hor,Y_\hor](f))=[X,Y](\pi(f)).$$
We conclude that
\begin{align}\label{eq:curv}-R(X_\hor(f),Y_\hor(f))&=P_V([X_\hor,Y_\hor])(f)\nonumber\\
&=[X_\hor,Y_\hor](f)-[X,Y]_\hor(f).
\end{align}
where $TF\ra{P_V}VF$ is the projection onto the vertical bundle. We have found the following.

\begin{corollary}\label{curveasy}
If we use the identification $\pi^* TM\stackrel{\iota}{\cong} HF\subset TF$, we can view $R$ as a map $\pi^*(TM\wedge TM)\ra{R} VF$ given by $(X_m\wedge Y_m,f)\mapsto (X\wedge Y,f)\mapsto [X,Y]_\hor(f)-[X_\hor, Y_\hor](f)$, where $X,Y$ are vector fields on $M$ such that $X(m)=X_m,Y(m)=Y_m\in T_m M$ and $m=\pi(f)$. 
\end{corollary}

\subsection{Connections and Trivialisations}
The connections on a fibre bundle $F\ra{\pi}M$ form a sheaf. It is therefore very useful to understand what a connection looks like locally. 

In particular, let us take a local trivialisation $(U_\alpha,\psi_\alpha)$ for our fibre bundle around a point $f\in F$. If we write $S$ for the standard fibre, $F|_{U_\alpha}\ra{\psi_\alpha} U_\alpha\times S$. Then in this trivialisation 
\begin{align*}T\psi_\alpha(HF|_{U_\alpha})&=T\psi_\alpha \circ \iota(\pi^*( TM|_{U_\alpha}))\\
&=\{T\psi_\alpha\circ \iota (X,f)\; | \; \pi_M(X)=\pi(f)\; \}\\
&=\{ (X, -\Gamma^\alpha (X,s))\in TU_\alpha \times TS\; | \; (X,s)\in TM|_{U_\alpha}\times S \},
\end{align*}
for some smooth function $TM|_{U_\alpha}\times S\ra{\Gamma^\alpha} TS$, known as the Christoffel form for the trivialisation. This is of course defined as the composition
\begin{center}
\scriptsize
\begin{tikzcd}[column sep=small,row sep=large]
(X,s) \arrow[rrr,mapsto] & & & (X,\psi_\alpha^{-1}(\pi_M(X),s))\\
TM|_{U_\alpha}\times S \arrow[rr,"{\mathrm{id}_{TM|_{U_\alpha}}\times (\psi_\alpha^{-1}\circ(\pi_M\times \mathrm{id}_S))}"] \arrow[rrdddddd,"{\Gamma^\alpha}"'] & & (TM\times_M F)|_{U_\alpha} \arrow[dd,"\iota"'] & {} \arrow[dd,mapsto]\\
& & &\\
& & \begin{array}{c}HF\\ \cap\\  TF\end{array} \arrow[dd,"{T\psi_\alpha}"'] & \iota(X,\psi_\alpha^{-1}(\pi_M(X),s)) \arrow[dd,mapsto]\\
& & &\\
& & TM|_{U_\alpha}\times TS \arrow[dd,"{-\pi_{TS}}"'] & (X,-\Gamma^\alpha(X,s)) \arrow[dd,mapsto]\\
& & &\\
& & TS & \Gamma^\alpha(X,s).
\end{tikzcd}
\normalsize
\end{center}
More explicitly,
$$\Gamma^\alpha (X,s)=\pi_{TS}\circ T\psi_\alpha \circ P_V \circ T\psi_\alpha^{-1}(X,0_s),$$
where $TU_\alpha\times TS\ra{\pi_{TS}}TS$ is the projection onto the second component and $0_s$ denotes the zero element in $T_s S$.
Then the projection $P_V$ looks as follows:
\begin{equation}\label{eq:chri}\pi_{TS}\circ T\psi_\alpha\circ P_V\circ T\psi_\alpha^{-1}(X,\eta)=\eta+\Gamma^\alpha(X,\pi_S(\eta)),
\end{equation}
where $\eta\in TS$ and $TS\ra{\pi_S}S$ is the canonical projection.
\subsection{The Pullback}
The pullback is an essential tool when working with fibre bundles. It is a well-known fact that the categorical pullback of a fibre bundle along an arbitrary map exists in the smooth category. This follows from the fact that a fibre bundle is a submersion and therefore transversal to any map\footnote{Remember that a pair of maps $M_1\ra{f_1}N$, $M_2\ra{f_2}N$ is said to be transversal if $\mathrm{im}\; T_{m_1}f_1+\mathrm{im}\; T_{m_2}f_2=T_nN$ for every triple $(m_1,m_2,n=f_1(m_1)=f_2(m_2))\in M_1\times M_2\times N$. It is a direct consequence of the submersion theorem that the topological pullback of a pair of transversal maps is a smooth submanifold of $M_1\times M_2$.}. Therefore, given a fibre bundle $F\ra{\pi}M$ with standard fibre $S$ and a map $N\ra{\phi}M$, we can form a pullback
\begin{center}
\begin{tikzcd}[column sep=large,row sep=large]
\phi^*F \arrow[r,"{\pi^*\phi}"] \arrow[d,"{\phi^*\pi}"'] \arrow[dr,phantom,"\lrcorner",very near start] & F \arrow[d,"\pi"]\\
N \arrow[r,"\phi"'] & M
\end{tikzcd}
\end{center}
in the smooth category. Moreover, we then have the following.

\begin{claim}The pullback $\phi^*\pi$ is a smooth fibre bundle over $N$ with standard fibre $S$. Furthermore, if $\pi$ is equipped with a connection $P_V$, then the pullback induces a unique connection\index{pullback connection} $\phi^*P_V$ on $\phi^*\pi$ such that $T (\pi^*\phi)\circ \phi^*P_V=P_V\circ T(\pi^*\phi)$.
\end{claim}
\begin{proof}Choose some fibre bundle atlas $(U_\alpha,\psi_\alpha)$ for $\pi$. Then $(\phi^{-1}(U_\alpha),(\phi^*\pi,\pi_S\circ \psi_\alpha\circ \pi^*\phi))$ is seen to be a fibre bundle atlas for $\phi^*\pi$.

Why is $\phi^*P_V$ a connection, i.e. a projection $T(\phi^*F)\ra{}V(\phi^*F)$? The main observation is that $T_\xi(\pi^*\phi)$ restricts to a linear isomorphism $(V\phi^*F)_\xi=T_\xi(\phi^*F)_{\phi^*\pi(\xi)}\ra{\cong}T_{\pi^*\phi(\xi)}F_{\pi(\pi^*\phi(\xi))}= (VF)_{\pi^*\phi(\xi)}$. (Indeed, $\pi^*\phi$ is an isomorphism $(\phi^*F)_{\phi^*\pi(\xi)}\ra{}F_{\pi(\pi^*\phi(\xi))}$.)

This means that the following commuting square uniquely defines $\phi^*P_V$. 
\begin{center}
\scriptsize
\begin{tikzcd}[column sep=large,row sep=large]
T_\xi(\phi^*F) \arrow[r,"{T_\xi(\pi^*\phi)}"] \arrow[d,"{(\phi^*P_V)_\xi}"'] \arrow[dr,phantom,"\com",description] & T_{\pi^*\phi(\xi)}F \arrow[d,"{(P_V)_{\pi^*\phi(\xi)}}"]\\
\begin{array}{c}(V\phi^*F)_\xi\\=T_\xi(\phi^*F)_{\phi^*\pi(\xi)}\end{array} \arrow[r,"\cong","{T_\xi(\pi^*\phi)}"'] & \begin{array}{c}(VF)_{\pi^*\phi(\xi)}\\=T_{\pi^*\phi(\xi)}F_{\pi(\pi^*\phi(\xi))}.
\end{array}
\end{tikzcd}
\normalsize
\end{center}
Finally, note that $T_\xi(\pi^*\phi)\circ (\phi^*P_V)_{\xi}^2=(P_V)_{\pi^*\phi(\xi)}\circ T_\xi(\pi^*\phi)\circ (\phi^*P_V)_\xi=(P_V)_{\pi^*\phi(\xi)}^2\circ T_\xi(\pi^*\phi)=(P_V)_{\pi^*\phi(\xi)}\circ T_\xi(\pi^*\phi)=T_\xi(\pi^*\phi)\circ (\phi^*P_V)_\xi$ and therefore, since $(V\phi^*F)_\xi\ra{T_\xi(\pi^*\phi)}(VF)_{\pi^*\phi(\xi)}$ is an isomorphism, $(\phi^*P_V)_\xi^2=(\phi^*P_V)_\xi$. The same defining square shows that $(\phi^*P_V)_\xi$ is the identity on $(V\phi^*F)_\xi$; hence its image is exactly $V\phi^*F$ and its kernel is a complementary horizontal subspace.
\end{proof}

\begin{remark}
Suppose we have two fibre bundles $F_1\ra{\pi_1}M\la{\pi_2}F_2$. We know that $FB_M$ has finite products. Now, suppose $\pi_i$ is equipped with a connection $(P_V)_i$ and write $F_1\la{\pi^1}F_1\times_M F_2\ra{\pi^2}F_2$ for the projections of the product. Then we can construct a connection $\pi^1 {}^*(P_V)_1+\pi^2 {}^*(P_V)_2=(P_V)_1\times_M(P_V)_2$ on $F_1\times_MF_2$, which is called the \emph{product connection}\index{product connection}.
\end{remark}

\subsection{Parallel Transport}
As noted above, there are many different (more or less) equivalent definitions of connections. I find one approach to be particularly intuitive, much more so than the one I have taken so far. It can be argued that a connection should be viewed as a rule that identifies fibres of a fibre bundle along paths in the base space $M$. This identification is called \emph{parallel transport}\index{parallel transport}. 

I will show how to derive a rule for parallel transport from our definition of a connection. Moreover, it can be shown that the connection can be recovered from its parallel transport. (See \cite{dumcon}.) Apart from being a very intuitive approach to connections and being a very practical tool in, say, general relativity, this definition generalises well to higher categorical analogues of fibre bundles (which appear in higher gauge theory). (See \cite{schpar}.)

In section \ref{sec:hor}, we have seen that we can construct a unique horizontal lift $X_\hor\in \mathcal{X}(F)$ for each vector field $X\in \mathcal{X}(M)$. This is by lifting each vector pointwise. Suppose we are given a curve $]a,b[\ra{c} M$. Notice that the tangent vectors to this curve are elements of $TM$ and can therefore be uniquely lifted to vectors in $HF\subset TF$. We can ask whether these vectors give rise to a (local) curve in $F$. The answer is the following.
\begin{theorem}[Parallel transport] \label{partrans}Let $HF$ be a connection on a fibre bundle $F\ra{\pi}M$ and let $I \ra{c} M$ be a smooth curve. (Here, $I$ is a connected neighbourhood of $0$ in $\R$.) Then there exists a neighbourhood $\Omega_c$ of $F_{c(0)}\times\{0\}$ in $F_{c(0)}\times I$ and a smooth map $\Omega_c\ra{T_c} F$ such that $T_c(f,0)=f$ and, for all $(f,t)\in \Omega_c$:
\begin{enumerate}
\item $\pi(T_c(f,t))=c(t)$;\\
\item $\frac{d}{dt} T_c(f,t)\in HF$;\\
\item $T$ is reparametrisation invariant: If $J\ra{\phi}I$ is an open interval, then, for all $s\in J$: $T_c(f,\phi(s))=T_{c\circ \phi}(T_c(f,\phi(0)),s)$.
\end{enumerate}
\end{theorem}
\begin{proof}My favourite proof of this theorem is the following. We construct the following two pullbacks along $c$:
\begin{center}
\begin{tikzcd}[column sep=large,row sep=large]
c^* HF \arrow[r,"{c^*\pi_F}"] \arrow[d,"{\pi_F^*\pi^*c}"] \arrow[dr,phantom,"\lrcorner",very near start] & c^* F \arrow[r,"{c^*\pi}"] \arrow[d,"{\pi^* c}"] \arrow[dr,phantom,"\lrcorner",very near start] & I \arrow[d,"c"]\\
HF \arrow[r,"{\pi_F}"] & F \arrow[r,"\pi"] & M.
\end{tikzcd}
\end{center}
Then $c^*HF$ is a connection on $c^*F\ra{c^*\pi}I$. Note that this distribution has dimension $1$, since it is isomorphic to $(c^*\pi)^* TI$. This implies that it is involutive and by Frobenius, it is therefore integrable to a foliation $\mathcal{F}$ of $c^* F$ in $1$-dimensional leaves. Let $\mathcal{F}_{f,0}$ denote the leaf of $(f,0)\in c^*F$. Then, since $T_{(f,0)}\mathcal{F}_{f,0}=c^*HF_{(f,0)}$ is an algebraic supplement to $\ker T_{(f,0)} c^*\pi$ (indeed, $c^*HF$ is a connection), we see (by the inverse function theorem) that $c^*\pi|_{\mathcal{F}_{f,0}}$ defines a local diffeomorphism around $(f,0)$. Denote its local inverse by $T_c(f,-)$. This map obviously satisfies the initial condition and conditions 1 and 2; the same argument after reparametrising $c$ gives condition 3.\\
\\
The following is more explicit.\\
Note that $c'(t)\in T_{c(t)}M$. By using cut-off functions, we can extend $c'(t)$ locally to a time-dependent vector field $M\supset U\ra{X_t} TM|_U$. Let $X_{t,\hor}$ be its horizontal lift. The time-dependent ODE
$$\frac{d}{dt}f(t)=X_{t,\hor}(f(t)),\qquad f(0)=f,$$
has a unique local solution. Its projection solves $\frac{d}{dt}\pi(f(t))=X_t(\pi(f(t)))$ with initial value $c(0)$, and hence equals $c(t)$ after shrinking the interval if necessary. Therefore $T_c(f,t):=f(t)$ satisfies all conditions.\\
\\
Finally, the standard proof is the following.\\
Let $(U_\alpha,\psi_\alpha)$ be a local trivialisation for $F$ around $c(0)$. Note that \textit{2.} can equivalently be stated by demanding that $P(\frac{d}{dt} T_c(f,t))=0$, if we write $P$ for the projection onto the vertical bundle. In the local chart this condition becomes
$$\frac{d}{dt} \Phi_c(s_0,t)=-\Gamma^\alpha (\frac{d}{dt} c(t),\Phi_c(s_0,t)),$$ 
where we write $\Phi_c(s_0,t)$ for $\pi_S(\psi_\alpha (T_c(\psi_\alpha ^{-1}(c(0),s_0),t)))\in S$. Now, $-\Gamma^\alpha (\frac{d}{dt} c(t),\Phi_c(s_0,t))$ is a smooth (time-dependent) vector field, so this equation has a unique maximal local solution $S\times I \supset \Omega_c ' \ra{\Phi_c} S$. Since $\psi_\alpha$ is a diffeomorphism, we obtain our $\Omega_c\ra{T_c}F$ by $T_c(f,t)=\psi_\alpha^{-1}(c(t),\Phi_c(\pi_S(\psi_\alpha(f)),t))$.

Property 3 is also easily checked. Indeed, write $c_\hor(t):=T_c(f,t)$ for a horizontal lift of $c$. Then $\frac{d}{ds}(c_\hor\circ\phi(s))=\phi '(s) \cdot c_\hor '(\phi(s))$ and $\pi\circ c_\hor\circ\phi=c\circ\phi$, so $c_\hor\circ\phi$ is a horizontal lift of $c\circ\phi$. Finally $c_\hor\circ\phi(0)=T_c(f,\phi(0))$, so the conclusion follows.

\end{proof}
$T_c(f,t)$ is called the parallel transport of $f$ along $c$ over time $t$. Note that this parallel transport might not be defined for all $t\in I$.\\
\\
\subsection{Complete Connections and Holonomy}
Unfortunately, the existence of parallel transport is only proved locally by theorem \ref{partrans}. Global transport might not be possible for all connections. However, many useful connections allow global parallel transport. Principal connections and affine connections, with the Levi-Civita connections of a (pseudo)-Riemannian manifold, come to mind. Such connections are said to be \emph{complete}\index{complete connection}.

Now, it is easily verified that, in the proof of theorem \ref{partrans}, we only needed $c$ to be a piecewise-$C^1$-curve\footnote{It can be shown that one can equivalently take the curves to be $C^k$, for any $k\in\N_{>0}$, or even $C^\infty$. \cite{kobfou}}. Thus, for a complete connection, parallel transport along piecewise-$C^1$ loops based at $m\in M$ assigns to each such loop a diffeomorphism of the fibre $F_m$:
\begin{center}
\begin{tikzcd}[column sep=large,row sep=large]
\mathrm{Loop}(M,m) \arrow[r,"{\mathrm{Hol}_m}"] & \mathrm{Diff}(F_m)\\
c \arrow[r,mapsto] & T_c(-,1).
\end{tikzcd}
\end{center}
This assignment is compatible with concatenation and reversal of paths, so its image is a subgroup of $\mathrm{Diff}(F_m)$, known as the \emph{holonomy group}\index{holonomy group} at $m$. This group is a useful tool for classification of fibre bundles with connections, (pseudo)-Riemannian manifolds with the Levi-Civita connection being an important example.

\subsection{Curvature as Infinitesimal Holonomy}
We have seen that holonomy cannot be defined for arbitrary connections, since the parallel transport might not be globally defined along any loop. However, if we make the loops small enough, one would expect that the transport would be possible. Here, "small enough" will mean infinitesimal. We will see that this infinitesimal transport is nothing but the curvature, which, indeed, is defined for any connection.

The following result, in some sense, tells us that we can view the curvature as a limit of the holonomy if we take smaller and smaller loops. (We take loops corresponding to two commuting vector fields, or equivalently (by Frobenius), two local coordinate vector fields.)
\begin{claim}[Curvature as infinitesimal holonomy]\index{curvature as infinitesimal holonomy}\label{claim:curvinf} If $X,Y\in \mathcal{X}(M)$ are two commuting vector fields, we have that
$$-R(X,Y)(f)=\frac{\dau^2}{\dau s\dau t}\left(e^{-t Y_\hor}\circ e^{-s X_\hor}\circ e^{t Y_\hor}\circ e^{s X_\hor}(f)\right)\Big|_{s=t=0}.$$
\end{claim}
\begin{proof}We have already derived formula \ref{eq:curv}, which expresses the curvature in terms of the horizontal lifts of $X$ and $Y$. If $[X,Y]=0$, this reduces to the equality
$$-R(X,Y)(f)=[X_\hor,Y_\hor](f).$$
The claim follows by using the following well-known identity for the bracket of two vector fields:
$$[X_\hor,Y_\hor](f)=\frac{\dau^2}{\dau s\dau t}\left(e^{-t Y_\hor}\circ e^{-s X_\hor}\circ e^{t Y_\hor}\circ e^{s X_\hor}(f)\right)\Big|_{s=t=0}. $$
\end{proof}
\section{Connections on Principal Bundles}
Suppose that $P\ra{\pis}M$ is a principal $G$-bundle. We ask ourselves what a good notion of a connection on a principal bundle is. $G$-equivariant connections turn out to be particularly interesting. Later, we will see that, in the case $G=\GL(\Bbbk^k)$, precisely these connections correspond to the ordinary affine connections on the canonical associated vector bundle $P[\lambda]$. In this sense, they are a generalisation of affine connections to general principal bundles. Before we explain the precise condition of equivariance that we put on the connection, however, let us formulate another, more conventional way of looking at connections on principal bundles. 
\subsection{Principal Connection 1-forms}
We note that the principal right action (which we write as $\rho$) defines fundamental vector fields for the Lie algebra $\g$ of $G$ as follows. For $Z\in\g$ and $p\in P$, we write $Z_P(p):=\frac{d}{dt}\rho_{\exp(tZ)}(p)|_{t=0}=T\rho^p(Z)$, where $G\ra{\rho^p}P$ denotes the smooth map $g\mapsto \rho(p,g)$. For a right action this assignment is a Lie-algebra homomorphism $\g\ra{}\mathcal{X}(P)$. Recall that, since the vertical bundle is integrable, $\Gamma(VP)$ forms a Lie subalgebra of $\mathcal{X}(P)$. Now, since the orbits of the principal right action are the fibres of $P$, the same assignment gives a linear map $\g \ra{} \Gamma(VP)$. Moreover, this map trivialises $VP$.
\begin{claim}\label{cl:triv}
The map 
\begin{center}
\begin{tikzcd}[column sep=large,row sep=large]
P\times\g \arrow[r,"\zeta"] & VP\\
(p,Z) \arrow[r,mapsto] & Z_P(p)
\end{tikzcd}
\end{center}
defines an isomorphism of vector bundles over $P$.
\end{claim}
\begin{proof}$\zeta$ respects the fibre, since $Z_P(p)\in VP_p$. (Put differently, $\zeta$ is the identity map on $P$.) Linearity on each fibre follows because $Z\mapsto Z_P(p)=T_e\rho^p(Z)$ is linear:
$$(Z+\lambda Z')_P(p)=T_e\rho^p(Z+\lambda Z')=Z_P(p)+\lambda Z'_P(p).$$
Let $p\in P$ and let $(U,\psi)$ be a trivialisation of $P$ around $p$.

Now, $\zeta(p,Z)=T\rho^p(Z)$. Note that $\rho^p$ is a diffeomorphism $G\ra{}P_p$. Therefore $\g= T_eG\ra{T_e\rho^p}T_p P_p=VP_p$ is an isomorphism of vector spaces. $\zeta$ is obviously smooth, as $\zeta(p,Z)=T\rho^p(Z)=T\rho(0_p,Z)$. Since $\zeta$ is a smooth map of fibre bundles that is a diffeomorphism on both the base space (there it is even the identity) and on each fibre, we conclude that it is a diffeomorphism on the total space.
\end{proof}
There is a generalised notion of a differential form that is very useful in this situation. Let $M$ be a manifold and let $V\ra{\pi}M$ be a vector bundle. We define a \emph{$\pi$-valued differential $k$-form}\index{vector-valued differential form} on $M$ as a map $\eta$ of vector bundles\footnote{Note that we can equivalently view these $\pi$-valued differential forms as sections of the vector bundle $V\otimes \bigwedge^kT^*M$.}:
\begin{center}
\begin{tikzcd}[column sep=large,row sep=large]
\bigwedge^k TM \arrow[r,"\eta"] \arrow[d] & V \arrow[d,"\pi"]\\
M \arrow[r,equal] & M
\end{tikzcd}
\end{center}
The space of such forms is denoted $\Omega^k(M,\pi)$. In the case of a trivial vector bundle $M\times W\ra{}M$ one also talks about $W$-valued differential forms and writes $\Omega^k(M,W)$.

The previous result tells us that we can equivalently see a (not necessarily principal) connection  $P_V$ and its curvature $R$ as maps of vector bundles $\omega$ and $\Omega^\omega$:
\begin{center}
\begin{tikzcd}[column sep=large,row sep=large]
TP \arrow[r,"{\omega=\zeta ^{-1} \circ P_V}"] \arrow[d] \arrow[dr,phantom,"\com",description] & P\times \g \arrow[d] & \mathrm{and} & TP\wedge TP \arrow[r,"{\Omega^\omega=\zeta^{-1}\circ R}"] \arrow[d] \arrow[dr,phantom,"\com",description] & P\times \g \arrow[d]\\
P \arrow[r,equal] & P & & P \arrow[r,equal] & P.
\end{tikzcd}
\end{center}
We therefore have an interpretation of $P_V$ as a $\g$-valued $1$-form $\omega\in\Omega^1(P;\g)$: $P_V(X_p)= \omega(X_p)_P(p)$ and we can see $R$ as a $\g$-valued $2$-form $\Omega^\omega\in\Omega^2(P;\g)$:
$R(X_p,Y_p)=(\Omega^\omega(X_p,Y_p))_P(p)$. 
\begin{claim}\label{claim:repro} Some $\omega\in\Omega^1(P,\g)$ represents a connection precisely if $\omega(Z_P)=Z$ for all $Z\in\g$.
\end{claim}
\begin{proof}
If $\omega$ represents a connection, then $\zeta\circ\omega=P_V$ is the vertical projection. Hence $\zeta(\omega(Z_P))=P_V(Z_P)=Z_P$, because $Z_P$ is vertical. Since $Z\mapsto Z_P(p)$ is a linear isomorphism $\g\to VP_p$ for each $p\in P$, this implies $\omega(Z_P)=Z$.

Conversely, suppose $\omega(Z_P)=Z$ for all $Z\in\g$ and put $P_V:=\zeta\circ\omega$. Then $P_V$ restricts to the identity on $VP$, since every vertical vector is of the form $Z_P(p)$. Thus $P_V^2=P_V$ with image $VP$, and therefore $P_V$ is the vertical projection of a connection.
\end{proof}
We can also express $\Omega^\omega$ directly in terms of $\omega$, like we expressed $R$ in terms of $P_V$.
\begin{claim}
 The following formula holds for the curvature: $$\Omega^\omega(X,Y)=-\omega([P_H X, P_H Y]).$$
\end{claim}
\begin{proof}Indeed, $-\omega([P_H X,P_H Y])_P=-P_V([P_H X,P_H Y])=R(X,Y)=\Omega^\omega(X,Y)_P$. Again, since $\zeta$ is an isomorphism, the conclusion follows.\end{proof}
Now, let us formulate what we mean precisely when we say that a connection is equivariant. We say that $P$ is equipped with a \emph{principal connection}\index{principal connection}, or \emph{equivariant connection}\index{equivariant connection}, if $HP\la{P_H}TP\ra{P_V}VP$ is equivariant in the following sense: $P_V\circ T(\rho_g)=T(\rho_g)\circ P_V$, for all $g\in G$. (Here I write $\rho_{(-)}$ for the principal right action to be explicit.) Note that we can equivalently demand that horizontal lifts are $G$-equivariant or that parallel transport along any curve is. We must ask ourselves what it means for $\omega$, if we are dealing with a principal connection. The answer turns out to be as follows. 

Remember that we have the so-called \emph{adjoint representation}\index{adjoint representation} of $G$ on its Lie algebra $\g$
\begin{center}
\begin{tikzcd}[column sep=large,row sep=large]
G \arrow[r,"{\mathrm{Ad}}"] & \GL(\g)\\
g \arrow[r,mapsto] & T_e C_g,
\end{tikzcd}
\end{center}
where $C_g=g\cdot -\cdot g^{-1}$ is the conjugation map.
\begin{claim}\label{claim:equiv}
Suppose that $\omega\in \Omega^1(P,\g)$ is the $\g$-valued 1-form corresponding to a (not necessarily principal) connection $P_V$. Then this is precisely a principal connection if $\omega$ is $G$-equivariant in the sense that for all $g\in G$ the following diagram commutes:
\begin{center}
\begin{tikzcd}[column sep=large,row sep=large]
TP \arrow[r,"{T\rho_g}"] \arrow[d,"\omega"'] & TP \arrow[d,"\omega"]\\
\g \arrow[r,"{\mathrm{Ad}(g^{-1})}"'] & \g,
\end{tikzcd}
\end{center}
where $G\ra{\mathrm{Ad}} \GL(\g)$ denotes the adjoint representation of $G$.
\end{claim}
\begin{proof}Let $X\in\mathcal{X}(P)$. Then
\begin{align*}\omega(T\rho_g(X))_P\circ\rho_g&=P_V(T \rho_g(X))\\
&=T\rho_g(P_V(X))\txt{(equivariance of $P_V$)}\\
&=T\rho_g(\omega(X)_P)\\
&=T\rho_g\circ \frac d{dt}\rho_{\exp(t\omega(X))}|_{t=0}\\
&=\frac d{dt}\rho_g\circ\rho_{\exp(t\omega(X))}|_{t=0}\\
&=\frac d{dt}\rho_g\circ\rho_{\exp(t\omega(X))}\circ\rho_{g^{-1}}|_{t=0}\circ\rho_g\\
&=\frac d{dt}\rho_{g^{-1}\exp(t\omega(X))g}|_{t=0}\circ\rho_g\\
&=\frac d{dt}\rho_{\exp(t\,\mathrm{Ad}(g^{-1})\omega(X))}|_{t=0}\circ\rho_g\txt{(*)}\\
&=(\mathrm{Ad}(g ^{-1})\omega(X))_P\circ\rho_g,
\end{align*}
where, in $(*)$, we use the fact that $\mathrm{Ad}(g)=T_e C_g$, if $C_g: h\mapsto ghg^{-1}$, which implies that
\begin{center}
\begin{tikzcd}[column sep=large,row sep=large]
G \arrow[r,"{C_g}"] & G\\
\g \arrow[u,"\exp"] \arrow[r,"{T_e C_g=\mathrm{Ad}(g)}"'] \arrow[ur,phantom,"\com",description] & \g \arrow[u,"\exp"'].
\end{tikzcd}
\end{center}
Finally, since $\rho_g$ is a diffeomorphism, equivariance of $P_V$ implies the displayed equivariance of $\omega$. Conversely, if $\omega$ satisfies the displayed equivariance, the same computation read backwards gives $P_V(T\rho_g(X))=T\rho_g(P_V(X))$. Hence $P_V$ is equivariant, and the claim follows.
\end{proof}
In general, given a representation $\lambda$ of $G$ on a vector space $W$, we say that $\alpha\in \Omega^k(P,W)$ is an \emph{equivariant differential form}\index{equivariant differential form} if $\alpha(T\rho_g(X_1),\ldots,T\rho_g(X_k))=\lambda(g^{-1})\alpha(X_1,\ldots,X_k)$, for all $g\in G$.
\begin{claim}\label{liead}
For a principal connection $\omega$ on a principal bundle $P$, $\Li_{Z_P}\omega=-\mathrm{ad}(Z)\circ \omega$.
\end{claim}
\begin{proof}
\begin{align*}\Li_{Z_P}\omega&:=\frac{d}{dt}\exp(tZ_P)^*\omega|_{t=0}\\
&=\frac{d}{dt}\omega\circ T\exp(tZ_P)|_{t=0}\\
&=\frac{d}{dt}\omega\circ T\rho_{\exp(tZ)}|_{t=0}\\
&=\frac{d}{dt}\mathrm{Ad}(\exp(-tZ))\circ\omega|_{t=0}\txt{(by $G$-equivariance)}\\
&=-\mathrm{ad}(Z)\circ\omega.
\end{align*}
Here the last equality holds by the chain rule. (Recall that $\mathrm{ad}:=T_e\mathrm{Ad}$.)
\end{proof}
The bracket on $\g$ defines a bracket on the graded vector space $\Omega(M,\g):=\bigoplus_{k\in\N}\Omega^k(M,\g)$ of all $\g$-valued differential forms on a manifold $M$, making it into a graded Lie algebra. For $\alpha\in\Omega^k(M,\g)$ and $\beta\in\Omega^l(M,\g)$, we define $[\alpha,\beta]\in\Omega^{k+l}(M,\g)$ by
$$[\alpha,\beta](X_1,\ldots,X_{k+l}):=\frac{1}{k!l!}\sum_\sigma \mathrm{sign}\,\sigma [\alpha(X_{\sigma(1)},\ldots,X_{\sigma(k)}),\beta(X_{\sigma(k+1)},\ldots,X_{\sigma(k+l)})],$$
where $\sigma$ ranges over all permutations of $\{1,\ldots,k+l\}$ and $X_i\in TM$.

The ordinary exterior derivative extends componentwise to $\g$-valued forms. For $\alpha\in \Omega^k(M,\g)$ and $\beta\in \Omega^l(M,\g)$ it satisfies the graded Leibniz rule
$$d[\alpha,\beta]=[d\alpha,\beta]+(-1)^k[\alpha,d\beta].$$
For a $\g$-valued $1$-form, $d\omega(X,Y):=\Li_X(\omega(Y))-\Li_Y(\omega(X))-\omega([X,Y])$. This implies that the following proposition holds for $G$-equivariant connections.

First, however, we introduce one more notation that makes the proposition somewhat cleaner. Moreover, the notation is used in practically all physics literature on gauge theory. Recall that we used $P_H$ to designate the projection onto the horizontal bundle. This allows us to define a map 
$$\Omega(P,W)\ra{P_H^*} \Omega(P,W)$$
$$(P_H^*\alpha)(X_1,\ldots,X_k):=\alpha(P_H(X_1),\ldots,P_H(X_k)).$$
We might call this the projection onto the subspace $\Omega_\hor(P,W)$ of \emph{horizontal differential forms}\index{horizontal differential form}. Obviously, these are precisely the differential forms that vanish on vertical vectors. Finally, we define the \emph{exterior covariant derivative}\index{exterior covariant derivative} or \emph{covariant differential}\index{covariant differential} associated with the connection $\omega$ as the map $\Omega^k(P,W)\ra{d^\omega}\Omega^{k+1}(P,W)$, given by $d^\omega:=P_H^*\circ d$.
 
\begin{theorem}[Cartan's second structural equation] \label{maurer}For a principal connection $\omega$ on $P$ the following formula holds for the curvature:
$\Omega^\omega=d\omega+\frac{1}{2}[\omega,\omega]=d^\omega \omega$.
\end{theorem}
\begin{proof}
Let $X,Y$ be horizontal vector fields on $P$. Then $R(X,Y):=-P_V([X,Y])=(-\omega([X,Y]))_P$ and \begin{align*}(d\omega+\frac{1}{2}[\omega,\omega])(X,Y)&=\Li_X(\omega(Y))-\Li_Y(\omega(X))-\omega([X,Y])+[\omega(X),\omega(Y)]\\
&=-\omega([X,Y]),
\end{align*}
since $0=P_VX=(\omega(X))_P$ implies that $\omega(X)=0$ and similarly for $Y$. So on horizontal vector fields, both expressions agree.\\
\\
Now note that, since the principal right action is transitive on the fibre, any vertical vector $Z'\in VP_p$ can be written as $Z_P(p)$ for some $Z\in\g$. Now, trivially, $i_{Z_P} R=i_{P_H Z_P} R=0$, so $i_{Z_P}\Omega^\omega=0$.

On the other hand, since $i_{Z_P}\omega=Z$, we have that $i_{Z_P} d \omega=\Li_{Z_P} \omega$. So
\begin{align*}i_{Z_P} (d\omega+\frac{1}{2} [\omega,\omega])&=\Li_{Z_P}\omega+ \frac{1}{2}i_{Z_P}[\omega,\omega]\\
&=\Li_{Z_P}\omega+[Z,\omega]\\
&=\Li_{Z_P}\omega+\mathrm{ad}(Z)\circ\omega
\end{align*}
(We mean $[Z,\omega]$ to be the pointwise $\g$-Lie bracket.) This last equality is the well-known result on the adjoint representation. Since $\omega$ is $G$-equivariant, claim \ref{liead} tells us that this vanishes. Both expressions therefore vanish on vertical vectors. By noting that both are $C^\infty(P)$-bilinear maps, the first equality in the theorem follows.

For the second one, note that
\begin{align*}d^\omega \omega(X\wedge Y)&=(P_H^* d\omega)(X\wedge Y)\\
&=d\omega(P_H (X)\wedge P_H (Y))\\
&=\Li_{P_H(X)}\omega(P_H(Y))-\Li_{P_H(Y)}\omega(P_H(X))-\omega([P_H(X),P_H(Y)])\\
&=-\omega([P_H(X),P_H(Y)])\txt{(since $\omega$ vanishes on horizontal vectors)}.
\end{align*}
Note that this expression is equal to $-\omega([X,Y])$ (and therefore agrees with the expression above) if $X,Y$ are horizontal vectors. Moreover it obviously vanishes on vertical vectors. We conclude that the second identity in the theorem holds.
\end{proof}
There is one more identity regarding the connection that I would like to show, since it is of fundamental importance in physics. It will give us one of the two equations governing the evolution of gauge fields.
\begin{theorem}[Bianchi Identity]\index{Bianchi identity}The curvature form $\Omega^\omega$ of a principal connection $\omega$ satisfies the following equation, known as the Bianchi Identity:
$$d^\omega \Omega^\omega=0.$$
\end{theorem}
\begin{proof}We have, using the previous result,
\begin{align*}
d^\omega \Omega^\omega&=d^\omega(d\omega+\frac{1}{2} [\omega,\omega])\\
&=P_H^*dd\omega+\frac 1 2 P_H^* d[\omega,\omega]\\
&=\frac 1 2 P_H^*([d\omega,\omega]-[\omega,d\omega])\\
&=P_H^* [d\omega,\omega]\\
&=[P_H^*d\omega,P_H^*\omega]\\
&=0.
\end{align*}
The last identity, of course, holds since $\omega$ vanishes on horizontal vectors.
\end{proof}
It should be noted that any connection on a fibre bundle satisfies a variant of this equation: it is not special to the case of a principal connection. However, the general statement would be much more involved. For more information, see \cite{micnat}.\\
\\
To express the second field equation (known as the Yang-Mills equation) we need a few more definitions. Let us, following convention, first introduce the notation $\Omega_\hor^k(P,W)^G$ for the space of $W$-valued $k$-forms that are both horizontal (vanish on vertical vectors) and equivariant. This is a very important class of differential forms. For instance, note that $\Omega^\omega\in\Omega_\hor^2(P,\g)^G$. We have a similar result for $\omega$.
\begin{claim}
Let $\omega$ be a principal connection on $P\ra{\pis}M$. Then another $\omega'\in \Omega^1(P,\g)$ is a principal connection precisely if $\omega-\omega'\in\Omega_\hor^1(P,\g)^G$. In particular, after fixing one principal connection, the space of principal connections is an affine space modelled on $\Omega_\hor^1(P,\g)^G$.
\end{claim}
\begin{proof}Obviously, $\omega'$ is precisely equivariant if $\omega-\omega'$ is. Moreover, for $X_V\in VP$ $\zeta\circ\omega'(X_V)=\zeta\circ(\omega'-\omega)(X_V)+\zeta\circ\omega(X_V)\exeq\zeta\circ\omega(X_V)=P_V(X_V)=X_V$, where $\exeq$ holds precisely if $\omega-\omega'$ is horizontal.\end{proof}
Let us write $\mathscr{C}(P)$ for the subset of $\Omega^1(P,\g)$ of principal connections. In fact, this is an (infinite-dimensional) submanifold of $\Omega^1(P,\g)$. \cite{blegau} We can interpret $\Omega_\hor^1(P,\g)^G$ as its tangent space. Indeed, let $]-\epsilon,\epsilon[\ra{\gamma}\mathscr{C}(P)$ be a path. Then, $\gamma(t)=\gamma(0)+\delta(t)$, with $]-\epsilon,\epsilon[\ra{\delta}\Omega_\hor^1(P,\g)^G$. We see that $\gamma'(0)\in T_0\Omega_\hor^1(P,\g)^G\cong \Omega_\hor^1(P,\g)^G$.

\begin{remark}\label{identifi}I can imagine the definition of $\Omega_\hor^k(P,W)^G$ can seem somewhat far-fetched. How precisely should we think about these differential forms? To answer this question note that, by claim \ref{seccoo}, we have a natural identification $\Omega_\hor^0(P,W)^G=C^\infty(P,W)^G\cong \Gamma(P[\lambda])$. Because of the identification $HP_p\cong TM_{\pis(p)}$ given by the connection and the transitivity of the principal right action on the fibres of $\pis$ we therefore obtain a natural identification $\Omega_\hor^k(P,W)^G\cong \Gamma(P[\lambda]\otimes \bigwedge^kT^*M)\cong\Omega^k(M,P[\lambda]\ra{\pi}M)$. 

If $M$ has an orientation, we can define a Hodge star operation on $\Omega(M,\pi)$. This can be transferred to $\Omega_\hor(P,W)^G$, which will be the subject of the next section. It contains some technical definitions that will be essential for the gauge theories. The reader can skip this paragraph and only return to it later when the physics demands these definitions.\end{remark}

\subsection{The Covariant Codifferential}
Note that $d^\omega$ restricts to a map
$$\Omega_\hor^k(P,W)^G\ra{d^\omega}\Omega_\hor^{k+1}(P,W)^G.$$
What we need to formulate the second field equation is an analogue of the \emph{codifferential} for principal bundles, like the covariant differential is an analogue of the ordinary differential. This will be a map
$$\Omega_\hor^k(P,W)^G\ra{\delta^\omega}\Omega_\hor^{k-1}(P,W)^G,$$
which we shall call the \emph{covariant codifferential}\index{covariant codifferential}. To define it, we first generalise the \index{Hodge star operation}\emph{Hodge star operation}\footnote{We assume the reader to be familiar with the principles of the Hodge star for (scalar-valued) differential forms on an oriented (pseudo)-Riemannian manifold. Recall that the definition of the ordinary codifferential $\delta$ does not depend on the choice of the orientation. So we can always define $\delta$ by choosing an arbitrary orientation locally. The situation here will be similar. We do not need an orientation to define the covariant codifferential.} to $\Omega_\hor(P,W)^G$. 

Let $(M,g)$ be an $n$-dimensional (pseudo)-Riemannian manifold. We generalise the Hodge star to vector-valued differential forms to be the unique $C^\infty(M)$-linear operator $\Omega^{k}(M,W)\ra{*}\Omega^{n-k}(M,W)$, such that $*(w\otimes \eta)=w\otimes *\eta$, for all $\eta\in \Omega^k(M)$ and $w\in W$. \cite{baegau} Now, we seem to run into a problem. There seems to be no way to define a Hodge star on $P$, since $P$ has no metric. We can, however, define the Hodge star on the horizontal differential forms.

Indeed, since we have the isomorphism $HP_p\ra{T_p\pis}T_{\pis(p)}M$, we can define $$\bigwedge^k(HP_p^*)\ra{\tilde *_p}\bigwedge^{n-k}(HP_p^*)$$
by $\tilde *_p(\pis^* \tau):=\pis^*(*_{\pis(p)}\tau)$, for $\tau\in \bigwedge^k T_{\pis(p)}^*M$. (If we are assuming we have a Hodge star $*$ on $M$.) Now, finally, define
$$\Omega_\hor^k(P,W)^G\ra{
\bar *}\Omega_\hor^{n-k}(P,W)^G$$
by setting $(\bar * \alpha)_p$ to be the unique $W$-valued $(n-k)$-form, vanishing on vertical vectors, such that $\bar * \alpha|_{HP_p}=\tilde *_p(\alpha|_{HP_p})$.

It is easily verified that the Hodge star of an equivariant differential form is again equivariant. Indeed, a volume form $\mu \in \Omega^n(M)$ pulls back to a horizontal $n$-form $\pis^*\mu\in\Omega_\hor^n(P)$. This restricts to a top form on $HP_p$ that is invariant under the principal right action. Therefore $\bar *$ respects equivariance. 

Now, in analogy with the ordinary codifferential, we define
$$\Omega_\hor^k(P,W)^G\ra{\delta^\omega}\Omega_\hor^{k-1}(P,W)^G,$$
$$\delta^\omega:=\mathrm{sign}\,(g)(-1)^{nk+n+1}\bar * d^\omega \bar *.$$
This will turn out to be an essential tool in physics, as the second field equation (in the homogeneous case) taking the form
$$\delta^\omega \Omega^\omega=0.$$
As was the case with the ordinary codifferential, this covariant codifferential can be understood to be the adjoint of the covariant differential with respect to a suitable inner product. 

\begin{definition*}[Inner product on tensor space] Let $E$ be a $\Bbbk$-vector space with an inner product $g$. Note that in the real case this defines an isomorphism $E\cong E^*$, and in the Hermitian case the corresponding antilinear identification with $E^*$. Therefore we immediately obtain an inner product $g^*$ on the dual space $E^*$ as well as the pullback of $g$ along this identification. Moreover, suppose that we have a second $\Bbbk$-inner product space $(E',g')$. Then, we obtain an inner product $g g'$ on the tensor space $E\otimes E'$. Indeed, we put $g g'(e\otimes e',f\otimes f'):=g(e,f)g'(e',f')$ and extend this definition to all of $E\otimes E'$ by sesqui-/bilinearity. In particular, if we start with an inner product $g$ on $E$ we obtain an inner product on the space $\bigotimes^{k,l}E=\bigotimes^k E\otimes \bigotimes^l E^*$ of $(k,l)$-tensors. For convenience, we also denote this inner product by $g$. Of course, this restricts to the space $\bigwedge^k E^*$ of $k$-forms.\end{definition*}

Let us now note that we have an inner product $\bar h_p$ on $HP_p$, i.e. the pullback of $g_{\pis(p)}$ along the isomorphism $HP_p\ra{T_p\pis}T_{\pis(p)}M$, and let us suppose that we have an inner product $k_W$ on $W$. By the above reasoning we have a natural inner product $\bar h_p$ on $\bigwedge^k HP_p^*$ and therefore a pointwise inner product $\bar h_p k_W$ on the space $\bigwedge^k HP_p^*\otimes W$ of $W$-valued horizontal $k$-forms. We write $\bar h k_W$ for the induced map $$\Omega_\hor^k(P,W)\times \Omega_\hor^k(P,W)\ra{\bar h k_W}C^\infty(P).$$

Finally, suppose that $k_W$ is invariant for the $G$-action on $W$. Then we have that, for $\alpha,\beta\in\Omega_\hor^k(P,W)^G$, $\bar h k_W(\alpha,\beta)$ is constant on the fibres of $P\ra{\pis}M$. We conclude that, in this case, we obtain a map\label{vreselijkip}
$$\Omega_\hor^k(P,W)^G\times\Omega_\hor^k(P,W)^G\ra{\bar h k_W}C^\infty(M).$$
Now, we can interpret the covariant codifferential as an adjoint. 

Let us define the \emph{projected support}\index{projected support} of a map $F\ra{\psi}V$ of fibre bundles over $M$, where $V\ra{}M$ is a vector bundle and $F\ra{\pi}M$ is an arbitrary fibre bundle, as the topological closure of the set $\{\pi(f)\in M\; | \; f\in F\txt{and}\psi(f)\neq 0\}$.
\begin{claim}\label{cl:codifadj}Suppose we have a global orientation $\mu$ on $M$, that is compatible with the metric. Let $U\subset M$ be open with compact closure, and let $\alpha\in \Omega_\hor^k(P,W)^G$ and $\beta\in\Omega_\hor^{k+1}(P,W)^G$ have projected supports in $U$. Then
$$\int_U\bar h k_W(d^\omega \alpha,\beta)\mu=\int_U\bar h k_W(\alpha,\delta^\omega \beta)\mu.
$$
\end{claim}
This statement is just a technical generalisation of the well-known result for the case in which $W=\Bbbk$ and $P=M$. The proof can be found in \cite{blegau}.

\begin{remark}
We have seen that we can construct a product connection when we are dealing with fibre bundles. Moreover, we know that $PB_M$ has finite products. The question is whether the product of two principal connections is again a principal connection. Obviously, the answer is yes. \cite{blegau} This connection is very important in gauge theory. Physicists often call it the \emph{spliced connection}\index{spliced connection}, in line with their interesting terminology for the product of two principal bundles.
\end{remark}

\subsection{Induced Connections}
\label{sec:indcon}
We have seen that, given a principal $G$-bundle $P\ra{\pis}M$ and an action $G\times S\ra{\lambda} S$ of $G$ on a manifold $S$, we can construct an associated fibre bundle $P[\lambda]$ with fibre $S$. Suppose that $P$ is also equipped with a principal connection $TP\ra{P_V=\zeta\circ\omega} VP$. The question is whether this induces a connection on the associated bundle in some natural way. This is indeed the case and the correspondence is particularly simple. It will be useful especially when we are dealing with the associated bundles that are part of an equivalence, as discussed in section \ref{vbequivalence}. Before we construct the induced connection, however, let us stop for a moment to think about what precisely a connection on the associated bundle would have to be.

Remember that we defined the associated bundle $P[\lambda]$ as the quotient $P\times S/G$, where $G$ acts on $P\times S$ on the right by $P\times S\times G\ra{\rho\times\lambda^{-1}}P\times S$. By this notation I mean $(p,s,g)\mapsto (\rho(p,g),\lambda(g^{-1},s))$. It is useful to describe tangent vectors to $P[\lambda]$ by representatives: every tangent vector at $[p,s]$ can be written as $Tq(\xi_p,\eta_s)$ for some $(\xi_p,\eta_s)\in T_pP\times T_sS$.

We construct the induced connection by the formula
$$P_V[\lambda](Tq(\xi_p,\eta_s)):=Tq(P_V(\xi_p),\eta_s).$$
This is well-defined. Indeed, the principal-connection condition says that $P_V$ commutes with the tangent lift of the principal right action, and on vectors tangent to the $G$-orbits it is the identity on the vertical part. Hence changing the representative of $Tq(\xi_p,\eta_s)$ changes $Tq(P_V(\xi_p),\eta_s)$ by a tangent vector to the same quotient orbit. The formula is fibrewise linear and idempotent. Its image is vertical because $T\pis(P_V(\xi_p))=0$, and it restricts to the identity on vertical tangent vectors, since then $\xi_p$ may be chosen vertical and $P_V(\xi_p)=\xi_p$. Thus $P_V[\lambda]$ is a connection on the associated bundle.

The parallel transport of this connection along a curve $c$ in $M$ is in a particularly simple correspondence with that of the original principal connection. That is, they are related precisely by the quotient map:
\begin{center}
\begin{tikzcd}[column sep=large,row sep=large]
P_{c(0)}\times S \arrow[r,"{T^P_c(-,t)\times \mathrm{id}_S}"] \arrow[d,"q"'] \arrow[dr,phantom,"\com",description] & P_{c(t)}\times S \arrow[d,"q"]\\
P[\lambda]_{c(0)} \arrow[r,dashed,"{T^{P[\lambda]}_c(-,t)}"'] & P[\lambda]_{c(t)}.
\end{tikzcd}
\end{center}
if we use the notation $T^F_c$ for the parallel transport along $c$ on a fibre bundle $F\ra{}M$.
This follows almost immediately from the uniqueness of solutions of ODEs. That is, after we observe that
\begin{align*}P_V[\lambda](\frac d{dt} q(T^P_c(p,t),s))&=P_V[\lambda](Tq(\frac d {dt}T^P_c(p,t),0_s))\\
&=Tq(P_V\times\mathrm{id}_{TS}(\frac d{dt}T^P_c(p,t),0_s))\\
&=Tq(P_V(\frac d{dt} T^P_c(p,t)),0_s)\\
&=Tq(0_{T^P_c(p,t)},0_s)=0,
\end{align*}
where $0_s$ denotes the zero element in $T_s S$ and similarly for $0_{T^P_c(p,t)}$. From this, we see that $q(T^P_c(p,t),s)$ satisfies the defining initial value problem for $T^{P[\lambda]}_c(q(p,s),t)$. The reader can check that the curvatures are also related by the quotient map, in the sense that $R[\lambda]\circ (Tq\wedge_PTq)=Tq\circ (R\times 0_S)$, where we write $R[\lambda]$ for the curvature of $P_V[\lambda]$ and $0_S$ for the zero section of $TS$. \cite{micnat}

\subsection{Principal Connections Are Complete}
Why are we interested in principal connections in the first place? One major advantage of such a connection is that it is automatically a complete connection: parallel transport is globally defined along any curve. One consequence is that we can construct a holonomy group. 

\begin{theorem}\label{prop:princompl}Principal connections are complete, i.e. for every curve $I\ra{c}M$, the parallel transport is defined for all times: $P_{c(0)}\times I\ra{T_c}P$.
\end{theorem}
\begin{proof}Let $P\ra{\pis}M$ be a principal $G$-bundle that is equipped with a principal connection $\omega$ and let $I\ra{c}M$ be a $C^1$-curve. We want to show that parallel transport is defined globally, so, in the notation of theorem \ref{partrans}, $\Omega_c=P_{c(0)}\times I$.\\
\\
The idea of the proof is that it is easy to find some lift of $c$ to $P$, which we can transform into a horizontal lift by the action of appropriate group elements. One easily derives the differential equation these group elements have to satisfy. Finally, a maximal solution to this differential equation is obtained by extending a local solution using the group multiplication.\\
\\
The precise argument is the following. Using local trivialisations, it is easy to construct a curve $I\ra{d}P$, such that $\pis\circ d=c$ and $d(0)=p$, if $p\in P_{c(0)}$. We look for a smooth curve $I\ra{g} G$, such that $c_\hor(t):=\rho(d_t,g_t)$ is a horizontal lift of $c$: $\pis\circ c_\hor=c$ and $c_\hor'(t)\in HP$, for all $t\in I$. (Here we write $d_t$ for $d(t)$ and similarly for $g_t$.)

We note that the principal right action respects the fibres of $P$, so the first demand, that $\pis\circ c_\hor=c$, is immediately satisfied for any choice of $t\mapsto g_t$. We verify what conditions the second demand imposes on $t\mapsto g_t$. In terms of $\omega$ it reads
$$0\exeq \omega\left(\frac{d}{dt}\rho(d_t,g_t)\right).$$
Now, by the Leibniz rule, $\frac d{dt} \rho(d_t,g_t)=T\rho_{g_t} \dot d_t+T\rho^{d_t} \dot g_t$, where, as before, $\rho_{h}(p):=\rho(p,h)=:\rho^p(h)$. Now, by claim \ref{claim:equiv}, we know that $\omega(T\rho_{g_t} \dot d_t)=\mathrm{Ad}(g_t^{-1})(\omega(\dot d_t))$. Moreover, since claim \ref{claim:repro} tells us that $\omega$ reproduces the generators of the infinitesimal action, we have that $\omega(T\rho^{d_t}(\dot g_t))=TL_{g_t^{-1}} \dot g_t$, where $G\ra{L_h}G$ denotes the left group multiplication by $h$. Remember that $\mathrm{Ad}(h)=T(L_h\circ R_{h^{-1}})$, if $G\ra{R_h}G$ denotes the right multiplication by $h$. Therefore, we find that the second condition reduces to the demand that
$$\dot g_t=-TR_{g_t}\omega(\dot d_t).\txt{(*)}$$
I claim that this differential equation admits a unique global solution $I\ra{g}G$, which is the integral curve of a (time-dependent) right-invariant vector field $X$:
\begin{center}
\scriptsize
\begin{tikzcd}[column sep=small,row sep=large]
G\times I \arrow[r,"X"] & T(G\times I)\cong TG\times I\times \R \\
(h,s) \arrow[r,mapsto] & X_s(h):=(-TR_h \omega (\dot d_s),d/dt|_s)\in T_h G\times T_s I\cong T_h G\times \R.
\end{tikzcd}
\normalsize
\end{center}
(Of course, by $d/dt$ I mean the canonical unit vector in $T_sI\cong \R$. Moreover, right-invariance is immediate: $X_s(R_a (h))=X_s(ha)=-TR_{ha} \omega(\dot d_s)=TR_a(-TR_h\omega(\dot d_s))=TR_a X_s(h)$.)
Clearly, solutions of $(*)$ correspond precisely with integral curves of $X$. The assertion follows from the basic fact that right-invariant vector fields on a Lie group are complete.
\end{proof}

\subsection{The Ambrose-Singer Theorem}
In the case of a principal connection, the interpretation of the curvature as infinitesimal holonomy that we already observed in claim \ref{claim:curvinf} takes a particularly simple form. This result is the famous theorem by Ambrose and Singer. \cite{kobfou}

\begin{theorem}[Ambrose-Singer]\index{Ambrose-Singer theorem} Let $P\ra{\pis}M$ be a principal bundle with a principal connection $\omega$. Then its holonomy group at a point $p_0\in P$ is a Lie group and its Lie algebra is isomorphic to the subalgebra of $\g$ spanned by the elements of the form $\Omega^\omega_p(X,Y)$, where $p$ runs over all points that can be connected to $p_0$ by a horizontal curve and $X,Y$ run over all elements of $T_pP$.
\end{theorem}

\section{Connections on Vector Bundles}
Just as in the case of principal bundles, where we singled out certain connections that are compatible with the group action, we have certain preferred connections on vector bundles: those that are compatible with the linear structure on each fibre. In fact, we will see that, in some sense, these are a special case of the principal connections we have already seen. What does this compatibility mean concretely?

Let $V\ra{\pi}M$ be a vector bundle with a connection $P_V$. In section \ref{geocon}, we have defined $TV\ra{P_V}VV\subset TV$ to be a morphism of vector bundles over $V$, since this was the only vector bundle structure that was at hand. Now, however, $TV$ also has a second vector bundle structure, namely $TV\ra{T\pi}TM$. We would like the vertical part of a tangent vector, after identifying $VV\cong V\times_M V$, to respect the linear structure on the fibres of $TV\ra{T\pi}TM$ as well. When this is the case, we call $P_V$ an \emph{affine connection}\index{affine connection}.

Note that, in the light of equation \ref{eq:chri}, the affinity of the connection is equivalent to the linearity over the ring $C^\infty(M)$ of the Christoffel forms in the second argument: $\Gamma^\alpha(X,f\cdot v)=f\cdot\Gamma^\alpha(X,v)$. This means that we can view $X\mapsto\Gamma^\alpha(X,-)$ as an $\mathrm{End}(\Bbbk^k)$-valued $1$-form $\Gamma^\alpha\in\Omega^1(U_\alpha,\mathrm{End}(\Bbbk^k))$.

This means that the equation that defines parallel transport will be a linear one, whose solution, obviously, depends on the initial values in a linear way. This means that the parallel transport along a curve $c$ over time $t$ defines a (fibrewise linear) map $V_{c(0)}\ra{T_c( - ,t)}V_{c(t)}$. We will denote this map by $T_{c,t}$. (Compare this with the $G$-equivariance of the parallel transport on a $G$-principal bundle with an equivariant connection.)

\subsection{Covariant Derivatives}
There is hardly any mathematical device that is more recurrent in the theory of general relativity than the so-called covariant derivative. It arises wherever an ordinary derivative would appear in non-relativistic physics. We will see that this derivative is closely related to the notions we discussed in the previous paragraphs: connections, curvature, horizontal lifts, Christoffel forms and parallel transport. It is defined as follows.

Let us take a vector field $M\ra{X}TM$ and a section $M\ra{\sigma}V$. We can canonically identify $VV_v=T_v V_{\pi(v)}\cong V_{\pi(v)}$ for every $v\in V$. On the level of vector bundles over $M$ this becomes an isomorphism
\begin{center}
\begin{tikzcd}[column sep=large,row sep=large]
V\times_M V \arrow[r,"\Phi"] & VV\\
(u_m,v_m) \arrow[r,mapsto] & \frac d {dt}|_{t=0} (u_m+t v_m).
\end{tikzcd}
\end{center}
This means that the effect of the connection $P_V$ is entirely captured in the map $K:=\pi_2\circ \Phi^{-1}\circ P_V$, known as the \emph{connector}\index{connector} of the connection. For an affine connection, $K$ is a map of vector bundles (i.e. respects the linear structure on the fibres) in the following two senses:
\begin{center}
\begin{tikzcd}[column sep=large,row sep=large]
TV \arrow[r,"K"] \arrow[d,"{\pi_V}"'] & V \arrow[d,"\pi"] & & TV \arrow[r,"K"] \arrow[d,"{T\pi}"'] & V \arrow[d,"\pi"]\\
V \arrow[r,"\pi"'] & M & & TM \arrow[r,"{\pi_M}"'] & M.
\end{tikzcd}
\end{center}
The first linearity follows from the fact that $P_V$ is a vector-bundle projection over $V$; the second is precisely the affinity condition.

The \emph{covariant derivative}\index{covariant derivative} $\nabla_X \sigma (m):=K\circ T\sigma\circ X (m)$ of $\sigma$ along $X$ is then an element of $V_{\pi(\sigma(m))}=V_m$. $\nabla_X \sigma$ is simply a section $\nabla_X\sigma\in\Gamma(V)$. This means that $\nabla_X$ is an operator $\Gamma(V)\ra{\nabla_X}\Gamma(V)$. In particular, we can iterate it, to obtain higher covariant derivatives. Moreover, it has the following properties.

\begin{claim}
Let $X,Y\in \mathcal{X}(M)$, $\sigma,\tau \in\Gamma(V)$, $f\in C^\infty (M)$ and $m\in M$. Then the covariant derivative has the following properties:
\begin{enumerate}
\item $\nabla_X\sigma(m)$ only depends on the value of $X$ at $m$ and on that of $\sigma$ on a curve $]-\epsilon,\epsilon[\ra{c}M$ with $c(0)=m$ and $c'(0)=X(m)$.\\
\item $\nabla_{f\cdot X+Y}\sigma=f\cdot \nabla_X\sigma+\nabla_Y \sigma$\\
\item $\nabla_X (\sigma+\tau)=\nabla_X\sigma+\nabla_X\tau$\\
\item $\nabla_X(f\cdot \sigma)=X(f)\cdot \sigma+f\cdot \nabla_X \sigma$.
\end{enumerate}
\end{claim}
\begin{proof}Property \textit{1.} follows directly from the definition $\nabla_X \sigma(m)=K(T_m\sigma(X(m))):=K\circ\frac{d}{dt}|_{t=0}(\sigma\circ c)$. Property \textit{2.} comes from the fact that $K$ is a $\pi_V$-$\pi$-fibrewise linear map, while property \textit{3.} comes from the $T\pi$-$\pi$-fibrewise linearity of $K$ (see diagrams above).

For \textit{4.}, choose a local trivialisation $(U,\phi)$ of $V$ around $m$ and write $\sigma_e=\pi_2\circ\phi\circ\sigma$. In these coordinates,
$$T(f\sigma)(X)=T\phi^{-1}\bigl(X, X(f)\sigma_e+f\,T\sigma_e(X)\bigr).$$
The local expression of the connector is
$$\pi_2\circ\phi\circ K\circ T\phi^{-1}(X,\eta)=\eta+\Gamma^\alpha(X,v),$$
for a tangent vector based at $(m,v)\in U\times\Bbbk^k$. Since the connection is affine, $\Gamma^\alpha(X,fv)=f\Gamma^\alpha(X,v)$. Therefore
\begin{align*}
\pi_2\circ\phi(\nabla_X(f\sigma))
&=X(f)\sigma_e+f\,T\sigma_e(X)+\Gamma^\alpha(X,f\sigma_e)\\
&=X(f)\sigma_e+f\bigl(T\sigma_e(X)+\Gamma^\alpha(X,\sigma_e)\bigr),
\end{align*}
which is the coordinate expression for $\nabla_X(f\sigma)=X(f)\sigma+f\nabla_X\sigma$.
\end{proof}
One can also define a map $\mathcal{X}(M)\times \Gamma(V)\ra{\nabla}\Gamma(V)$ satisfying these properties to be a covariant derivative. From this one can construct a connection\footnote{Corollary \ref{contriv} gives a formula for the Christoffel forms in terms of the covariant derivative. Equation \ref{eq:chri} then gives a formula for what connection $P_V$ should be.}. This is easily (see for example \cite{banrie}) seen to be an equivalent characterisation of an affine connection.
\begin{corollary}
\label{contriv}
Let $(U_\alpha,\phi_\alpha)$ be a local trivialisation for $V$ and write $\Gamma^\alpha$ for the associated Christoffel $1$-form. Then, if we write $(e_i)$ for the local frame corresponding to the trivialisation and $\sigma_e=(\sigma^i)$ for the components of a section $\sigma$ with respect to this frame,
$$(\nabla_X\sigma)_e=T\sigma_e(X)+\Gamma^\alpha(X)\sigma_e.$$
\end{corollary}
\begin{proof}We first deal with a simple case. Denote by $\Gamma^\alpha(X)_i^j$ the matrix coefficients of $\Gamma^\alpha(X)$ with respect to the frame $(e_i)$ and write $\hat e_i$ for the $i$-th basis vector of $\Bbbk^k$. Furthermore write $T\Bbbk^k\cong \Bbbk^k\times \Bbbk^k\ra{\pi_{\bullet\times\Bbbk^k}}\Bbbk^k$ for the projection $(x,y)\mapsto y$, where the second $\Bbbk^k$ denotes the tangent space of the first. Then\\
\begin{align*}
(\nabla_X e_i)_e&=\pi_{\bullet\times\Bbbk^k}\circ T\phi_\alpha\circ P_V(T e_i (X))\\
&=\pi_{\bullet\times\Bbbk^k}\circ T\phi_\alpha\circ P_V(T\phi_\alpha ^{-1}(X,0_{\hat{e_i}}))\\
&=\Gamma^\alpha(X)\hat{e_i}.
\end{align*}
Hence $\nabla_X e_i=\sum_j \Gamma^\alpha(X)^j_i e_j$.

This implies that for a general section $\sigma\in\Gamma(V)$,
\begin{align*}\nabla_X \sigma&=\nabla_X \sum_{i=1}^k \sigma^i e_i\\
&=\sum_{i=1}^k T\sigma^i (X)e_i+\sum_{i=1}^k \sigma^i \nabla_X e_i\\
&=\sum_{i=1}^k T\sigma^i (X)e_i+\sum_{i,j=1}^k \sigma^i \Gamma^\alpha(X)^j_i e_j\\
&=\sum_{i=1}^k (T\sigma^i (X)+ (\Gamma^\alpha(X)\sigma_e)^i)e_i.
\end{align*}
\end{proof}
This result shows that $\nabla_X\sigma(m)$ depends only on the value of $X$ at $m$ and that of $\sigma$ on an arbitrarily small path $]-\epsilon,\epsilon[\ra{c}M$, with $c(0)=m$ and $c'(0)=X(m)$.
\begin{claim}\label{cl:covdivpartrans}Let $]-\epsilon,\epsilon[\ra{c} M$ be a $C^1$-curve with $c(0)=a$ and $c'(0)=X(a)$. Then, in the canonical identification $T_{\sigma(a)}V_a\cong V_a$,  $\nabla_X \sigma (a)=\frac{d}{dt}|_{t=0} T_{c,t}^{-1}(\sigma(c(t)))$.\end{claim}
\begin{proof}This is a simple computation. \begin{align*}\nabla_X\sigma(a)&=K ( T_a \sigma (X(a)))\\
&=K\frac d{dt}|_{t=0}\sigma(c(t))\\
&=K\frac d{dt}|_{t=0}T_{c,t}\circ T_{c,t}^{-1}(\sigma(c(t)))\\
&=K\frac d{dt}|_{t=0} T_{c,t} (\sigma(c(0)))+K T_{c,0} \frac d{dt}|_{t=0} T_{c,t}^{-1}(\sigma(c(t)))\\
&=K\frac d{dt}|_{t=0} T_{c,t} (\sigma(a))+K \frac d{dt}|_{t=0} T_{c,t}^{-1}(\sigma(c(t))).\\
\end{align*}
Now, the first term vanishes, since $\frac d{dt}|_{t=0} T_{c,t} (\sigma(a))$ is horizontal, by definition of parallel transport. Moreover, since $T_{c,t}^{-1}(\sigma(c(t)))\in V_a$ for all $t$, $\frac d{dt}|_{t=0} T_{c,t}^{-1}(\sigma(c(t)))$ is a vertical vector. The conclusion follows.
\end{proof}
Finally, I would like to introduce one more notation that is ubiquitous in physics literature. Let $]-\epsilon,\epsilon[\ra{c}M$ denote a smooth curve. Then, we can consider the pullback connection $c^*P_V$ if $P_V$ denotes a given affine connection on $V\ra{\pi}M$. It is conventional to write $D/dt$ for its covariant derivative. Explicitly, $Ds/dt=K\circ s'$, if $]-\epsilon,\epsilon[\ra{s}V$ is a smooth map such that $\pi\circ s=c$ (i.e. $s$ is a section of the pullback bundle $c^*\pi$).

\subsection{Curvature of Affine Connections}
The definition of curvature of corollary \ref{curveasy} takes a particularly easy form if we are dealing with an affine connection. Pointwise, $R$ corresponds to maps $$T_{\pi(v)}M\wedge T_{\pi(v)} M \ra{R_v} VV_v=T_v V_{\pi(v)}\cong V_{\pi(v)},$$ depending smoothly on the basepoint $v\in V$. In terms of the connector, this becomes $R_v(X,Y)=-K([X_\hor,Y_\hor])$. In this form, the curvature is closely related to the second covariant derivative.

\begin{claim}
Suppose $V$ is a vector bundle that is equipped with an affine connection. Let $m\in M$, $\sigma\in \Gamma(V)$ and $X,Y\in \mathcal{X}(M)$. Then the following formula holds for the curvature:
\begin{equation}\label{eq:curvcov}R_{\sigma(m)}(X,Y)=\nabla_X\nabla_Y \sigma(m) -\nabla_Y\nabla_X \sigma(m)-\nabla_{[X,Y]}\sigma(m).\end{equation}
\end{claim}
\begin{proof}We prove the formula in a local frame, using the formula from corollary \ref{contriv}. With respect to a local frame $(e_i)$ with Christoffel $1$-form $\Gamma^\alpha$, we have
\begin{align*}
(\nabla_X\nabla_Y\sigma)_e&=\Li_X(\Li_Y(\sigma_e))+\Li_X(\Gamma^\alpha(Y))\sigma_e+\Gamma^\alpha(Y)\Li_X(\sigma_e)\\
&+\Gamma^\alpha(X)\Li_Y(\sigma_e)+\Gamma^\alpha(X)\Gamma^\alpha(Y)\sigma_e.
\end{align*}
and since $\Li_X(\Li_Y(\sigma_e))-\Li_Y(\Li_X(\sigma_e))=\Li_{[X,Y]}(\sigma_e)$, we find that the right-hand side of equation \ref{eq:curvcov} equals
\begin{align*}&\Big(\Li_X(\Gamma^\alpha(Y))-\Li_Y(\Gamma^\alpha(X))+\Gamma^\alpha(X)\Gamma^\alpha(Y)-\Gamma^\alpha(Y)\Gamma^\alpha(X)-\Gamma^\alpha([X,Y])\Big)\sigma_e.
\end{align*}
On the other hand, in local coordinates, for $v\in \Bbbk^k$,
\begin{align*}(-R_v(X,Y))_e&=\pi_2\circ\psi_\alpha\bigl(K([X_\hor,Y_\hor](v))\bigr)\\
&=\pi_{\bullet\times\Bbbk^k}( T\psi_\alpha([X_\hor,Y_\hor](v)-[X,Y]_\hor(v)))\\
&=\pi_{\bullet\times\Bbbk^k}( [T\psi_\alpha(X_\hor),T\psi_\alpha(Y_\hor)](v)-T\psi_\alpha([X,Y]_\hor(v)))\\
&=\pi_{\bullet\times\Bbbk^k}(T_v(T\psi_\alpha(Y_\hor ))(T\psi_\alpha(X_\hor(v) ))-\\
&(T_v(T\psi_\alpha(X_\hor ))(T\psi_\alpha(Y_\hor (v))) -T\psi_\alpha([X,Y]_\hor(v))),\txt{(*)}
\end{align*}
while
\begin{align*} T\psi_\alpha(X_\hor )(U,W)&= (TX(U)\; ,\;-U(\Gamma^\alpha(X))v-\Gamma^\alpha(X)W),
\end{align*}
which implies that
\begin{align*}
T_v(T\psi_\alpha(Y_\hor ))(T\psi_\alpha(X_\hor ))&=(TY(X)\; , \; -X(\Gamma^\alpha(Y))v+\Gamma^\alpha(Y)\Gamma^\alpha(X)v).
\end{align*}
Moreover, $$\pi_{\bullet\times\Bbbk^k}(T\psi_\alpha[X,Y]_\hor(v))=-\Gamma^\alpha([X,Y])v.$$
We therefore find from $(*)$ that
\begin{align*}(R_v(X,Y))_e&=(\Li_X(\Gamma^\alpha(Y))-\Li_Y(\Gamma^\alpha(X)))v+(\Gamma^\alpha(X)\Gamma^\alpha(Y)-\Gamma^\alpha(Y)\Gamma^\alpha(X)\\
&-\Gamma^\alpha([X,Y]))v.
\end{align*}
The equality follows.
\end{proof}
Using this formula, we can interpret curvature as a measure of the extent to which $X\mapsto \nabla_X$ fails to be a Lie algebra homomorphism\footnote{Here, we understand this to be a map from the Lie algebra of vector fields on $M$ to the commutator Lie algebra of the associative algebra of operators $\Gamma(V)\ra{}\Gamma(V)$ under composition.}.
Moreover, it shows us that $R_v$ depends linearly on $v$, so we can equivalently view $R$ as a morphism of vector bundles over $M$:
$$TM\wedge TM\ra{R}\mathrm{End}(V).$$
This fact can equivalently be expressed as $R\in \Omega^2(M,\mathrm{End}(V)\ra{}M)$. 

In the preceding proof, we have found another useful expression for the curvature. If we agree to write $d\Gamma^\alpha(X,Y):=\Li_X(\Gamma^\alpha(Y))-\Li_Y(\Gamma^\alpha(X))-\Gamma^\alpha([X,Y])$ and $\Gamma^\alpha\wedge\Gamma^\alpha(X,Y):=\Gamma^\alpha(X)\Gamma^\alpha(Y)-\Gamma^\alpha(Y)\Gamma^\alpha(X)$, then the following holds.
\begin{corollary}
In a local frame $(e_i)$, the curvature can be expressed as follows:
$$R(X,Y)_e=\Big(d\Gamma^\alpha (X,Y) +\Gamma^\alpha\wedge\Gamma^\alpha(X,Y)\Big).$$
\end{corollary}
This result closely resembles that of theorem \ref{maurer} for principal connections. Remember, though, that this result only holds locally, as $\Gamma^\alpha$ is only defined in a local frame. Nevertheless, the resemblance is not a coincidence. As we will see, it is a consequence of the fact that an affine connection is induced from a principal connection on the frame bundle.

\subsection{Equivalent Notions of a Connection}
In section \ref{vbequivalence}, we have seen that there is an equivalence of categories between $\GL(\Bbbk^k)$-principal bundles and $k$-dimensional vector bundles over a manifold $M$, in the form of the associated vector bundle construction and the frame bundle construction. The question is whether one can, given a suitable connection on a vector bundle, construct one on the frame bundle and vice versa. The answer turns out to be affirmative.
\begin{claim}\label{afconind}
Suppose $V\ra{\pi}M$ is a $k$-dimensional vector bundle. Then, there is a canonical one-to-one correspondence between affine connections on $V$ and principal connections on the frame bundle $F^\lambda(V)$, where one direction of the correspondence is precisely the induced connection (if we use the natural identification $F^\lambda(V)[\lambda]\cong V$).
\end{claim}
\begin{proof}Suppose we are given an affine connection $P_V$ on $V$. We take $F^{T\lambda}(P_V):=P_V\circ -$ as a connection on $F^{\lambda}(V)$. (This notation is of course suggestive, as it resembles that of the frame bundle functor $F^{T\lambda}$ for the effective action $T\lambda$. However, $P_V$ is not invertible on the fibre, so technically $F^{T\lambda}$ is not defined on it.) Clearly, this yields an equivariant connection precisely since $P_V$ is affine. Then, if we take the induced connection on the associated bundle, we obtain
\begin{align*}(F^{T\lambda}(P_V))[\lambda]([(f,w),(u,v)])&=[F^{T\lambda}(P_V)(f,w),(u,v)]\\
&=[P_V\circ (f,w),(u,v)].
\end{align*}
We conclude that $(F^{T\lambda}(P_V))[\lambda]\circ Tq((f,w),-)=Tq(P_V\circ(f,w), -)$, i.e. $P_V$ and $(F^{T\lambda}(P_V))[\lambda]$ are related by the derivative of the natural isomorphism
\begin{center}
\begin{tikzcd}[column sep=large]
V \arrow[r] & F^\lambda(V)[\lambda].
\end{tikzcd}
\end{center}
(Compare this with equation $(*)$ on page \pageref{eq:asfrid}.) We conclude that $F^{T\lambda}$ is injective.

Conversely, suppose that $P_V'$ is an equivariant connection on $F^\lambda(V)$. Then, by $\GL(\Bbbk^k)$-equivariance of $P_V'$, the induced connection on the associated vector bundle is affine. Suppose that it corresponds to a connection $P_V$ on $V$ through the identification $F^\lambda(V)[\lambda]\cong V$. This means that
$$Tq(P_V'(f,w),-)=P_V'[\lambda]\circ Tq((f,w),-)\exeq Tq(P_V\circ (f,w),-).$$
This implies that $(P_V'(f,w)\cdot g,g^{-1}\cdot -)=(P_V\circ (f,w),-)$, for some $g\in T\GL(\Bbbk^k)$. Since this equality holds for every second component, $g$ acts trivially on $T\Bbbk^k$. By faithfulness of the tangent representation $T\lambda$, $g=(\mathrm{id}_{\Bbbk^k},0)$. Therefore $P_V'(f,w)=P_V\circ (f,w)=F^{T\lambda}(P_V)(f,w)$, so $F^{T\lambda}$ is surjective.
\end{proof}
\begin{remark}
Perhaps this result would be better placed in section \ref{sec:indcon}, in a more general form. Certainly, if one were to formulate an appropriate notion of a `$(G,\lambda)$-connection' in a fibre bundle with structure group $(G,\lambda)$, one would expect that the same construction would work to establish a one-to-one correspondence between such connections on the $(G,\lambda)$-fibre bundle and principal connections on the corresponding generalised frame bundle. (I would opt to define such connections as ones whose vertical projections are maps of $(TG,T\lambda)$-fibre bundles over $TM$). One could even go a step further and define a category of $(G,\lambda)$-fibre bundles with $(G,\lambda)$-connections, in which a natural definition of arrows would be maps of $(G,\lambda)$-fibre bundles such that the two connections involved would be related by pullback along such a fibre bundle, and try to extend the equivalence of categories of theorem \ref{fbcon} to one between such categories. However, such a result would probably have few applications and would be established purely for aesthetic purposes. Therefore, it is omitted. \end{remark}

We see that we can interpret the theory of vector bundles with affine connections as a special case of that of principal bundles with equivariant connections. In particular, we have the following.
\begin{corollary}
Affine connections are complete.
\end{corollary}
\begin{proof}
In the light of claim \ref{afconind}, this follows immediately from the fact that the induced parallel transport on the associated bundle is $q$-related to that on the original principal bundle, which we already know to be defined along all curves by theorem \ref{prop:princompl}.

Of course, we could also write out the differential equation defining parallel transport in local coordinates, which would be linear in the case of an affine connection. The global existence theorem for linear differential equations would do the rest.
\end{proof}

To complete the equivalence between affine connections on a vector bundle and equivariant connections on its frame bundle, we state the following result, due to Kol\'a\v{r} et al. \cite{micnat}, relating the two derivatives. We have seen in remark \ref{identifi} that for a principal $G$-bundle $P\ra{\pis}M$ and a vector space $W$ with a linear $G$-action $\lambda$ we have an identification $\Omega_\hor^k(P,W)^G\stackrel{\Psi_\lambda}{\ra{\simeq}}\Omega^k(M,P[\lambda])$. In particular, we can identify an equivariant function $P\ra{\Phi}W$ with a section $\sigma$ of the associated bundle $P[\lambda]$. In the presence of a connection, we can take the exterior covariant derivative along the horizontal lift of a vector field $X\in\mathcal{X}(M)$ on one side to obtain another equivariant function $P\ra{d^\omega \Phi(X_\hor)}W$ and on the other we can take the covariant derivative (of the induced connection) along the vector field to obtain another section: $\nabla_X\sigma\in\Gamma(P[\lambda])$. One wonders if the two are related. Indeed, they turn out to coincide in the identification $\Psi_\lambda$.

If we extend $\nabla$ from $\Omega^0(M,P[\lambda])=\Gamma(P[\lambda])$ to $\Omega^k(M,P[\lambda])$, by setting, for $\alpha\in \Omega^k(M,P[\lambda])$ and $X_0,\ldots, X_{k}\in \mathcal{X}(M)$:
$$(d_\nabla\alpha)(X_0,\ldots,X_{k}):=\sum_{i=0}^k (-1)^i\nabla_{X_i}(\alpha(X_0,\ldots, \hat{X_i},\ldots X_k))$$
$$+\sum_{0\leq i<j\leq k}(-1)^{i+j}\alpha([X_i,X_j],X_0,\ldots,\hat X_i,\ldots,\hat X_j,\ldots X_k)$$
- this is just the definition of the ordinary differential $d$ with $\nabla_{X_i}$ replacing $\Li_{X_i}$ - then we have the following.
\begin{claim}\index{covariant derivative, correspondence with exterior covariant derivative}$d_\nabla$ acts on $P[\lambda]$-valued differential forms on $M$ as $d^\omega$ does on equivariant horizontal ones on $P$, in the sense that
\begin{center}
\begin{tikzcd}[column sep=large,row sep=large]
\Omega_\hor^k(P,W)^G \arrow[r,"{\Psi_\lambda}","\simeq"'] \arrow[d,"{d^\omega}"'] \arrow[dr,phantom,"\com",description] & \Omega^k(M,P[\lambda]) \arrow[d,"{d_\nabla}"]\\
\Omega_\hor^{k+1}(P,W)^G \arrow[r,"{\Psi_\lambda}","\simeq"'] & \Omega^{k+1}(M,P[\lambda]).
\end{tikzcd}
\end{center}\end{claim}
The proof, which can be found in \cite{micnat}, proceeds as one would expect. One first checks the identity on $\Omega^0$. (This is a short computation.) Next, one extends it to elements of $\Omega^k(M,P[\lambda])$ that can be decomposed as $\sigma\otimes \eta$, with $\sigma\in \Omega^0(M,P[\lambda])$ and $\eta\in \Omega^k(M)$. Finally, one obtains the general statement by working locally in a frame of $P[\lambda]$: every vector-valued form is a finite sum of such terms, and the two sides are local and $\Bbbk$-linear.

\subsection{Categorical Constructions on Affine Connections}\label{catconstrafcon}
It turns out to be possible to transfer a connection to quite a few categorical constructions we can build out of a vector bundle. We can, for instance, give a sensible definition of a \emph{dual connection}\index{dual connection} on the dual bundle of a vector bundle with a connection. The same applies for the tensor product. 

Suppose $V\ra{\pi}M$ is a $\Bbbk^k$-vector bundle that is equipped with an affine connection $P_V$. We have seen in claim \ref{afconind} that this canonically induces a connection $F^{T\lambda}(P_V)$ on the frame bundle $F^\lambda(V)$. Note that we have a canonical left action $\lambda^*$ of $\GL(\Bbbk^k)$ on $\Bbbk^k{}^*$ by (inverse) precomposition $\GL(\Bbbk^k)\ni g\mapsto -\circ g^{-1}\in \GL(\Bbbk^k{}^*)$. We are in the familiar setting in which we construct an associated vector bundle. Note that this is just the dual vector bundle $V^*\ra{\pi^*}M$. This gives us a canonical way to transfer the connection on $V$ to $V^*$: $F^{T\lambda}(P_V)$ induces an affine connection on the associated bundle $V^*\to M$:\; $F^{T\lambda}(P_V)[\lambda^*]$. It is easily verified, using the characterisation of the covariant derivative given in claim \ref{cl:covdivpartrans} in combination with the Leibniz rule, that this yields the following relation between the covariant derivatives:
$$\ip{\nabla^*_X \xi}{\sigma}+\ip{\xi}{\nabla_X\sigma}=X(\ip{\xi}{\sigma}),$$
where $X\in \mathcal{X}(M)$, $\sigma\in \Gamma(\pi)$ and $\xi\in\Gamma(\pi^*)$.

In the case of a complex vector bundle, there is the concept of the complex conjugate bundle. This is easily seen to be an associated bundle as well. Indeed, let $V\ra{\pi}M$ be a vector bundle and describe it as an associated bundle of $P\ra{\pis}M$ corresponding to the canonical action $\lambda$ of $\GL(\C^k)$ on $\C^k$. Then we have an action $\bar\lambda$ of $\GL(\C^k)$ on $\overline{\C^k}$, by $\bar\lambda(g,\bar v):=\overline{\lambda(g,v)}$. This induces a connection on the complex conjugate bundle $\bar{V}\ra{\bar{\pi}}M$. The covariant derivative $\bar\nabla$ of this \emph{complex conjugate connection}\index{complex conjugate connection} acts on sections $\psi\in\Gamma(\bar\pi)$ as one would expect:
$$\bar\nabla_X \psi=\overline{\nabla_X \bar\psi},$$
where $\nabla$ is the covariant derivative of the original connection on $\pi$ (of which we understand $\bar\psi$ to be a section). This construction is important in physics when one encounters spinors.

Similarly, if $V'\ra{\pi'}M$ is a second vector bundle, of dimension, say, $l$ over $\Bbbk$, on which an affine connection $P_V'$ is defined, we have a canonical notion of a \emph{tensor product connection}\index{tensor product connection} on $V\otimes V'\ra{\pi\otimes\pi'}M$. Again, we note that we obtain principal connections $F^{T\lambda}(P_V)$ and $F^{T\lambda'}(P_V')$ on the frame bundles. This defines a principal connection $F^{T\lambda}(P_V)\times_M F^{T\lambda'}(P_V')$ on the spliced principal $\GL(\Bbbk^k)\times\GL(\Bbbk^l)$-bundle  $F^\lambda(V)\times_M F^{\lambda'}(V')\to M$. Again, we can form an associated vector bundle by the tensor product action $\lambda\otimes \lambda'$ of $\GL(\Bbbk^k)\times\GL(\Bbbk^l)$ on $\Bbbk^k\otimes \Bbbk^l$. This vector bundle is, of course, just the tensor product bundle $V\otimes V'\to M$. Since we have realised this bundle as an associated bundle, we obtain an induced connection $(F^{T\lambda}(P_V)\times_M F^{T\lambda'}(P_V'))[\lambda\otimes \lambda']$ on it. In terms of the covariant derivative, this means that,
$${}^\otimes \nabla_X(\sigma\otimes \sigma')=(\nabla_X\sigma)\otimes \sigma'+\sigma\otimes(\nabla'_X \sigma'),$$
if $X\in\mathcal{X}(M)$, $\sigma\in \Gamma(\pi)$ and $\sigma'\in\Gamma(\pi')$ and if we denote the covariant derivative on $V\otimes V'$ by ${}^\otimes \nabla$. Again, this identity is almost immediate if one uses the relation of claim \ref{cl:covdivpartrans} between the covariant derivative and parallel transport. Note that physicists use these induced connections all the time when dealing with tensor bundles, for example in the context of general relativity.

An analogous construction for which we can also define an induced connection is the direct sum (categorical product). We have already seen how to define the product connection for general fibre bundles. In this context, it will be no different. We can also understand the direct sum bundle to be an associated bundle and construct the connection as we did for the tensor product (using the direct sum action). The result on the covariant derivative is as follows:
$${}^\oplus \nabla_X(\sigma\oplus\sigma')=\nabla_X\sigma\oplus\nabla '_X\sigma'$$
and we extend by linearity.

Finally, the construction works equally for $\mathrm{Hom}$-bundles. Obviously, these can also be seen as associated vector bundles, since we have the canonical isomorphism $\mathrm{Hom}(V,V')\cong V^*\otimes V'$.

\section{Manifolds with Connections}
An important context in which one encounters connections is through a (pseudo-)metric on the tangent bundle of a manifold $M$. This induces the so-called Levi-Civita connection on $TM$. This subject, however, will be dealt with in section \ref{sec:riem}. 

The more general class of arbitrary affine connections on a tangent bundle is dealt with first, in the hope that the reader realises which aspects are special to the case of Riemannian geometry and which structure derives from a more general principle. If a manifold is equipped with an affine connection on its tangent bundle, one sometimes simply says a connection is defined on the manifold. These manifolds form a broader class than the (pseudo)Riemannian manifolds. They have important applications. 

For instance, they are used in the Newton-Cartan theory of gravity, which is a beautiful theory of Newtonian gravity, using a geometry very similar to that of the general theory of relativity. In this theory, the particle moves through Newtonian spacetime along geodesics that are determined by a non-metric connection. In this way, one can use the language of manifolds with connections to describe Newtonian as well as Einsteinian gravity and one can give a particularly clear demonstration that the latter reduces to the former in the limit $c\to \infty$. A useful reference on this subject is a text by Cartan himself, \cite{caronm}.\\
\\
A first remark, although it is not of very fundamental importance, is the following. Let us adopt the notation $T^{(k,l)}M:=\bigotimes_M^k TM\otimes \bigotimes^l T^* M$. As we have seen, the affine connection on $TM$ induces one and therefore also a corresponding covariant derivative on these tensor bundles. We have seen that $\nabla_X\sigma(m)$ only depends on the value of $X$ at $m$ (and on the value of $\sigma$ along a compatible curve). Therefore, if $S\in \Gamma(T^{(k,l)}M)$, then we can interpret $\nabla S$ as an element of $\Gamma(T^*M\otimes T^{(k,l)}M)\cong\Gamma(T^{(k,l+1)}M)$, where
$$(\nabla S)(X;\alpha_1,\ldots,\alpha_k,Y_1,\ldots,Y_l)=(\nabla_XS)(\alpha_1,\ldots,\alpha_k,Y_1,\ldots,Y_l).$$
This is a notation that physicists use a lot in general relativity.

\subsection{Geodesics}
The most important new concept that we encounter when specialising from vector bundles with affine connections to manifolds with affine connections is that of \emph{geodesics}\index{geodesic}. The reader might know the concept of a geodesic from Riemannian geometry, where a geodesic might be defined as a curve that locally minimises the path length (or, equivalently, the energy integral) between two points that lie on it. Of course, this definition cannot be formulated in the context of an arbitrary manifold with an affine connection. However, a geodesic can equivalently be defined as a curve whose tangent vectors are related by parallel transport along the curve itself. To be precise, a curve $]-\epsilon,\epsilon[\ra{c}M$ is called a(n) (affinely parametrised) geodesic if
$$\frac{Dc'}{dt}=0.$$
This definition makes equal sense when the connection is not a Levi-Civita connection of some metric.

Equivalently, for a given affine connection on $TM$, one can define the \emph{geodesic spray}\index{geodesic spray} to be the unique horizontal vector field $W$ on $TM$ such that $T\pi(W(v))=v$ for all $v\in TM$. A geodesic is then precisely the projection along $\pi$ of an integral curve of $W$. More generally, we call a vector field $W$ on $TM$ a \emph{spray}\index{spray}, if $T\pi_M\circ W=\mathrm{id}_{TM}$ and, in a chart induced by the double tangent bundle functor, $W(x,y)=(x,y;y,A_x(y))$, where $y\mapsto A_x(y)$ is quadratic. In this terminology, the geodesic spray of a connection is the unique horizontal spray.

\subsection{Torsion}
We have seen that an affine connection on the tangent bundle uniquely determines a geodesic spray. One can ask whether the converse is also true. The answer turns out to be: almost, once we fix a tensor field known as the torsion.

Given a manifold with an affine connection $(M,P_V)$ we define the \emph{torsion}\index{torsion} $T\in\Gamma(T^{(1,2)}M)$ of the connection to be the $(1,2)$-tensor
$$T(X,Y):=\nabla_X Y-\nabla_Y X-[X,Y].$$
The following theorem gives an interpretation of the torsion tensor as the freedom in the choice of a connection, after we have specified a certain geodesic spray. \cite{ambspr}

\begin{theorem}[Ambrose-Palais-Singer]\index{Ambrose-Palais-Singer theorem}\label{ambpalsing}
Let $M$ be a smooth manifold with a spray $W$. Then, for each type $(1,2)$-tensor field $T$ that is anti-symmetric in the last pair of arguments, there exists a unique affine connection on $TM\ra{}M$ such that $W$ is its geodesic spray and such that $T$ is its torsion.
\end{theorem}

In particular, for every affine connection on $M$, we can find a unique torsion-free one that has the same geodesic spray. If we were to divide the affine connections on a manifold $M$ into equivalence classes of those that have the same geodesic spray, we could take this to be a canonical representative of each class. In computations, torsion-free connections are often easiest to work with; therefore, for example in the general theory of relativity, one uses this freedom to simplify calculations.

In light of this theorem, moreover, we see that it was no coincidence that Cartan was able to formulate Newtonian gravity using the geometry of an affine connection\footnote{Indeed, after adjoining the absolute time coordinate in Newton-Cartan spacetime, the Newtonian gravitational force of one point source gives rise to a spray whose spatial part is $A_{(t,x)}(\dot t,\dot x)=-\gamma x\dot t^2/\norm{x}^3$, where $\gamma$ is some scale constant for the gravitational interaction. Similarly, arbitrary Newtonian gravitational fields give spatial acceleration $-\nabla\Phi(t,x)\dot t^2$. In particular, the acceleration is quadratic in the full spacetime velocity.}. However, when Cartan formulated his Newton-Cartan theory in 1923, this result would not be proved for forty years.

\section{(Pseudo)Riemannian Manifolds and the Levi-Civita Connection}\label{sec:riem}
Suppose $M$ is a (pseudo-)Riemannian manifold, that is a manifold equipped with a symmetric, non-degenerate tensor $g\in \Gamma(T^* M\otimes T^* M)$. Then, there is a canonical choice for a connection on $M$ (an affine connection on the tangent bundle). I will not give the standard derivation\footnote{This would involve demanding that the metric have covariant derivative zero and showing that this defines a unique connection for a specified torsion tensor (the canonical choice being torsion zero). My derivation is somewhat less standard, but yields the same connection.} of this connection. The one I give might be somewhat more involved, but is more in line with the rest of this thesis. In particular, it will be more natural from a physicist's point of view, since it starts from an action principle.

The idea is the following. From the metric, we can define a so-called \emph{energy integral}\index{energy integral} along each path $]-\epsilon,\epsilon[\ra{c}M$: $S_g(c):=\frac 1 2 \int_{-\epsilon}^\eps g(c'(t),c'(t))dt$. We will define the so-called \emph{Levi-Civita connection}\index{Levi-Civita connection} associated with $g$ in such a way that its geodesics are precisely the curves for which the energy integral is stationary with respect to all compactly supported variations in the curve. (Note that this yields precisely the straight lines in the case of $M=\R^{p,q}$ with the standard metric.)

\begin{claim}[First variational formula]\index{first variational formula} \label{firvar} Let $]-\eps,\eps[\times]-\delta,\delta[\ra{H}U\subset \R^n$ be a $1$-parameter family of curves in an open set $U$ of $\R^n$ on which we have a (pseudo)-Riemannian metric and write $c:=H(-,0)$. Suppose that the variation has compact support in the sense that $\mathrm{supp}\, (t\mapsto r(t):=\frac{d}{ds}|_{s=0} H(t,s))$ is compact. Then we have
$$\frac{\dau}{\dau s}|_{s=0} S_g(H(-,s))=$$
$$\int_{-\eps}^\eps \left(-T_{c(t)}g(c'(t))(c'(t),r(t))-g_{c(t)}(c''(t),r(t))+\frac 1 2  T_{c(t)}g(r(t))(c'(t),c'(t))\right)  dt.$$
\end{claim}
\begin{proof}Let us write $H_s$ for $\frac{\dau}{\dau s}H$ and similarly for $H_t$. Then, we Taylor expand $H$ with respect to $s$ at $s=0$: $H(t,s)=H(t,0)+s H_s(t,0)+\mathcal{O}(s^2)=H(t,0)+s r(t)+\mathcal{O}(s^2)$. This leads to
\begin{align*}E_g(H(-,s))&=\frac 1 2 \int_{-\eps}^\eps g_{H(t,s)}(H_t(t,s),H_t(t,s))dt\\
&=\frac 1 2\int_{-\eps}^\eps g_{c(t)+sr(t)+\mathcal{O}(s^2)}(c'(t)+sr'(t)+\mathcal{O}(s^2),c'(t)+sr'(t)+\mathcal{O}(s^2))dt\\
&=\frac 1 2 \int_{-\eps}^\eps \left(g_{c(t)}+s T_{c(t)}g(r(t))+\mathcal{O}(s^2)\right)\left(c'(t)+sr'(t)+\mathcal{O}(s^2),\right. \\&\left. c'(t)+sr'(t)+\mathcal{O}(s^2)\right)dt\\
&=\frac 1 2 \int_{-\eps}^\eps g_{c(t)} (c'(t),c'(t))+2s g_{c(t)}(c'(t),r'(t))\\
&+s T_{c(t)} g (r(t))(c'(t),c'(t))dt+\mathcal{O}(s^2)\\
&=E_g(c)+s\left(\int_{-\eps}^\eps g_{c(t)}(c'(t),r'(t))dt+ \frac 1 2 \int_{-\eps}^\eps T_{c(t)}g(r(t))(c'(t),c'(t))dt \right)\\
&+\mathcal{O}(s^2).
\end{align*}
Thus, we find that
\begin{align*}\frac{\dau}{\dau s}|_{s=0} S_g(H(-,s))&=\int_{-\eps}^\eps g_{c(t)}(c'(t),r'(t))dt+ \frac 1 2 \int_{-\eps}^\eps T_{c(t)}g(r(t))(c'(t),c'(t))dt \\
&=g_{c(t)}(c'(t),r(t))|_{t=-\eps}^\eps-\int_{-\eps}^\eps \left(T_{c(t)}g(c'(t))(c'(t),r(t))+g_{c(t)}(c''(t),r(t))\right)dt\\
&+\frac 1 2 \int_{-\eps}^\eps T_{c(t)}g(r(t))(c'(t),c'(t))dt \txt{(integration by parts of the first term)}\\
&=-\int_{-\eps}^\eps \left(T_{c(t)}g(c'(t))(c'(t),r(t))+g_{c(t)}(c''(t),r(t))\right)dt\\
&+\frac 1 2 \int_{-\eps}^\eps T_{c(t)}g(r(t))(c'(t),c'(t))dt. \txt{(note that $r$ had compact support)}
\end{align*}
\end{proof}
\begin{notation}\label{chrissym}Note that this claim implies that the energy integral of a curve is stationary with respect to all variations with compact support if and only if
$$g_{c(t)}(c''(t),Z)=\frac 1 2 T_{c(t)}g(Z)(c'(t),c'(t))-T_{c(t)}g(c'(t))(c'(t),Z)$$
for all $Z\in\R^n$. Let us therefore define (using non-degeneracy of $g_m$) a map $\R^n\times \R^n\ra{\Gamma_m}\R^n$ by
$$g_m(\Gamma_m(X,Y),Z)=\frac 1 2 \left(T_mg(X)(Y,Z)+T_mg(Y)(X,Z)-T_mg(Z)(X,Y)\right).$$
In this notation, the condition becomes $c''(t)=-\Gamma_{c(t)}(c'(t),c'(t))$.
\end{notation}
Now, let $(U_\alpha,\phi_\alpha)$ be a chart on a pseudo-Riemannian manifold $(M,g)$. Then $(\phi_\alpha(U_\alpha),(\phi_\alpha^{-1})^*g)$ satisfies the conditions of claim \ref{firvar}. Therefore, the energy integral of a curve $]-\epsilon,\epsilon[\ra{c}M$ is stationary with respect to all compactly supported variations if and only if, in a chart $(U_\alpha,\phi_\alpha)$,
$$(\phi_\alpha\circ c)''(t)=-\Gamma^\alpha_{\phi_\alpha\circ c(t)} ((\phi_\alpha\circ c)'(t),(\phi_\alpha\circ c)'(t)).$$
\footnote{Here $\Gamma^\alpha$ is the $\Gamma$ defined in notation \ref{chrissym} for $U=\phi_\alpha(U_\alpha)$. With the sign convention of equation \ref{eq:chri}, these are the Christoffel forms for the Levi-Civita connection with respect to the local trivialisations of $TM$ induced by the charts $\phi_\alpha$.} We see that the tangent lifts of these coordinate curves $\phi_\alpha\circ c$ are precisely the integral curves of the local geodesic spray whose vertical component is $-\Gamma^\alpha\circ \Delta^\alpha$, where $T(\phi_\alpha(U_\alpha))\ra{\Delta^\alpha}T(\phi_\alpha(U_\alpha))\times_{\phi_\alpha(U_\alpha)}T(\phi_\alpha(U_\alpha))$; $X\mapsto (X,X)$. If we now specify a torsion tensor (say, torsion zero), theorem \ref{ambpalsing} gives us a unique affine connection on $TM$ with these curves as geodesics (or, equivalently, with this local geodesic spray). This connection is known as the Levi-Civita connection induced by the metric.

\chapter{Gauge Theories}
This chapter will start with a very imprecise introduction, departing from Einstein's general relativity and progressing along Maxwell's electromagnetism towards a general first-quantised formalism of gauge theories (of the Yang-Mills kind). It will be assumed that the reader has some familiarity with the first two theories. I will therefore not be very strict when discussing them, leaving a lot of mathematical and physical details to be filled in by the reader. They only serve as a stepping stone, setting the philosophical tone for the culmination of this thesis in a general framework of gauge theories on principal bundles.

We will work in units $c=G_\mathrm{Newton}=\hbar=\epsilon_0=\mu_0=1$.

\section{Principal Bundles and Physics}
After this abstract mathematical introduction, I can imagine the reader finds himself wondering how principal bundles ever found their way into physics. Funnily enough, the answer is to be sought in the abstract results of section \ref{catfib}, in particular the interpretation of principal bundles as bundles of frames and of their associated bundles. In this paragraph, I will try to sketch a philosophical framework of ideas, which I hope will help the reader understand the use of principal bundles in physics better. However, to fully appreciate the matter, a specific understanding of the technical details of the physical theories is essential.\\
\\
The reader must realise that a conceptual shift has taken place in mathematical physics, over the last century, which is reflected in the mathematical language used to describe nature. I am talking, of course, about the transition from a coordinate approach to physics, using mostly tools from real analysis, towards a coordinate-free approach, focussing on the geometry of manifolds. I claim that this has at the same time been a shift in a philosophical point of view: from empiricist to realist.

In Newtonian mechanics, according to the principles of Galilean invariance, positions and velocities are defined only relatively. These quantities represent nothing intrinsic about what is happening in spacetime. They are the result of measurements that were performed by a certain observer and therefore tell us as much about this particular observer as they tell us about what that observer is observing. Here, the Galilean transformations relate one (inertial) observer to another and therefore capture, in a sense, the extent to which our description of nature is not intrinsic.

Of course, a similar principle remains when we proceed to the classical form of the special theory of relativity. Here, it is reflected by the Poincar\'e invariance\footnote{The group of Poincar\'e transformations consists of all isometries of Minkowski spacetime, i.e. translations and Lorentz transformations.} of the theory, the Poincar\'e transformations representing the relation between all possible (inertial) observers. If we include Maxwell's electromagnetism, not only the positions and velocities are non-intrinsic quantities of spacetime: the same holds for electric and magnetic fields. What one observer would call an electric field would be a magnetic field for another observer. (See, for example, \cite{jaccla}.)

With general relativity, however, the language of manifolds made its entrance in physics. By using this language, one seeks to remove this ambiguity in the description of nature. In general relativity, all quantities that are of physical interest are represented by mathematical objects related to the spacetime manifold $M$ that are independent of any choice of coordinates on $M$ or any choice of a (local) frame for $TM$. Particular measurements of these quantities correspond to the description of these objects with respect to some coordinates or a (local) frame associated with the observer. One sees the same phenomenon in the recent Hamiltonian formulation of classical mechanics in the language of symplectic manifolds.

However, such an abstract framework is somewhat unnatural or, rather, somewhat incomplete. Since all measurements are made with respect to some reference frame (of an observer) and the measured quantity can never completely be divorced from the measurement process, it would seem foolish not to represent the measurement properly in the mathematical theory describing the physics. It is precisely here that principal bundles come into play. Principal bundles provide a more natural framework for modern theories of physics: they combine the realist aspect of an intrinsic description of nature and the empiricist aspect of the measuring process. They seem to be a useful compromise between the mathematician and the physicist.\\
\\
In many theories of physics, observable quantities are represented by sections of certain fibre bundles over spacetime. If such a fibre bundle is well behaved, in the sense that its structure group is finite-dimensional (i.e. it is a $(G,\lambda)$-fibre bundle for an effective Lie group action $\lambda$), then we can (according to section \ref{sec:asbun}) understand it to be an associated bundle of some principal bundle (the bundle of generalised frames). From this point of view the sections of the original fibre bundle can be interpreted (using claim \ref{seccoo}) as equivariant maps from the principal bundle to some model fibre for the fibre bundle.

The reader is urged to think of the elements of the principal bundle as generalised frames for the original fibre bundle. This means that they correspond to different ways that we can convert the intrinsic dynamics that is described in an abstract way by a section of a fibre bundle to something concrete that we observe. These generalised frames are known to physicists as \emph{gauges}\index{gauge}, the structure group of the principal bundle is called the \emph{gauge group}\index{gauge group} and an automorphism of the principal bundle that fixes the base is called a \emph{gauge transformation}. In the case of a (metrised) vector bundle\footnote{From a physical perspective, this is by far the most important case.} this means that, by choosing a gauge, or an (orthonormal) frame for the fibre, we obtain a set of numbers: the coordinates of the section with respect to the frame.

The gauges often correspond to certain classes of observers (or, more precisely, to the clocks and metre sticks they use). Take, for example, Maxwell's classical electromagnetism on flat spacetime. Here, the different inertial frames are precisely the gauges and the principal right action corresponds to the Lorentz transformations between the inertial frames. In general, however---in particular in modern particle physics---the gauges might represent something more abstract and might include the choice of a basis for the space of internal degrees of symmetry of a quantum particle (think of spin). The picture of the fibre of a principal bundle as a Lie group without a specified position of the identity then reflects the physical principle that all inertial observers are equivalent (and, more generally, that all gauges are equivalent). All equations of motion, as formulated on the principal bundle, are therefore invariant under the principal right action.

The general idea of the principal bundle approach is thus the simultaneous description of nature in a realist and empiricist fashion. The intrinsic description is made by viewing a quantity as a section of an associated bundle. Specific coordinate representations of this section are obtained by choosing a gauge and the structure group acts transitively and freely on the fibre, thereby interrelating all allowed gauges.

\section{Einstein's General Relativity}
The idea of gauge theories starts with Einstein's general theory of relativity, as this is where an intrinsic formulation in terms of manifold geometry was first introduced in physics. Remember that Einstein's special theory of relativity is a theory of mechanics that was devised to solve the incompatibility between Newton's laws of mechanics and Maxwell's electromagnetism\footnote{Indeed, analogues of Newton's three laws of mechanics hold in special relativity.}. Similarly, Einstein's general theory of relativity was born from the need to make the special theory compatible with Newton's law of gravitation\footnote{Here, one can still formulate a form of Newton's three laws of mechanics that holds equally in the presence of gravity.}. It is a theory of gravity and mechanics; however, one would do it more justice by calling it a theory of spacetime geometry, as the paradigm shift it encompasses seems far more interesting than its precise numerical predictions. It is the best theory of gravity we have and its predictions, such as the perihelion shift of Mercury, gravitational lensing, gravitational redshift and even (arguably) the existence of black holes have been successfully verified. However, it is a classical theory in the sense that it does not deal with quantum-mechanical effects. The merger of quantum mechanics and general relativity is still, after almost a century, the holy grail of modern physics. I will briefly discuss the principles of general relativity from a point of view that fits in with the topics discussed in this thesis.

The central observation that might lead one to the mathematical framework of general relativity is the following. Let us consider a test particle in a gravitational field $g$ (a time-dependent vector field on $\R^3$) and let us denote its position at time $t$ by $x(t)\in \R^3$. Then Newton's second law, combined with his law of gravitation, tells us that
$$m_i x''(t)=m_g g,$$
where $m_i$ is the particle's `inertial mass' and $m_g$ is its `gravitational mass', which we have known to coincide since Galileo; we therefore just talk about plain `mass'. This means that the equation of motion of the particle,
$$x''(t)=g,$$
does not depend on the particle's mass. This principle is known as the \emph{weak equivalence principle}. It opens the possibility to describe gravitation not as an interaction that might differ from particle to particle but as an intrinsic property of spacetime itself.

Another, more vague idea that motivates the formalism of general relativity is \emph{Mach's principle}. Briefly, it is the idea that the definition of what one would call non-accelerating or non-rotating motion should not be absolute. It should depend on all matter in the universe. Acceleration should be measured relatively. This principle is satisfied neither by Newtonian nor by special-relativistic mechanics.

These two principles combined with the requirement that general relativity reduces to special relativity and Newtonian gravity in the appropriate limits allegedly led Einstein to his theory of gravitation. \cite{walgen} In general relativity, spacetime is represented by a four-dimensional pseudo-Riemannian manifold $(M,g)$ with metric signature $(-1,+1,+1,+1)$. There are, in general, no global notions of space and time. Of course, we obtain the Levi-Civita connection from the metric and therefore have a concept of spacetime geodesics, which will take the place of straight lines in special relativity, and of curvature. An observer is often modelled as a point particle and therefore as a curve $]-\epsilon,\epsilon[\ra{\gamma}M$. The fact that no observer can exceed the speed of light is formalised by the fact that $\gamma$ has to be a timelike curve, i.e. its tangent vectors have a negative length.

Matter/energy is described by the so-called stress-energy tensor\index{stress-energy tensor} $T\in \Gamma(T^*M\otimes T^*M)$, a symmetric $(0,2)$-tensor on spacetime. The analogue of Newton's second law, the so-called \emph{Einstein equation}\index{Einstein equation} reads
$$8\pi T= \mathrm{Ric}-\frac 1 2 \mathcal{S} g=:G,$$
where the \emph{Ricci tensor}\index{Ricci tensor} $\mathrm{Ric}\in \Gamma(T^*M\otimes T^*M)$ is given by taking the trace\footnote{What do we mean by the trace? By using the metric isomorphism $TM\ra{g}T^*M$ we assume without loss of generality that $T$ is a tensor of type $(0,n)$. By taking the trace $\mathrm{tr}^{k,l}\, T$ of $T$ over input $k$ and $l$, we mean that we pick a local orthonormal frame $(e_\mu)$ for $TM$ and a dual frame $(\omega^\mu)$ for $T^*M$ and set $(\mathrm{tr}^{k,l}\, T)(X_1,\ldots,\hat{X}_k,\ldots,\hat{X}_l,\ldots,X_n)=\sum_{\mu=0}^3T(X_1,\ldots,X_{k-1},e_\mu,X_{k+1},\ldots,X_{l-1},g^{-1}(\omega^\mu),X_{l+1},\ldots,X_n)$  } of the curvature tensor over the second and fourth inputs, so explicitly $\mathrm{Ric}(X,Y)=(\mathrm{tr}^{2,4}\, R)_m(X,Y)=\mathrm{tr}\, (R_m(X,-)Y)$ ($R$ denotes the curvature tensor) and the \emph{curvature scalar}\index{curvature scalar} $\mathcal{S}$ is defined to be $\mathcal{S}(m):=\mathrm{tr}\: \mathrm{Ric}_m$. The right side of the Einstein equation is denoted by $G$ and is known as the Einstein tensor.

Just as Newton's second law has to be completed with a specification of a form of the force $F$ (take, for example, the force of a spring that satisfies Hook's law: $F=- k x$), Einstein's equation has to be completed with a specification of the form of the stress-energy tensor, called an equation of state. There are many common forms that lead to relevant solutions: for instance, the stress-energy tensor of a perfect fluid\footnote{This is used as a model for various states of matter ranging from dust clouds to the interior of stars. Moreover, it is used in the Robertson-Friedmann-Walker model of cosmology as the average energy-momentum tensor of our universe. \cite{walgen}} reads $T=\rho\, u\otimes u+P(g+u\otimes u)$, where $\rho,P\in C^\infty(M)$ represent the mass density and the pressure, respectively, and $u\in \mathcal{X}(M)$ represents the unit tangent vectors to the world lines of the fluid particles. The expression for the stress-energy tensor of an electromagnetic field is also well-known and it can be found in section \ref{sec:max}. 

The Bianchi identity for the curvature $R$ of the Levi-Civita connection gives us the second important equation of general relativity:
$$\mathrm{tr}^{1,2}\, (\nabla G)=0,$$
where $\nabla$ denotes the covariant derivative of the connection on the $(0,2)$-tensor bundle, induced by the Levi-Civita connection. The derivation of this result is a useful exercise for the reader, or can be found in any textbook on general relativity (take for instance \cite{misgra}). The Einstein equation gives a physical interpretation of this equation as a generalised statement of energy-momentum conservation:
$$\mathrm{tr}^{1,2}\, (\nabla T)=0.$$
One of the best known statements of general relativity is that a point particle moves through spacetime along timelike geodesics. What is less known is that this is not an extra assumption but indeed a consequence of the previous equation, if one models a point particle as a limit of certain stress-energy tensors with decreasing support. This result is clearly explained in \cite{ehlequ}.\footnote{In physics textbooks the result is usually derived arguing that on physical grounds an action principle should hold for the motion of the point particle, i.e. they argue that the energy integral of the point particles should be stationary. Of course, we have seen in section \ref{sec:riem} that this is precisely saying that they should be geodesics of the Levi-Civita connection of the metric.}

Although I have never found the comparison in the literature, it seems almost foolish not to make the following remark, in the context of this thesis. I encourage the reader to compare the Einstein equation with the (second) inhomogeneous Maxwell/Yang-Mills equation and to compare this equation of energy-momentum conservation with the (first) homogeneous Maxwell/Yang-Mills equation. Note that in all cases the inhomogeneous equation involves sources of the fields, while the homogeneous equation, which is actually a Bianchi identity in disguise, does not.\\ 
\\
Finally, it is useful to formulate an action from which the Einstein equation follows.\footnote{I will not be very formal about this, since I will be dealing with action principles in a formal way later on. However, this remark should help sketch the resemblance between Einstein's general theory of relativity and more modern gauge theories we will see later.} This was first done by Hilbert. The action he found was
$$S(g):=\int_M \mathcal{S}\;\mu,$$
which is known as the \emph{Einstein-Hilbert action}\index{Einstein-Hilbert action}. Here $\mu$ denotes the Lorentzian volume element on $M$ which we know to exist if $M$ is oriented. (If $M$ is non-orientable, the use of the orientation can be avoided, as we shall later see. Moreover, if $M$ is non-compact, the integral might not converge. Both of these issues are dealt with by considering the integral over small patches of $M$ and demanding it to be stationary on each patch.) If we see this action as a function from the infinite-dimensional manifold of Lorentzian metrics on $M$ to the real numbers, we find that the vacuum Einstein equation is satisfied by $g$ precisely if $S$ has a stationary point at $g$. \cite{walgen}

\section{Maxwell's Electromagnetism}\label{sec:max}
This paragraph on electromagnetism will serve as a stepping stone for proceeding to general Yang-Mills theory. I will start by discussing Maxwell's electromagnetism and will demonstrate that, in a natural fashion, this can give rise to considering a certain principal $U(1)$-bundle over spacetime: we interpret electromagnetism as a Yang-Mills theory.

Let $(M,g)$ be a four-dimensional Lorentz manifold, representing spacetime. We will use the metric convention $(-1,+1,+1,+1)$. Classical electromagnetism is usually described by a $1$-form $j\in\Omega^1(M)$, representing the electric charge and current, and a $2$-form $F\in \Omega^2(M)$, representing the electromagnetic field strength. These charges and electromagnetic fields interact according to \emph{Maxwell's equations}\index{Maxwell's equations}, which, in modern language, take the form
$$dF=0\txt{and} \delta F=j,$$
where we write $\delta$ for the codifferential $\Omega^k(M)\ra{\delta}\Omega^{k-1}$, which is related to $\Omega^k(M)\ra{d}\Omega^{k+1}$ by the Hodge star as follows. On a pseudo-Riemannian manifold $(M,g)$
$$\delta=\mathrm{sign}\; (g) (-1)^{nk+n+1}* d *,$$
where $n=\dim M$ and $\mathrm{sign}\; (g)=\pm 1$ denotes the signature of the metric (so in our case $n=4$ and $\mathrm{sign}\; (g)=-1$). 

These equations describe the propagation of electromagnetic fields along spacetime. Of course, if these electromagnetic fields carry a significant amount of energy, this influences the geometry of spacetime. The electromagnetic contribution to the stress-energy tensor is given by (\cite{walgen})
$$T=\frac 1{4\pi}\left(\mathrm{tr}^{2,4}F\otimes F-\frac 1 4 g\otimes \mathrm{tr}^{1,3}\mathrm{tr}^{2,4}F\otimes F \right).$$
This set of three equations forms the foundation of a theory of electromagnetism on curved spacetime, known as the \emph{Einstein-Maxwell formalism}.

Of course, we can give physical interpretations to these differential forms. Suppose $]- \epsilon, \epsilon[\ra{ c}M$ is the world line of an observer, i.e. an affinely parametrised timelike geodesic (so $g(c'(t),c'(t))=-1$). This is called the observer's \emph{proper time} parametrisation. We can find a neighbourhood $U\subset M$ of $c(0)$ on which geodesic normal coordinates $(x^\mu)_{0\leq \mu \leq 3}$ are defined such that $x^0(c(t))=t$. The interpretation of these coordinates will be what the observer would call time and space, $(x^\mu)_{1\leq \mu\leq 3}$ being the space coordinates and $x^0$ the time. This chart defines a local frame for $TM$ and, more generally, for an arbitrary tensor bundle of $M$. Let us write $T^{k\ldots m}_{n\ldots p}$ for the component of a tensor $T$ with respect to the basis vector $\frac \dau {\dau x^k}\otimes \ldots \otimes \frac \dau{\dau x^m}\otimes dx^n\otimes\ldots\otimes dx^p$. Then the interpretations of $j$ and $F$ are the following for the observer.

$j_0$ is interpreted as a $-\rho$, where $\rho$ is the electric charge density, while $j_1$, $j_2$ and $j_3$ correspond to the three spatial components of the electric current density. Moreover,
$$(F_{\mu\nu})_{0\leq\mu,\nu\leq 3}=\left(\begin{array}{cccc} 0 & E_1 & E_2 & E_3\\
-E_1 & 0 & -B_3 & B_2\\
-E_2 & B_3 & 0 & -B_1\\
-E_3 & -B_2 & B_1 & 0
\end{array}\right),$$
where $(E_1,E_2,E_3)$ and $(B_1,B_2,B_3)$ are the spatial components of what the observer would call, respectively, electric and magnetic fields. It is easily verified that these equations give the ordinary Maxwell equations for $E$ and $B$ for Minkowski spacetime. \cite{thicla}

The Maxwell equations tell us how a charge-current distribution $j$ creates an electromagnetic field. Conversely, the \emph{Lorentz force law}\index{Lorentz force law} describes how the electromagnetic field influences the motion of a charged point particle. In the context of general relativity it can be stated as follows. Let $]-\epsilon,\epsilon[\ra{\gamma} M$ be the world line of a point particle with charge $q$ and mass $m$ and assume that it is parametrised such that $g(\gamma'(t),\gamma'(t))=-1$. Then, in any local frame $(e_\mu)$ for $TM$, using the Einstein summation convention:
$$\frac{D\gamma'}{dt}=\frac{q}{m}\,\bigl(i_{\gamma'}F\bigr)^\sharp.$$
(Here $\sharp$ denotes raising an index with the spacetime metric. Remember that $D/dt$ denotes the covariant derivative corresponding to the pullback connection $\gamma^*P_V$ if $P_V$ denotes the Levi-Civita connection on $TM$. Explicitly, for a vector field $s$ along $\gamma$, $Ds/dt=K\circ s'$.) In particular, a charged particle does not move along spacetime geodesics if electromagnetic fields are present!

It is important to note that the first Maxwell equation tells us that $F$ is closed. By the Poincar\'e lemma, this means that it is locally exact: for every point $m\in M$, we find a neighbourhood $U\subset M$ and an $A\in \Omega^1(U)$ such that $F=dA$ on $U$. This $1$-form $A$ is known as the \emph{vector potential}. It is worth noting that we can globally define $A$ in the case of Minkowski spacetime since it is star-shaped. Note that we have freedom in the choice of $A$. Indeed, we can add arbitrary closed $1$-forms to our vector potential. Of course, physicists use this to their advantage and choose appropriate local potentials for an enormous simplification of calculations.

Again, in the source-free case, we can find an action, called the \emph{Einstein-Hilbert-Maxwell action}, from which the field equations will follow. Indeed, if locally $F=dA$, then $(g,A)$ satisfies the Einstein equation with electromagnetic stress-energy and the source-free Maxwell equation precisely when
$$S(g,A)=\int_M(\mathcal{S}-\frac{1}{2}\mathrm{tr}^{1,3}\mathrm{tr}^{2,4}(F\otimes F))\mu$$
is stationary with respect to $g$ and to compactly supported variations of the potential $A$.

\section{Intermezzo: Geometrised Forces}
\label{sec:geomfor}
There is an even more modern approach to classical electromagnetism, using the language of principal bundles, which this entire thesis has been working towards. Let us first ask why we are not satisfied with the formulation I just discussed. I think the primary argument, which I found in \cite{blegau}, might be the following: it does not provide enough geometric insight. Just like Einstein geometrised Newton's theory of gravity, we would like to geometrise Maxwell's electromagnetism.

In the absence of any forces, Newton's second law tells us that the motion of a particle is a geodesic (i.e. straight line) in three-dimensional Euclidean space. This is a beautiful geometrical property of the theory: if the trajectories of two point particles are tangent at some point, they are the same (although they might be parametrised differently).

This simplicity is lost when we allow gravitational forces to act on the particles. Indeed, imagine we launch two projectiles from the same launching platform in exactly the same direction but with different velocity. It may happen that one escapes the earth's gravitation while the other does not. This problem was resolved by Einstein. By considering a four-dimensional spacetime with a connection, we obtain new geodesics along which point particles travel. In this setting, the motion of two particles again coincides when it is tangent at some point of spacetime. The trajectories of the two projectiles were not tangent in Euclidean space but are in spacetime. By adding a dimension to our space, we geometrised the gravitational force.

However, the same problem still holds for electromagnetism. Suppose the two projectiles were charged, say electrons, then they would still follow different trajectories because of the earth's magnetic field. We can hope to resolve this problem by a similar trick. We try to geometrise electromagnetism by considering the motion of the particle on a five-dimensional `charge-spacetime' with a connection, on which the charged point particle will travel along geodesics. This construction was first worked out by Kaluza (1921) and Klein (1926). 

\section{Electromagnetism as a $U(1)$-Gauge Theory}
In modern language the Kaluza-Klein approach to electromagnetism is best described using a principal $U(1)$-bundle $P\ra{\pis}M$ over spacetime. We will interpret the electromagnetic tensor as the curvature of some principal connection $\omega\in \Omega^1(P,\mathfrak{u}(1))$. The local choice of a gauge potential $A_\alpha$ will then correspond to the pullback of the connection $\omega$ along a local section $e_\alpha\in \Gamma(P|_{U_\alpha})$ of $\pis$: $A_\alpha=ie\, e_\alpha^*\omega$, where $e\in \R$ is the fundamental unit of electric charge of which every other charge is a multiple\footnote{Experiments tell us that such a smallest quantity of charge should exist. \cite{minwea}}. The question still is, however, what principal bundle we should be considering.

\subsection{Potentials and Principal Connections}
Suppose $(U_\alpha)$ is an open cover of $M$ and that we have made a choice of $A_\alpha\in \Omega^1(U_\alpha)$ of local electromagnetic vector potentials, i.e. $dA_\alpha=F|_{U_\alpha}$. Then, on $U_\alpha\cap U_\beta,$\; $d(A_\alpha-A_\beta)=0$. Suppose we have taken a suitable cover, such that $U_\alpha\cap U_\beta$ is contractible, say. Then the first de Rham cohomology group is trivial, so we conclude that, on $U_\alpha\cap U_\beta,$\; $A_\beta-A_\alpha=df_{\alpha\beta}$, where $f_{\alpha\beta}$ is some real-valued smooth function. Obviously, we can define these such that $f_{\alpha\beta}=-f_{\beta\alpha}$. Moreover, let us, by convention, write $\omega_\alpha:=-i A_\alpha/e$. Then,
$$\omega_\beta=\omega_\alpha-i\, df_{\alpha\beta}/e.$$
Let us write 
$$g_{\alpha\beta}:=\exp(-i f_{\alpha\beta}/e).$$
Then $U_\alpha\cap U_\beta \ra{g_{\alpha\beta}} U(1)$ and $g_{\alpha\beta}\cdot g_{\beta\gamma}\cdot g_{\gamma\alpha}=\exp(-i (\delta f)_{\alpha\beta\gamma}/e)=1$,
where we write $(\delta f)_{\alpha\beta\gamma}=f_{\alpha\beta}+f_{\beta\gamma}+f_{\gamma\alpha}$ (following the convention in \v Cech cohomology), if and only if $(\delta f)_{\alpha\beta\gamma}\stackrel{!}{\in} 2\pi e \Z$. By theorem \ref{cfbcon}, we see that this is precisely the condition that states that $(g_{\alpha\beta})$ would define a coordinate principal $U(1)$-bundle $P\ra{\pis}M$, with local trivialisations which we shall write as $(\psi_\alpha)$. Now, why does this condition hold? (This is highly non-trivial.)

We easily see that $d((\delta f)_{\alpha\beta\gamma})=d(f_{\alpha\beta}+f_{\beta\gamma}+f_{\gamma\alpha})=
A_\beta-A_\alpha+A_\gamma-A_\beta+A_\alpha-A_\gamma=0$, thus $(\delta f)_{\alpha\beta\gamma}$ is constant, i.e. equal to some $c_{\alpha\beta\gamma}\in\R$. Although it is far from easy to see that $c_{\alpha\beta\gamma}\in 2\pi e \Z$, this can indeed be argued on physical grounds\footnote{The argument rests on the fact that there is a smallest quantity of electric charge in nature.}. The argument is originally due to Dirac and it is known as the \emph{Dirac quantisation argument}\index{Dirac quantisation argument} or \emph{Dirac quantisation condition}. The interested reader can find a modern treatment in \cite{minwea}\footnote{The nLab page on the 'electromagnetic field' by Urs Schreiber features a brief but very insightful explanation of the argument as well.}.

What can be easily demonstrated, however, is that we can choose $f_{\alpha\beta}$ such that $c_{\alpha\beta\gamma}=0$, if the second \v Cech cohomology group $\breve H^2(M,\R)$ with values in the sheaf of constant real functions is trivial\footnote{The condition that $\breve H^2(M,\R)=\{0\}$ is equivalent to demanding that the second de Rham cohomology group $H^2_{dR}(M)$ is trivial. \cite{botdif} In particular, we see that the condition is satisfied locally by any spacetime. The interested reader can find an account of \v Cech cohomology in \cite{botdif}.}. Indeed, suppose we had made an original choice $\tilde f_{\alpha\beta}$, such that $A_\beta-A_\alpha=d\tilde f_{\alpha\beta}$ and we write $\tilde c_{\alpha\beta\gamma}:=(\delta \tilde f)_{\alpha\beta\gamma}\in\R$. Note that $(\delta\tilde c)_{\alpha\beta\gamma\delta}=(\delta^2\tilde f)_{\alpha\beta\gamma\delta}=0$ to see that $\delta \tilde f=\tilde c$ defines a class $[\tilde c]\in \breve H^2(M,\R)$. If $\breve H^2(M,\R)$ is trivial, we therefore have $\tilde c_{\alpha\beta\gamma}=(\delta \tilde c')_{\alpha\beta\gamma}$, for some $\tilde c'_{\alpha\beta}\in\R$. We choose $f_{\alpha\beta}:=\tilde f_{\alpha\beta}-\tilde c'_{\alpha\beta}$. Then, obviously, $d f_{\alpha\beta}=d\tilde f_{\alpha\beta}-d\tilde c'_{\alpha\beta}=d\tilde f_{\alpha\beta}=A_\beta-A_\alpha$ and $c_{\alpha\beta\gamma}=(\delta f)_{\alpha\beta\gamma}=(\delta \tilde f)_{\alpha\beta\gamma}-(\delta \tilde c')_{\alpha\beta\gamma}=0$.

We see that the electromagnetic potential defines a principal $U(1)$-bundle in this way. If $H^2_{dR}(M)$ is trivial this is obvious, while the general case is more subtle.\\
\\
In this setting, the potential itself will be interpreted as a principal connection on this bundle. Let us write $e_\alpha$ for the local section of $\pis$ defined by $e_\alpha(m):=\psi_\alpha^{-1}(m,1)$. Then, a collection $\omega_\alpha\in \Omega^1(U_\alpha,\mathfrak{u}(1))$ is easily seen (\cite{blegau}, p. 32) to define\footnote{Explicitly, let $U_\alpha\ra{\sigma_\alpha}P$ be the local section of $P\ra{\pis}M$ associated to the local trivialisation $(U_\alpha,\psi_\alpha)$. If $p=\sigma_\alpha(m)\cdot g$, each $X_p\in T_p(\pis^{-1}U_\alpha)$ can be uniquely decomposed as $X_p=T\rho_g(T\sigma_\alpha(Y_m))+Z|_P(p)$, for some $Y_m\in T_mU_\alpha$ and $Z\in\g$. Define $\omega(X_p):=\mathrm{Ad}(g^{-1})\omega_\alpha(Y_m)+Z$. One readily checks, using the displayed compatibility condition on overlaps, that these local definitions glue to a well-defined principal connection.} a principal connection $\omega\in \Omega^1(P,\mathfrak{u}(1))$ such that $\omega_\alpha:=e_\alpha^* \omega$ if and only if
$$\omega_\beta(X_m)=\mathrm{Ad}(g_{\alpha\beta}(m)^{-1})\, (\omega_\alpha(X_m))+TL_{g_{\alpha\beta}(m)}^{-1}(Tg_{\alpha\beta}(X_m))\qquad \forall X_m\in T_m (U_\alpha\cap U_\beta).$$
In our case of an abelian Lie group, this condition takes the simple form (if we also simplify the notation somewhat)
$$\omega_\beta=\omega_\alpha+g_{\alpha\beta}^{-1} d g_{\alpha\beta},$$
which we immediately see to be met in our definition\footnote{Indeed, $dg_{\alpha\beta}=d\exp(-i f_{\alpha\beta}/e)=-i/e\exp(-i f_{\alpha\beta}/e)df_{\alpha\beta}=
-i/eg_{\alpha\beta}(A_\beta-A_\alpha)=g_{\alpha\beta}
(\omega_\beta-\omega_\alpha)$.}. We conclude that we obtain a principal connection $\omega\in \Omega^1(P,\mathfrak{u}(1))$, representing the electromagnetic potential. Its curvature represents the electromagnetic tensor. Indeed, by equation \ref{maurer},
$$e_\alpha^*\Omega^\omega=e_\alpha^*d^\omega \omega=e_\alpha^*d\omega+\frac12 e_\alpha^*[\omega,\omega]=e_\alpha^*d\omega=d\omega_\alpha=-i\, dA_\alpha/e=-i\, F|_{U_\alpha}/e.$$
Note that the vanishing of the term involving the Lie bracket is characteristic for electromagnetism, which is the only gauge theory of particle physics involving an Abelian group. Because of this, we easily see that the Bianchi identity for $\Omega^\omega$ implies the first Maxwell equation
$$0=e_\alpha^*d^\omega \Omega^\omega=e_\alpha^* d\Omega^\omega=-i\, dF|_{U_\alpha}/e,$$
where the first equality holds because of the Bianchi identity and the second because of commutativity of $U(1)$.\\
\\
So we see that electromagnetism gives rise to a principal $U(1)$-bundle over spacetime\footnote{Although we should note that the principal bundle we ended up with depended on a lot of choices, for example that of the cover $(U_\alpha)$ of $M$.} with an equivariant connection. The assumption $H^2_{dR}(M)=\{0\}$ is sufficient for the elementary construction above, and in any case the construction is local on spacetime. Globally, non-trivial $H^2_{dR}(M)$ is precisely the kind of topological obstruction that can force one to work with a non-trivial $U(1)$-bundle and local potentials rather than a single global potential. Of course, whatever our spacetime is, we can construct the principal bundle locally.

We see that such a formalism of electromagnetism can be constructed for important spacetimes and is equivalent to the classical Einstein-Maxwell formalism, but why should we care? For one thing, it interprets electromagnetism as a so-called \emph{Yang-Mills theory}, and it can be used, following the historical development, as a stepping stone to general Yang-Mills theories. In general, a Yang-Mills theory describes some fundamental interaction of matter (electromagnetism, weak interaction, electroweak interaction, strong interaction, etc.) using a very similar setup, in which $U(1)$ is replaced by some other compact Lie group $G$. These theories have been intensively studied and, after quantisation of the fields (also called second quantisation), yield the \emph{standard model of particle physics}, the best theory of particle physics currently available.

However, up to this point, nothing has been said about quantum mechanics. We have seen that the formalism arises naturally in a non-quantum-mechanical setting. This is often overlooked in the literature and many references, like \cite{frageo}, immediately start with a quantum-mechanical formalism and use more complicated arguments to advocate the $U(1)$-principal bundle formalism. Svetlichny, in his celebrated lecture notes on gauge theory, even goes as far as saying: ``Interpreting electromagnetism as a gauge theory on PU(1) does not offer any special advantages except when one considers quantum theory or extensions to situations that transcend classical Maxwellian theory. One such advantage is seen in trying to define magnetic monopoles." \cite{svepre}

I happen to disagree. I think this under-appreciated approach to classical electrodynamics provides a beautiful geometric picture of what is going on in nature and is a more natural formulation of electromagnetism on curved spacetime, along the lines of the discussion in section \ref{sec:geomfor}. I will elaborate.

\subsection{Geodesic Motion}
Note that any real inner product on $\mathfrak{u}(1)$ is $\mathrm{Ad}$-invariant. Let us, for convenience, pick the one such that $k(i,i)=1$. This defines a non-degenerate pseudo-Riemannian metric (of signature $(-1,+1,+1,+1,+1)$) on $P$ by the formula
$$h:=\pis^* g+k\circ (\omega\otimes\omega).$$
(Non-degeneracy is verified by decomposing a vector $X\in TP$ into its horizontal and vertical parts with respect to the connection $\omega$.) This metric includes contributions from both the spacetime metric $g$ (which is analogous to a gravitational potential\footnote{Strictly speaking, this analogy is not entirely correct. Although the curvature tensor tells us a lot about the gravitational field and it is often said in the physics literature, e.g. \cite{walgen}, to represent the field strength, it is not enough to completely characterise it; nor is the Levi-Civita connection. Indeed, the Einstein equation relates the stress-energy tensor (and therefore the spacetime distribution of matter) to the Einstein tensor (which is built from the metric and cannot be built from the connection alone), which therefore comes closest to something that we would like to call a `gravitational field strength'. On the other hand, the analogy is not too bad, as the connection indeed determines the motion of point particles (i.e. geodesics).}) and the electromagnetic potential $\omega$. 
\begin{remark}\label{rem:isom}Note that the principal right action acts by isometries of this metric.\end{remark}

We can argue that this metric plays a similar role in a combined electromagnetic-gravitational theory as the spacetime metric does in the theory of relativity. For one thing, we will see that its scalar curvature defines the correct action density, from which we obtain both Einstein's equation and the (non-Bianchi) Maxwell equation, when we vary both the spacetime metric and the electromagnetic connection. (See theorem \ref{thm:scacurv}.) Thus, it helps us obtain the correct equations that describe the evolution of the fields.

Now, it turns out that we can also obtain the equations of motion for a charged point particle in these fields from it. Indeed, they are just the geodesics of $h$, just as the motion of a point particle in only a gravitational field is along the geodesics of $g$.

\begin{theorem}[Geometrised forces] \label{thm:genchar}
Let $(M,g)$ be a pseudo-Riemannian manifold and let $G$ be a Lie group with an Ad-invariant inner product $k_\g$ on its Lie algebra $\g$. Let $P\ra{\pis}M$ be a principal $G$-bundle over $M$ that is equipped with a principal connection $\omega$ and let us write $\Omega^\omega$ for its curvature. Finally, let $]-\epsilon,\epsilon[\ra{\gamma}P$ be a geodesic for the metric
$$h=\pis^*g+k_\g\circ  (\omega\otimes\omega).$$
Then $\omega\circ \gamma'(t)$ is constant and equal to some element $Q\in \g$. Moreover, $c:=\pis\circ \gamma$ satisfies the equation
$$\frac {Dc'}{dt}=-\left(k_\g(Q,i_{\gamma'}\Omega^\omega)\right)^\sharp,$$
where the horizontal $1$-form $k_\g(Q,i_{\gamma'}\Omega^\omega)$ is identified with a $1$-form along $c$ by $T\pis$, and $\sharp$ raises an index with $g$.
\end{theorem}
Comparison with the Lorentz force law of section \ref{sec:max} gives the following. In particular, we see that we should interpret $Q$ as a charge-to-mass ratio for a kind of generalised charge, where the number of generalised charges of the particle corresponds to the dimension of the Lie algebra $\g$.
\begin{corollary}
Geodesics on the Kaluza-Klein principal $U(1)$-bundle project down to the motion of a charged pointlike test particle, with electric charge-to-mass ratio $q/m=-iQ/e$, on $M$, as governed by the Lorentz force law. In particular, we see that the motion of a charged point particle in an electromagnetic field on spacetime is determined by the demand that the action (energy integral)
$$S_h(\gamma)=\frac 1 2\int_{-\eps}^\eps h(\gamma'(t),\gamma'(t))dt$$
is stationary with respect to all compactly supported variations.
\end{corollary}
We see that the principal-bundle setup of electromagnetism provides a beautiful geometric framework for describing the motion of classical point particles. In particular, it provides a close parallel between electromagnetism and general relativity. Moreover, it exhibits the CPT symmetry (charge, parity, time) which is fundamental in nature: if $]-\epsilon,\epsilon[\ra{\gamma}P$ is a geodesic, so is $\gamma^*(t):=\gamma(-t)$ (with reversed spatial and temporal motion and opposite charge). On the other hand, any other combination of the $C$, $P$ and $T$ reversals does not map geodesics to geodesics.

\begin{remark}[The Aharonov-Bohm Effect]\index{Aharonov-Bohm effect}
The principal bundle formulation of electromagnetism places great emphasis on the potential $A$. We should therefore ask, in retrospect, what its physical significance is. Should we care about it at all, or is it only a mathematical trick to ease calculations? Since Maxwell's original formulation of electromagnetism, this had long been the predominant idea. Indeed, in classical electrodynamics there is no reason to think of the potential as anything more.

In a quantum-mechanical setting, however, an experiment can be conceived to distinguish between two states that have equal electromagnetic fields, but differ in potential. This experiment has been conducted and the predicted effect, called the Aharonov-Bohm effect, has indeed been observed. 

This effect is of tremendous theoretical importance. Indeed, it compels us to interpret the electromagnetic potential as something of real physical importance, as the more fundamental quantity than the electromagnetic field\footnote{To be precise, it presents us with a choice: either we abandon the principle of locality, which most physicists are certainly not willing to do, or we lose the electromagnetic field as the fundamental quantity in electromagnetism.}. Moreover, this marks the triumph of the Lagrangian approach to physics, using action principles and potential energies, over the Newtonian approach, using forces. For a clear exposition of the mathematical details of this effect (which has a useful topological interpretation), the reader is referred to \cite{baegau}.
\end{remark}

\section{Lagrangian Gauge Theories of the Yang-Mills Kind}
The principal bundle formalism for Yang-Mills-like gauge theories is a formalism for describing elementary particles and their interactions in a specially or generally covariant fashion. It is a semi-classical theory in the sense that the matter particles are treated in a quantum-mechanical way and are described by wave functions, while their interactions are governed by classical fields, called \emph{gauge fields}. The fields still have to be ``quantised'' in order to obtain a real quantum theory of fields, like the \emph{standard model}. This, however, goes beyond the scope of these notes. We will restrict ourselves to giving a brief outline of the framework of Yang-Mills theories, in this section.\\
\\
These gauge theories are a generalisation of the Kaluza-Klein electromagnetism in two senses. On the one hand, the generalisation is from electromagnetic fields to more general gauge fields\footnote{Examples of such fields include
\begin{enumerate}
\item $G=U(1)$: electromagnetism
\item $G=SU(2)$: weak interaction
\item $G=U(1)\times SU(2)$: electroweak interaction
\item $G=SU(3)$: strong interaction
\end{enumerate}
and combinations of these.}. We take this first step of generalisation by replacing the $U(1)$ in our formulation of electromagnetism by a general Lie group. On the other hand, it is a generalisation, since our framework will be able to deal with quantum-mechanical matter fields.\\ 
\\
Since physicists are often somewhat vague about what they mean by gauge fields and Yang-Mills-like theories, I will try to give a strict definition. 
\begin{definition*}[Gauge theory (of the Yang-Mills kind)] A gauge theory consists of the following information:
\begin{enumerate}
\item A Lorentzian manifold $(M,g)$;
\item a Lie group $G$ with an $\mathrm{Ad}$-invariant\footnote{Again, such an invariant inner product may be constructed by integration, for a compact Lie group.} inner product $k_\g$ on $\g$;
\item a principal $G$-bundle $P\ra{\pis}M$;
\item an action density $\mathscr{C}(P)\ra{\mathscr{L}_s}C^\infty(M)$, which is smooth in the obvious way, called the \emph{self-action density};\\
---------------------------------------------------------------------
\item a linear left $G$-action on a vector space $W$, on which we have an invariant inner product $k_W$ (from which we can form an associated vector bundle with an inner product);
\item an action density $C^\infty(P,W)^G\times \mathscr{C}(P)\ra{\mathscr{L}_i}C^\infty(M)$, called the \emph{interaction action density}.
\end{enumerate}
Quantities that are invariant under gauge transformations are said to be \emph{gauge-invariant}. In particular, we (almost always) require the action densities to be gauge-invariant.
\end{definition*}
The gauge fields are described by elements of $\mathscr{C}(P)$, while the (quantum) matter fields are sections of the associated vector bundle or, equivalently, elements of $C^\infty(P,W)^G$. One can also use only the ingredients before the horizontal line to obtain a source-free (or vacuum) Yang-Mills theory. These are direct generalisations of the previous section: $U(1)$ is replaced by an arbitrary Lie group and the Einstein-Maxwell action density is replaced by a general one\footnote{In this setting, we can construct the metric $h=\pis^*g+k_\g\circ(\omega\otimes\omega)$ on $P$. A classical (pointlike) test particle of mass $m$ will follow a geodesic $]-\epsilon,\epsilon[\ra{\gamma}P$. Theorem \ref{thm:genchar} tells us that the generalised charge $q=Qm=m\,\omega(\gamma'(t))\in \g$ is a constant of motion. Note that this constant depends on the initial point $\gamma(0)\in P$. In particular, this means that the concept of a generalised charge may not be gauge-invariant if we are dealing with a non-abelian Lie group $G$.}\\
\\
In first-quantised gauge theory, we usually model the states of matter particles as sections of certain metrised vector bundles $(V\ra{\pi}M,g)$ over spacetime. These can be seen as generalisations of the $\C$\,-valued wave functions, i.e. sections of the trivial complex line bundle, from ordinary quantum mechanics. The Hilbert space of the resulting quantum theory will, in general, be the space of $L^2$-sections of this vector bundle, with the inner product induced by the bundle metric. The physical interpretation of $g(\psi,\psi)\in C^\infty(M)$ is a probability density $4$-current corresponding to the probability of finding a particle in a certain region of spacetime. \cite{wacrel}

It is best to think of these vector bundles, along the lines of section \ref{sec:ortheq}, as being associated bundles to some principal bundle $P\ra{\pis}M$ - on which the gauge fields are modelled by principal connections - with a group $G=O(\Bbbk,(p,q))$ of orthogonal linear transformations of the standard fibre $W=\Bbbk^{p,q}$ as a structure group. In this line of thought, using claim \ref{seccoo}, the sections of $\pi$ can also be interpreted as $G$-equivariant maps from $P$ to $W$.

Again, one equation of evolution for the fields will be the Bianchi identity, $d^\omega \Omega^\omega=0$, while the other is derived from an action principle (as in the second (inhomogeneous) Maxwell equation), as is the equation for the matter fields (as in the Dirac equation for electrons). The actions will be integrals over the spacetime manifold $(M,g)$ of the corresponding action densities. These action principles will be the next thing to understand.

\subsection{Action Principles on Principal Bundles}
We have already encountered various action principles, from which the equations of evolution of the physical system can be derived. However, we have been somewhat vague about them. I will now elaborate on a rigorous treatment of action principles in general and more specifically in gauge theories of the Yang-Mills kind, as they are formulated on principal bundles.

The situation is often as follows. One has a mathematical formalism in which we have a space of (formally) possible states of the physical system: the \emph{configuration space} \index{configuration space} $Q$, which we take to be a manifold (that is infinite-dimensional in the case of field theories). However, some states turn out to be physically impossible. The criterion that a state is physically relevant can often be put into the form that it is a stationary point of a certain function $Q\ra{S}\R$, called the \emph{action}\index{action}. 

Very often, $Q$ will be a manifold of mappings $Q\subset C^\infty(M,N)$ (or $C^\infty(N,M)$!) and $S$ is given by an integral $S(q)=\int_M \mathscr{L}(q)$, where $q\in Q$ and $Q\ra{\mathscr{L}} \mathrm{measures}(M)$. Under locality assumptions\footnote{Namely, that $\mathscr{L}$ factors through the bundle of germs of maps from $M$ to $N$. For more details, see section \ref{actionlocal}.} satisfied in many situations\footnote{This has to do with the principle of locality in physics.}, the global demand that $S$ be stationary at $q$ with respect to all compactly supported variations translates into a local differential equation for $q$. Concretely, after integrating by parts one obtains an Euler-Lagrange expression $EL(q)$ such that
$$\frac{d}{dt}\Big|_{t=0}S(q+t\tau)=\int_M \langle EL(q),\tau\rangle$$
for all compactly supported variations $\tau\in T_qQ$, and the fundamental lemma of the calculus of variations then gives stationarity for all such $\tau$ if and only if $EL(q)=0$. This is called the \emph{Euler-Lagrange equation}\index{Euler-Lagrange equation}. We see that the global demand on $q$ that it be a stationary point of the action reduces to a demand that $q$ satisfies a certain differential equation locally (which we can hope to solve).

The idea is that it is easier to argue what the action (that we demand to be stationary) should be than to derive the Euler-Lagrange equation immediately. Take, for example, soap films. It is a well-known fact (and it can indeed be easily argued from a simple physical model) that soap films are stationary points of the action function that takes a soap film and computes its area. However, this is not a very workable criterion for determining which shape a soap film will take given some boundary conditions. We can solve this problem by writing out the corresponding Euler-Lagrange equation, which, in this case, states that at each point of the film the mean curvature\footnote{This is an extrinsic measure of curvature of an embedded submanifold of a Riemannian manifold. (In this case, the soap film is embedded in $\R^3$.) To be more specific, it is a constant multiple of the trace of the second fundamental form.} should be zero. Although it is much harder to guess that this equation should hold for soap films, it is just a partial differential equation, which we can hope to solve. Another example is provided by Fermat's principle, which states that light propagates through an inhomogeneous medium between two points along the path that minimises (actually, we only have that the path is a stationary point of the action) the transit time (which is the action, in this case). From this we can derive Euler-Lagrange equations for the path that we can indeed solve, but which are much less insightful. \cite{gelcal}

Let us now proceed to the particular case at hand: action principles on principal bundles in the context of gauge theories of the Yang-Mills kind. In my treatment of these action principles I follow \cite{blegau}.

\subsection{Mathematical Intermezzo: Jet Bundles}
To give a mathematically rigorous treatment of such an action principle, we need the concept of a \emph{jet bundle}\index{jet bundle}. Recall that given a (pre)sheaf, we can construct its \emph{bundle of germs}, which is a topological \'etal\'e space whose projection is a local homeomorphism. In the case where we started with a sheaf, we can interpret the original sheaf as the sheaf of sections of this space. (The sheaf of sections of this space is isomorphic to the original sheaf.) \cite{macshe} Let us start with a pair of smooth manifolds $M$ and $N$ and consider the sheaf $U\mapsto C^\infty(U,N)$ of smooth maps from open subsets of $M$ to $N$. We obtain a topological \'etal\'e space $\mathrm{germ}(M,N)\ra{}M; \;\mathrm{germ}_mf\mapsto m$ of germs of smooth maps.

Now, we define the bundle of $k$-jets of maps from $M$ to $N$ as a quotient of $\mathrm{germ}(M,N)$: we identify $\mathrm{germ}_m f\sim \mathrm{germ}_m g$ if $f(m)=g(m)$ and, after choosing charts around $m$ and $f(m)$, the coordinate representatives have the same partial derivatives at $m$ up to order $k$. This condition is independent of the chosen charts. Note that this equivalence relation respects the fibres of the bundle of germs. In this way we obtain a topological fibre bundle $J^k(M,N)\ra{}M$, called the jet bundle of order $k$. One can think of $J^k(M,N)$ in charts as recording $$\{\; (m,f(m),T_m f, \partial^2f_m,\ldots, \partial^kf_m)\; | \; U\subset M\txt{open},\; m\in U,\; f\in C^\infty(U,N)\, \}.$$ Note that if we replace $N$ with a vector space $W$, $J^k(M,W)$ has the structure of a vector bundle. Finally, it is easily checked that we can define a smooth structure on $J^k(M,N)$ using that on $M$ and $N$, making it into a smooth fibre bundle. Indeed, we construct a chart for $J^k(M,N)$ from each pair of charts for $M$ and $N$ in the obvious way.

Note that we can also start with a fibre bundle $F\ra{\pi}M$ and consider the bundle of germs of sections of this bundle. We can take a quotient of this to obtain a jet bundle $J^k\pi$ of sections of $\pi$.

\subsection{Action Densities and Lagrangians}
\label{actionlocal}
Suppose we want to deal with a particle whose wave functions are sections of a metrised vector bundle $V\ra{\pi}M$, with standard fibre an inner product space $W$. Write $P\ra{\pis}M$ for its orthonormal frame bundle and write $G$ for its structure group. Then we can equivalently take $C^\infty(P,W)^G \subset C^\infty(P,W)$ as the set of wave functions. What we want to work towards is a function $C^\infty(P,W)^G\ra{\mathscr{L}}C^\infty(M)$, called the action density, from which we will be able to derive the evolution equations of the system. There is a very important principle in physics, known as the \emph{principle of locality}\index{principle of locality}, which states that an event in spacetime can only be influenced by certain events that are sufficiently near it. Put differently, we should be able to formulate the laws of physics locally and make predictions on the basis of only local information. Of course, this is a very practical demand. This demand is represented mathematically by demanding that the action density factors through the bundle of germs of maps from $P$ to $W$. 

There is another, less fundamental\footnote{Indeed, a similar theory can be set up for the case in which the action density factors through the $n$-th jet bundle.}, and even more practical demand that we put on the action density. Ever since Newton, we have been used to formulating our physical theories in a form in which the equations of motion for the particles become second-order differential equations (ordinary or partial). This means that if we know the initial configuration of a particle and its first derivative (or momentum), we can predict its trajectory. Equivalently, in the Lagrangian formalism (which we shall use) of field theory, we demand that the action density depends only on the value of the fields and their first derivatives. Put differently, it does not merely factor through the bundle of germs, but even through the jet bundle $J^1(P,W)$.

An \emph{interaction Lagrangian}\index{interaction Lagrangian} will be a $G$-invariant\footnote{We define the Lagrangian to be an invariant function on the (jet bundle of) the principal bundle. Such a function corresponds precisely to a function on the (jet bundle of) equivariant maps from $P$ to $W$ (or by claim \ref{seccoo}, on the (jet bundle of) sections of the associated vector bundle, which is what physicists would call a Lagrangian instead). Note that the action on $J^1(P,W)$ for which we demand invariance is a natural one induced by the principal right action on $P$. Indeed, if $\psi$ is an equivariant map, then $(p\cdot g,\psi(p\cdot g),d\psi(p\cdot g))=(p\cdot g,g^{-1}\cdot \psi(p),g^{-1}\cdot d\psi(p)\circ T\rho_{g^{-1}})$.} smooth function $J^1(P,W)\ra{L_i}\R$, in the sense that (remember that $G$ acts on $W$ on the left) $L_i(p,w,\xi)=L_i(p\cdot g,g^{-1}\cdot w, g^{-1}\cdot \xi\circ T\rho_{g^{-1}})$, where I write $\rho$ for the principal right action and make the canonical identification of $T_w W\cong W$. 

This compatibility condition with the $G$-actions is chosen precisely so that an interaction Lagrangian defines a function $C^\infty(P,W)^G\ra{\mathscr{L}_n}C^\infty(M)$, called the \emph{naive interaction action density}\index{naive interaction action density}, by $\mathscr{L}_n(\psi)(m):=L_i(p,\psi(p),d\psi(p))$, where $p$ is an arbitrary point in $\pis^{-1}(m)$. (Here I write $d\psi$ for the derivative $T\psi$.)
\begin{claim}
In this way, $\mathscr{L}_n$ is well-defined.
\end{claim}
\begin{proof}
This will follow from the fact that the principal right action is transitive on the fibre. Indeed,
$g^{-1}\cdot\psi(p)=\psi(p\cdot g)$ and $g^{-1}\cdot d\psi(p)=d\psi(p\cdot g)\circ T\rho_g$. Therefore $L_i(p\cdot g,\psi(p\cdot g),d\psi(p\cdot g))=L_i(p\cdot g,g^{-1}\cdot \psi(p),g^{-1}\cdot d\psi(p)\circ T\rho_{g^{-1}})=L_i(p,\psi(p),d\psi(p))$, where the last equality holds by definition of the Lagrangian.\end{proof}
As the terminology `naive action density' suggests, there is a problem with this definition. It has to do with a failure of gauge invariance. Almost all interaction Lagrangians that arise in practice are \emph{$G$-invariant}\index{interaction Lagrangian, $G$-invariant} in the sense that $L_i(p,w,\xi)=L_i(p,g\cdot w, g\cdot \xi)$. If we want this principal-bundle setup of physics to do what I sketched in the introduction to this chapter, namely to formalise the idea that physics cannot be separated from observers and that all observers are equivalent, then the action density (which will give us our equations of motion of the particles) should be gauge-invariant. In particular, we would very much hope that $G$-invariant Lagrangians give rise to gauge-invariant action densities. This turns out not to be the case for our naive definition of an action density.
\begin{claim}
$\mathscr{L}_n$ is not necessarily gauge-invariant, if $L_i$ is a $G$-invariant Lagrangian. 
\end{claim}
\begin{proof}This proof and the next will make use of the interpretation of gauge transformations given in claim \ref{isomgauge}. Indeed, let $f\in\mathrm{Gau}\, P$. Then this claim gives us a unique $\tau\in C^\infty(P,G)$ such that $f(p)=p\cdot \tau(p)$. Then $f^* \psi=\tau^{-1}\cdot \psi$. Let $]-\epsilon,\epsilon[\ra{c}P$ such that $c(0)=p$ and $c'(0)=X$. Then
\begin{align*}d(\tau^{-1}\cdot \psi)(X)&=\frac d {dt}|_{t=0} \tau^{-1} (c(t))\cdot \psi(c(t))\\
&=\frac d {dt}|_{t=0} \tau^{-1}(p)\cdot \psi(c(t))+\frac{d}{dt}|_{t=0} \tau^{-1} (c(t))\cdot \psi(p)\\
&=\tau^{-1}(p)\cdot d\psi(X)+\frac d {dt}|_{t=0}\tau^{-1}(c(t))\tau(p)\tau(p)^{-1}\cdot \psi(p)\\
&=\tau^{-1}(p)\cdot d\psi(X)+ TR_{\tau(p)}(T_p\tau^{-1})(X)\cdot \tau(p)^{-1}\cdot \psi(p).
\end{align*}

Therefore
\begin{align*}
\mathscr{L}_n(f^*\psi)(m)&=L_i(p,(f^*\psi)(p),d(f^*\psi)(p))\\
&=L_i(p,\tau^{-1}(p)\cdot \psi(p),\tau^{-1}(p)\cdot d\psi(p)+TR_{\tau(p)}(T_p\tau^{-1})\cdot \tau(p)^{-1}\cdot \psi(p))\\
&\exeq L_i(p,\psi(p),d\psi(p))\\
&= L_i(p,\tau^{-1}(p)\cdot \psi(p),\tau^{-1}(p)\cdot d\psi(p)), \txt{(by $G$-invariance)}
\end{align*}
where the $\exeq$ is the condition for gauge invariance. In general, however, the term $TR_{\tau(p)}(T_p\tau^{-1})\cdot \tau(p)^{-1}\cdot\psi(p)$ does not vanish and therefore gauge invariance is not satisfied. 
\end{proof}
Up to this point I have tried to argue why principal bundles arise in physics. However, no real answer has been given to the question why physicists would be interested in principal connections. The connection in the Kaluza-Klein theory of electromagnetism arose as a technical curiosity rather than as an inevitable consequence of some fundamental demand on our physical theories. However, principal connections arise in a very natural way in gauge theory. In fact, we are forced to introduce them when we want our theory to be gauge-invariant!

Remember that we write $\mathscr{C}(P)\subset \Omega^1(P,\g)$ for the space of principal connections. Note that the right action of $\mathrm{Gau}\, P$ on $\Omega^1(P,\g)$ by pullback restricts to $\mathscr{C}(P)$. 
\begin{claim}
Let $f\in\mathrm{Gau}\, P$ and $\omega\in\mathscr{C}(P)$ then $f^*\omega\in\mathscr{C}(P)$.
\end{claim}
\begin{proof}Let $Z\in\g$ and $f\in\mathrm{Gau}\, P$. Then $(f^*\omega)(Z_P(p))=\omega(Tf Z_P(p))= \omega(\frac d {dt}|_{t=0} f(p\cdot \exp(tZ)))=\omega(\frac d {dt}|_{t=0} f(p)\cdot \exp(tZ))=\omega(Z_P(f(p)))=Z$, so $f^*\omega$ reproduces vertical vectors.

Moreover, since $f$ is a map of principal bundles, $\rho_g^*f^*\omega=f^*\rho_g^*\omega=f^*\mathrm{Ad}_{g^{-1}}\omega=\mathrm{Ad}_{g^{-1}} f^*\omega$, i.e. $f^*\omega$ is equivariant in the sense of claim \ref{claim:equiv}.
\end{proof}
Moreover, $\mathrm{Gau}\, P$ acts on $C^\infty(P,W)^G$ on the right by $f^*\psi(p):=\psi\circ f(p)$.
\begin{theorem}
Let $L_i$ be a $G$-invariant Lagrangian. If we define the \emph{interaction action density}\index{interaction action density} $C^\infty(P,W)^G\times \mathscr{C}(P)\ra{\mathscr{L}_i}C^\infty(M)$ by $\mathscr{L}_i(\psi,\omega)(m):=L_i(p,\psi(p),d^\omega\psi(p))$, for $p\in\pis^{-1}(m)$, then $\mathscr{L}_i$ is gauge-invariant in the sense that $\mathscr{L}_i(\psi,\omega)=\mathscr{L}_i(f^*\psi,f^*\omega)$, for all $f\in\mathrm{Gau}\, P$.
\end{theorem}
\begin{proof}We begin by checking that $\mathscr{L}_i$ is well-defined in this way. Note that by equivariance of $d^\omega \psi\in \Omega_\hor^1(P,W)^G$, we have $d^\omega\psi(p\cdot g)=g^{-1}\cdot d^\omega\psi(p)\circ T\rho_{g^{-1}}$. Therefore
\begin{align*}L_i(p\cdot g,\psi(p\cdot g),d^\omega\psi(p\cdot g))&=L_i(p\cdot g, g^{-1}\cdot \psi(p),g^{-1}\cdot d^\omega\psi(p)\circ T\rho_{g^{-1}})\\
&=L_i(p, \psi(p),d^\omega\psi(p)),
\end{align*}
where the last identity holds by definition of the interaction Lagrangian.

Now, we verify gauge invariance.
\begin{align*}\mathscr{L}_i(f^*\psi,f^*\omega)(m)&=L_i(p,(f^*\psi)(p),d^{f^*\omega}f^*\psi(p))\\
&=L_i(p,(f^*\psi)(p),f^*(d^\omega \psi)(p))\\
&=L_i(p,\tau(p)^{-1}\cdot \psi(p),\tau(p)^{-1}\cdot d^\omega\psi(p))\\
&=L_i(p,\psi(p),d^\omega\psi(p))\txt{(by $G$-invariance)}\\
&=\mathscr{L}_i(\psi,\omega)(m).
\end{align*}
\end{proof}
We see that the use of principal connections was forced on physicists if they wanted their theory to be gauge-invariant. Now, it turns out that these principal connections have a physical interpretation: that of gauge potentials for the interactions between particles. They are generalisations of the electromagnetic potential.

Now, we know from the theory of electromagnetism that the fields can exhibit non-trivial behaviour, even in the absence of matter. This means that we have to add a term to the action density that depends only on the gauge potential, if we want to formulate an action principle from which this behaviour can be derived. Let us write this term, called the \emph{self-action density}\index{self-action density} of the fields, as $\mathscr{L}_s$. This term will be a map $\mathscr{C}(P)\ra{\mathscr{L}_s}C^\infty(M)$, but we will interpret it as a map $ C^\infty(P,W)^G\times \mathscr{C}(P)\ra{\mathscr{L}_s}C^\infty(M)$. We will write $\mathscr{L}:=\mathscr{L}_i+\mathscr{L}_s$ for the total action density.

Again, $\mathscr{L}_s$ will only depend on the potential and its covariant derivative (i.e. the curvature) at a point. This can be made explicit by demanding that it comes from a \emph{self Lagrangian for the fields} $J^1\mathscr{C}(P)\ra{L_s}\R$, in a similar way as $\mathscr{L}_i$ did. (Here, $J^1\mathscr{C}(P)$ denotes the bundle, over $P$, of $1$-jets of germs of local principal connection forms; it is not the jet bundle of the global affine space $\mathscr{C}(P)$.) 

We have a natural action of $\mathrm{Gau}\, P$ on $J^1\mathscr{C}(P)$. Indeed, we take $f\in\mathrm{Gau}\, P$ and $\tau=(p,\omega(p),d\omega(p))\in J^1\mathscr{C}(P)$ and set $f\cdot\tau:=(f^{-1}(p),(f^*\omega)(f^{-1}(p)),d(f^*\omega)(f^{-1}(p)))$. Moreover, $G$ acts on $J^1\mathscr{C}(P)$ by $(p,\omega(p),d\omega(p))\cdot g:=(p\cdot g,\omega(p\cdot g),d\omega(p\cdot g))$.

A \emph{self Lagrangian for the fields}\index{self Lagrangian for the fields} is then defined to be a map $J^1\mathscr{C}(P)\ra{L_s}\R$ such that $L_s(\tau\cdot g)=L_s(\tau)$, for all $\tau\in J^1\mathscr{C}(P)$ and $g\in G$. We now define $\mathscr{L}_s$ by $\mathscr{L}_s(\pis(p)):=L_s(\tau_p)$, where $\tau_p=(p,\omega(p),d\omega(p))\in (J^1\mathscr{C}(P))_{p}$. This is well-defined since the principal right action is transitive on the fibre and $L_s(\tau_p)=L_s(\tau_p\cdot g)=L_s(\tau_{p\cdot g})$. Moreover, we have the following.
\begin{claim}
$\mathscr{L}_s$ is gauge-invariant if and only if $L_s$ is.
\end{claim}
\begin{proof}The statement follows from comparison of the following two equalities.\\
$\mathscr{L}_s(\omega)(\pis(p))=L_s(p,\omega(p),d\omega(p))$ and \begin{align*}\mathscr{L}_s(f^*\omega)(\pis(p))&=\mathscr{L}_s(f^*\omega)(\pis(f^{-1}(p)))\\
&=L_s(f^{-1}(p),(f^*\omega)(f^{-1}(p)),d(f^*\omega)(f^{-1}(p)))\\
&=L_s(f\cdot (p,\omega(p),d\omega(p))),
\end{align*}
for $f\in\mathrm{Gau}\, P$.
\end{proof}
Note that the curvature $\mathscr{C}(P)\ra{\Omega}\Omega_\hor^2(P,\g)^G$ depends locally on the connection in the sense that it is derived from a function $\mathrm{germ}\,\mathscr{C}(P)\ra{\tilde \Omega}\mathrm{germ}\,\Omega^2(P,\g)$. This map even factors through a map $J^1\mathscr{C}(P)\ra{\hat\Omega}J^0(\Omega_\hor^2(P,\g)^G)$. (Here, $J^0\Omega_\hor^2(P,\g)^G$ denotes the zeroth jet bundle of the bundle of germs of sections of $\Omega_\hor^2(P,\g)^G$.) Indeed, $\hat\Omega(\tau)=d\omega(p)+\frac 1 2 [\omega(p),\omega(p)]=\Omega^\omega(p)$, for $\omega\in\mathscr{C}(P)$ and $\tau=(p,\omega(p),d\omega(p))$. 

This observation is important to understand the following theorem, which is due to Utiyama and gives a characterisation of the gauge-invariant self Lagrangians.
\begin{theorem}[Utiyama's theorem\index{Utiyama's theorem}] A self Lagrangian $L_s$, or equivalently its self-action density $\mathscr{L}_s$, is gauge-invariant if and only if it is of the form $L_s=K\circ\hat\Omega$ for some $\mathrm{Ad}$-invariant $J^0(\Omega_\hor^2(P,\g)^G)\ra{K}\R$.\end{theorem}
A simple proof can be found in \cite{blegau}. What does this theorem tell us? It says that a self Lagrangian is gauge-invariant precisely if its value at a point only depends on the value of the curvature at that point and also depends on it in a $\mathrm{Ad}$-invariant way. Obviously, this result shows that the possibilities for gauge-invariant self Lagrangians are very limited.

\subsection{Actions and the Euler-Lagrange Equations}
We have been talking about action principles for a while, but up to this point we have not been very formal about them. With the machinery of the preceding paragraphs we are in a position to make them mathematically sound.

We will define a function $C^\infty (P,W)^G\times \mathscr{C}(P)\ra{S^U} \R$ for each open $U\subset M$ that has compact closure. These functions will be called the \emph{action functionals}. We will show that the action $S^U$ is stationary for some $(\psi,\omega)\in C^\infty(P,W)^G\times \mathscr{C}(P)$ if and only if $(\psi,\omega)$ satisfies a set of partial differential equations on $\pis^{-1}(U)$, called the \emph{Euler-Lagrange equations} for this action.

Let us fill in the details. We define $S^U(\psi,\omega):=\int_U \mathscr{L}(\psi,\omega)\mu$, where $\mu$ is the volume form on $M$ that we obtain from the spacetime metric\footnote{Note that $\mu$ does not need to exist globally. Indeed, not all spacetimes are orientable. However, every spacetime is locally orientable. We can therefore find a neighbourhood $U$ for each point $m\in M$, such that we can choose a Lorentzian volume form on $U$. Then we can define the action over $U$ and derive the Euler-Lagrange equation, which we find not to depend on the choice of orientation. However, if this technicality bothers you, assume that $M$ has been given a global orientation! (This can be done for many important spacetimes. \cite{walgen})}. Now, we say that $S^U$ is \emph{stationary}\index{stationary action} at $(\psi,\omega)$, if for all $(\sigma,\tau)\in C^\infty(P,W)^G\times \Omega_\hor^1(P,\g)^G$ with projected support contained in $U$,
$$\frac d {dt}|_{t=0} S^U(\psi+t\sigma,\omega+t\tau)=0.$$
We want to show that this is equivalent to $(\psi,\omega)$ satisfying the Euler-Lagrange equations.

To write down the Euler-Lagrange equation for the particle field, we need some notation. I will use the convention of \cite{blegau}, although I do not think it is standard. Let $(p,w,\theta)\in J^1(P,W)$ and let $L_i$ be an interaction Lagrangian. Note that $w'\mapsto \frac{d}{dt}|_{t=0}L(p,w+tw',\theta)$ is a linear map $W\ra{}\R$. By non-degeneracy of the inner product $k_W$ on $W$, this defines a unique element $\nabla_2 L_i(p,w,\theta)\in W$, such that $k_W(\nabla_2L_i(p,w,\theta),w')=\frac{d}{dt}|_{t=0}L(p,w+tw',\theta)$, for all $w'\in W$. Similarly, note that $\theta'\mapsto \frac{d}{dt}|_{t=0} L_i(p,w,\theta+t\theta')$ is a linear map $HP_p^*\otimes W\ra{}\R$. By non-degeneracy of the inner product $(\bar h k_W)_p$ on $HP_p^*\otimes W$ (the space of $W$-valued $1$-forms, vanishing on vertical vectors), we obtain a unique $\nabla_3L_i(p,w,\theta)\in HP_p^*\otimes W$, such that $\bar h k_W(\nabla_3L_i(p,w,\theta),\theta')=\frac{d}{dt}|_{t=0} L_i(p,w,\theta+t\theta'))$, for all $\theta'\in HP_p^*\otimes W$. Finally, we write \label{moeilijkeafgeleide} $\frac{\dau L_i}{\dau \psi}(p):=\nabla_2 L_i(p,\psi(p),d^\omega\psi(p))$ and $\frac{\dau L_i}{\dau(d^\omega\psi)}(p):=\nabla_3 L_i(p,\psi(p),d^\omega\psi(p))$. Obviously, these define a function in $C^\infty(P,W)$ and an element of $\Omega^1(P,W)$, vanishing on vertical vectors, respectively. However, the definitions even turn out to be $G$-equivariant.
\begin{claim}With these definitions, $\frac{\dau L_i}{\dau \psi}\in C^\infty(P,W)^G=\Omega_\hor^0(P,W)^G$ and $\frac{\dau L_i}{\dau(d^\omega\psi)}\in\Omega_\hor^1(P,W)^G$.
\end{claim}
\begin{proof}We verify equivariance of $\frac{\dau L_i}{\dau \psi}$. The computation for $\frac{\dau L_i}{\dau(d^\omega\psi)}$ is almost identical and it is left to the reader.

Let $w'\in W$. Then \begin{align*}k_W(\frac{\dau L_i}{\dau \psi}(p\cdot g),w')&= \frac{d}{dt}|_{t=0} L_i(p\cdot g,\psi(p\cdot g)+tw',d^\omega\psi(p\cdot g))\\
&=\frac{d}{dt}|_{t=0} L_i(p\cdot g,g^{-1}\cdot\psi(p)+tw',g^{-1}\cdot d^\omega\psi(p)\circ T\rho_{g^{-1}})\\
&=\frac{d}{dt}|_{t=0} L_i(p,\psi(p)+tg\cdot w', d^\omega\psi(p))\txt{(by definition of $L_i$)}\\
&=k_W(\frac{\dau L_i}{\dau \psi}(p),g\cdot w')\\
&=k_W(g^{-1}\cdot\frac{\dau L_i}{\dau \psi}(p), w').\txt{(since $k_W$ is an invariant inner product)}\\
\end{align*}
Non-degeneracy of $k_W$ does the rest.
\end{proof}
With this notation, the Euler-Lagrange equation for the particle field takes the following simple form.
\begin{theorem}[Euler-Lagrange equation for the particle field]\index{Euler-Lagrange equation for the particle field}
The action is stationary with respect to the particle field, $\frac{d}{dt}|_{t=0}S^U(\psi+t\sigma,\omega)=0$, for all open $U\subset M$ with compact closure and $\sigma\in C^\infty(P,W)^G$ with projected support in $U$, if and only if $\psi$ satisfies the Euler-Lagrange equation $$\delta^\omega\left(\frac{\dau L_i}{\dau (d^\omega\psi)}\right)+\frac{\dau L_i}{\dau\psi}=0.$$
\end{theorem}
\begin{proof} We compute
\begin{align*}\int_U\frac d {dt}|_{t=0} \mathscr{L}_i(\psi+t\sigma,\omega)(\pis(p))\mu&=\int_U\frac d {dt}|_{t=0} L_i(p,\psi(p)+t\sigma(p),d^\omega \psi(p)+t d^\omega \sigma(p))\mu\\
&=\int_U\frac d {dt}|_{t=0} L_i(p,\psi(p)+t\sigma(p),d^\omega \psi(p))\mu\\
&+\int_U\frac d {dt}|_{t=0} L_i(p,\psi(p),d^\omega \psi(p)+t d^\omega \sigma(p))\mu\\
&=\int_Uk_W(\frac{\dau L_i}{\dau \psi}(p),\sigma(p))\mu+\int_U\bar h k_W(\frac {\dau L_i}{\dau(d^\omega\psi)}(p),d^\omega\sigma(p))\mu\\
&=\int_Uk_W(\frac{\dau L_i}{\dau \psi}(p),\sigma(p))\mu+\int_U\bar h k_W(\delta^\omega \frac {\dau L_i}{\dau(d^\omega\psi)}(p),\sigma(p))\mu\\
&=\int_Uk_W(\delta^\omega \frac {\dau L_i}{\dau(d^\omega\psi)}(p)+\frac{\dau L_i}{\dau \psi}(p),\sigma(p))\mu,
\end{align*}
where, in the last equality we use that $\bar h k_W=k_W$ on $0$-forms and in the equality before we use claim \ref{cl:codifadj}. Thus, we find that $\frac d {dt}|_{t=0}S^U(\psi+t\sigma,\omega)=0$ iff $\int_U k_W(\delta^\omega\frac{\dau L_i}{\dau(d^\omega\psi)}(p)+\frac{\dau L_i}{\dau \psi}(p),\sigma(p))\mu=0$. Now, the demand that this holds for all open $U\subset M$ with compact closure is equivalent to the demand that the integrand be zero (since the integrand is continuous). The non-degeneracy of $k_W$ does the rest.
\end{proof}
We want to demand that the action be stationary with respect to the fields as well and derive a similar equation. To do this, we have to introduce the concept of a \emph{current}\index{current}, which will play an analogous role to that of the electric current in electromagnetism. Note that my treatment, and therefore the terminology as well, is non-standard. 

Let us note that $T_\omega \mathscr{C}(P)=\Omega_\hor^1(P,\g)^G$. The map $T_\omega\mathscr{C}(P)\ra{}C^\infty(M)$ given by $\tau\mapsto \frac{d}{dt}|_{t=0}\mathscr{L}_i(\psi,\omega+t\tau)$ is linear in $\tau$. This means that, pointwise and by non-degeneracy of $\bar h k_\g$, we have a unique $J(\psi,\omega)\in \Omega_\hor^1(P,\g)^G$ such that $\frac{d}{dt}|_{t=0}\mathscr{L}_i(\psi,\omega+t\tau)=\bar h k_\g(J(\psi,\omega),\tau)$ for all $\tau\in T_\omega \mathscr{C}(P)$. We shall call this $J(\psi,\omega)$ the \emph{interaction current}\index{interaction current}\footnote{This quantity is usually called simply \emph{the} current in the literature.}. Since $\mathscr{L}_s$ may depend on the first jet of $\omega$, its variation need not be pointwise algebraic in $\tau$. We therefore define (following the terminology in \cite{aldint}) a \emph{self current}\index{self current} of the gauge field as the unique $\mathscr{J}(\omega)\in\Omega_\hor^1(P,\g)^G$ such that
$$-\frac d {dt}\Big|_{t=0}\int_U\mathscr{L}_s(\omega+t\tau)\mu=\int_U\bar h k_\g(\mathscr{J}(\omega),\tau)\mu$$
for all open $U\subset M$ with compact closure and all $\tau\in T_\omega \mathscr{C}(P)$ with projected support in $U$, whenever such a form exists. These definitions enable us to formulate the second Euler-Lagrange equation.

\begin{theorem}[Euler-Lagrange equation for the gauge field]\index{Euler-Lagrange equation for the gauge field} The action is stationary at $(\psi,\omega)$ with respect to the gauge field, $\frac{d}{dt}|_{t=0}S^U(\psi,\omega+t\tau)=0$, for all open $U\subset M$ with compact closure and $\tau\in T_\omega\mathscr{C}(P)$ with projected support in $U$, if and only if $\omega$ satisfies the Euler-Lagrange equation
$$\mathscr{J}(\omega)=J(\psi,\omega).$$
\end{theorem}
\begin{proof}Using the definitions of the interaction and self currents, we compute
\begin{align*}
\frac{d}{dt}\Big|_{t=0}S^U(\psi,\omega+t\tau)
&=\frac{d}{dt}\Big|_{t=0}\int_U\mathscr{L}_i(\psi,\omega+t\tau)\mu+\frac{d}{dt}\Big|_{t=0}\int_U\mathscr{L}_s(\omega+t\tau)\mu\\
&=\int_U\bar h k_\g(J(\psi,\omega)-\mathscr{J}(\omega),\tau)\mu.
\end{align*}
Again, by continuity of the integrand, we see that $\bar h k_\g(J(\psi,\omega)-\mathscr{J}(\omega),\tau)$ is zero if and only if $\frac{d}{dt}|_{t=0}S^U(\psi,\omega+t\tau)=0$ for all $U$. Now, the non-degeneracy of $\bar h k_\g$ tells us that this holds for all $\tau$ precisely if $J(\psi,\omega)-\mathscr{J}(\omega)=0$.
\end{proof}
One may wonder why these quantities are called currents. For one thing, as we shall see, in the case of electromagnetism this equation reduces to the second (inhomogeneous) Maxwell equation and $J(\psi,\omega)$ will equal the electric current. Another reason for this nomenclature may be the following variant of Noether's theorem\footnote{Actually, Noether's theorem can be proved in a much broader context. Indeed, we obtain a similar conserved quantity if the Lagrangian is invariant under an arbitrary $G$-action. See, for example, \cite{jantra}}.
\begin{theorem}[Noether's theorem]\index{Noether's theorem} Let $L_i$ be a $G$-invariant interaction Lagrangian and suppose that $(\psi,\omega)$ satisfies the Euler-Lagrange equation for the particle field, then $\delta^\omega(J(\psi,\omega))=0$.\end{theorem}
This is a kind of continuity equation for $J(\psi,\omega)$. Apparently, $J(\psi,\omega)$ represents the current density of some quantity that is locally conserved in spacetime. The proof of this theorem requires some more definitions. Therefore, the interested reader is referred to \cite{blegau}. However, as we shall see, if we have a special form of the self-action, called the Yang-Mills action, we immediately obtain that $\delta^\omega\mathscr{J}(\omega)=0$, even if the interaction Lagrangian is not $G$-invariant. Application of the Euler-Lagrange equation for the gauge field again gives us conservation of the interaction current. We see that we obtain an interpretation of $J(\psi,\omega)$ as a current if either the interaction Lagrangian or the self Lagrangian has a particularly simple form.\\
\\
Summarising, we have found that $S^U(\psi,\omega)$ is stationary with respect to all compactly supported variations in $(\psi,\omega)$ if and only if $\psi$ and $\omega$ satisfy both Euler-Lagrange equations:
$$\delta^\omega\left(\frac{\dau L_i}{\dau (d^\omega\psi)}\right)+\frac{\dau L_i}{\dau\psi}=0\txt{and}\mathscr{J}(\omega)=J(\psi,\omega).$$
(Specific cases of these two equations will be the Dirac equation and the second (inhomogeneous) Maxwell equation, respectively.) As in the case of Maxwell's electromagnetism, we can hope that these two equations, together with the integrability condition provided by the Bianchi identity,
$$d^\omega\Omega^\omega=0,$$
will uniquely determine the evolution of the fields. (Of course, the first Maxwell equation is a specific example of this Bianchi identity.)

\subsection{The Yang-Mills Action}
We have seen that we can interpret Maxwell's electromagnetism as a $U(1)$-gauge theory. From this point of view, the electromagnetic tensor acquired the interpretation of (a multiple of) the curvature $\Omega^\omega$ of some connection $\omega$ ($F=i e\,e_\alpha^*\Omega^\omega$ in a local gauge). We have seen that the Maxwell equations can be derived from an action density $-\frac 1 2 \mathrm{tr}^{1,3}\mathrm{tr}^{2,4} F\otimes F$, which is proportional to $-\frac{1}{2}\bar h k_{\mathfrak{u}(1)}(\Omega^\omega,\Omega^\omega)$ under the normalisation $F=i e\,e_\alpha^*\Omega^\omega$. This turns out to be a very general form of the self-action density for the gauge field. Indeed, it can be used to describe the weak, the unified electroweak, and the strong interactions as well\footnote{In the case of the (electro)weak interaction, this requires some alterations, such as a Higgs mechanism (which can also be understood in this formalism for gauge theory). \cite{blegau}}, albeit on a different principal bundle. For instance, one uses an $SU(2)$-bundle for the weak interaction, a $U(1)\times SU(2)$-bundle for the electroweak interaction, and an $SU(3)$-bundle for the strong one. We therefore write out the Euler-Lagrange equation for the gauge field for this specific form of the self-action density. This leads to the famous \emph{Yang-Mills equation}\index{Yang-Mills equation}.

For this derivation, we need the following result.
\begin{claim}\label{cl:lemmayangmills}If $\tau\in \Omega_\hor^1(P,\g)^G$, then $d^\omega \tau=d\tau+[\omega,\tau]$.
\end{claim}
\begin{proof}Let $X,Y\in \mathcal{X}(P)$. We should verify that $(*)$:\\ $d^\omega\tau(X,Y)=d\tau(P_H X,P_H Y)=d\tau(X,Y)+
[\omega,\tau](X,Y)=d\tau(X,Y)+[\omega(X),\tau(Y)]-[\omega(Y),\tau(X)]$. We use bilinearity of both sides to reduce the equality to the case in which $X$ is either horizontal or vertical and the same for $Y$.

If $X$ and $Y$ are both horizontal, then the equality obviously holds, since both sides are equal to $d\tau(X,Y)$. 

Moreover, if both $X$ and $Y$ are vertical, $(*)$ becomes the condition that $d\tau(X,Y)=0$. Now, $d\tau(X,Y)=\Li_X\tau(Y)-\Li_Y\tau(X)-\tau([X,Y])=-\tau([X,Y])$. To see that this equals zero, note that the vertical bundle is integrable (claim \ref{cl:vertb}), so $[X,Y]$ is a vertical vector field.

Finally, if $X$ is vertical and $Y$ is horizontal, we should verify that $0=d\tau(X,Y)+[\omega(X),\tau(Y)]$. Now, $d\tau(X,Y)=\Li_X\tau(Y)-\Li_Y\tau(X)-\tau([X,Y])=\Li_X\tau(Y)-\tau([X,Y])$, so $(*)$ reduces to $\Li_X\tau(Y)-\tau([X,Y])=-[\omega(X),\tau(Y)]$. We verify this equality pointwise. Since both sides are tensorial, it is enough at $p$ to take $X=Z_P$, where $Z\in\g$ is chosen so that $Z_P(p)=X(p)$. Then
\begin{align*}\Li_X (\tau(Y))(p)&=\frac d {dt}\Big|_{t=0}\tau_{\rho_{\exp tZ}(p)}\bigl(Y(\rho_{\exp tZ}(p))\bigr)\\
&=\frac d {dt}\Big|_{t=0} \left(\mathrm{Ad}_{\exp(-tZ)}\circ \tau_p\right)\left( T\rho_{\exp -tZ} \circ Y\circ \rho_{\exp tZ}(p)\right)\txt{(by equivariance of $\tau$)}\\
&=\left(\frac d {dt}\Big|_{t=0} \mathrm{Ad}_{\exp(-tZ)}\circ \tau_p\right)(Y(p))+\tau_p\left(\frac{d}{dt}\Big|_{t=0}T\rho_{\exp -tZ} \circ Y\circ \rho_{\exp tZ}(p) \right)\\
&=-\mathrm{ad}(Z)(\tau(Y))(p)+\tau([X,Y])(p)\\
&=-[Z,\tau(Y)(p)]+\tau([X,Y])(p)\\
&=-[\omega(X)(p),\tau(Y)(p)]+\tau([X,Y])(p)\txt{(by claim \ref{claim:repro})}.
\end{align*}
\end{proof}
The Yang-Mills equation now easily follows.
\begin{theorem}[The Yang-Mills equation] For the \emph{Yang-Mills self-action density}\index{Yang-Mills self-action density} $\mathscr{L}_s(\omega)=-\frac 1 2\bar h k_\g(\Omega^\omega,\Omega^\omega)$, the Euler-Lagrange equation for the gauge field takes the form
$$\delta^\omega \Omega^\omega=J(\psi,\omega).$$
\end{theorem}
\begin{proof}We assume that $U\subset M$ is open with compact closure and $\tau\in\Omega_\hor^1(P,\g)^G$ with projected support in $U$.

By theorem \ref{maurer}, $\Omega^{\omega+t\tau}=d(\omega+t\tau)+\frac 1 2 [\omega+t\tau,\omega+t\tau]=d\omega+\frac 1 2[\omega,\omega]+t(d\tau+\frac 1 2 [\tau,\omega]+\frac 1 2[\omega,\tau])+\mathcal{O}(t^2)=d\omega+\frac 1 2[\omega,\omega]+td^\omega \tau+\mathcal{O}(t^2)$, where the last equality holds by claim \ref{cl:lemmayangmills}. We see that $\frac{d}{dt}|_{t=0}\Omega^{\omega+t\tau}=d^\omega \tau$. Therefore $\int_U\bar h k_\g (\mathscr{J}(\omega),\tau)\mu=-\int_U\frac{d}{dt}|_{t=0}\mathscr{L}_s(\omega+t\tau)\mu=\int_U\frac{d}{dt}|_{t=0}\frac 1 2\bar h k_\g(\Omega^{\omega+t\tau},\Omega^{\omega+t\tau})\mu=\int_U\bar h k_\g(\Omega^\omega,d^\omega \tau)\mu=\int_U\bar h k_\g(\delta^\omega\Omega^\omega,\tau)\mu$, where the last equality holds by claim \ref{cl:codifadj}. Finally, the continuity of the integrand and non-degeneracy of $\bar h k_\g$ allow us to conclude that $\mathscr{J}(\omega)=\delta^\omega \Omega^\omega$.
\end{proof}
We see that this equation reduces to the second (inhomogeneous) Maxwell equation in the case $G=U(1)$. It is well known that electric charge conservation can be derived from this Maxwell equation. We obtain a similar result for the Yang-Mills equation.

\begin{corollary}[Generalised charge conservation]\index{charge conservation}  If $(\psi,\omega)$ satisfies the Yang-Mills equation, we have
$$\delta^\omega J(\psi,\omega)=0.$$
\end{corollary}
\begin{proof}
This is immediate: $\delta^\omega J(\psi,\omega)=\delta^\omega(\delta^\omega(\Omega^\omega))=\pm \bar * d^\omega \bar * \bar * d^\omega \bar * \Omega^\omega=\pm\bar * d^\omega d^\omega \bar * \Omega^\omega=\pm \bar * [\Omega^\omega,\bar * \Omega^\omega]=0$. For the last equality, write $\Omega^\omega=\sum_a E_a\Omega^a$ in a basis of $\g$; then $\Omega^a\wedge \bar *\Omega^b$ is symmetric in $a,b$, while $[E_a,E_b]$ is antisymmetric.
\end{proof}

\begin{remark}
The logical continuation of this thesis would be the case of a specific interaction Lagrangian, called the \emph{Dirac interaction Lagrangian}, that can be used to describe matter fields like electrons, nucleons (protons and neutrons), quarks and many other particles, which then obey the \emph{Dirac equation}\index{Dirac equation}. Its importance in physics can be compared to that of the Yang-Mills self Lagrangian for the gauge fields. However, there is no room for this here. The reader can find a treatment for the case of the electron and the nucleons with electroweak interaction in \cite{blegau}. Unfortunately, this work does not feature a more general treatment involving general spinor fields obeying a Dirac equation. In particular, an account of the strong interactions between quarks is missing. However, I have yet to find a better text on the subject. In particular, most other texts, even when they feature a theory of the strong interaction, do not take a geometric approach. Usually they are full of phenomenology and computations, proceeding immediately to the mathematically problematic procedure of second quantisation. An introduction to second quantisation can be found in \cite{fadgau}.
\end{remark}

\subsection{The Einstein-Yang-Mills Self-Action Density}
The bundle metric $h$ can be seen to be of even more physical importance than we have seen so far by computing its curvature scalar. This lengthy (but straightforward) computation, which I shall omit (see \cite{blegau}), gives the following theorem.
\begin{theorem}\label{thm:scacurv}
Let us write $\mathcal{S}_h$ for the curvature scalar of $h$ (which, by remark \ref{rem:isom}, is constant along fibres of $\pis$ and therefore descends to a function on $M$), $\mathcal{S}_g$ for that of $g$ and $\mathcal{S}_k$ for that of $k$ (which is constant, since, by invariance of $k$, $G$ acts transitively on itself by $k$-isometries by left multiplication). Then, we have an equality of smooth functions on $M$: 
$$\mathcal{S}_h=\mathcal{S}_g-\frac{1}{2}(gk_\g)(\Omega^\omega,\Omega^\omega)+\mathcal{S}_k.$$
\end{theorem}
We see that this action density involves contributions from both a Yang-Mills action density and the Einstein-Hilbert action density. Up to the constant term $\mathcal{S}_k$ (which is zero for Abelian one-dimensional fibres, and otherwise contributes a cosmological-constant term when $g$ is varied), this means that the corresponding action functional is stationary at values $(g,\omega)$ of the spacetime metric and the principal $G$-connection, respectively, if and only if $g$ satisfies the Einstein equation with the stress-energy tensor of the Yang-Mills field and $\omega$ satisfies the source-free Yang-Mills equation.

\section{Gravity on the Frame Bundle?}
We have seen that there exists a parallel between the equations of Einstein's general relativity and those of Yang-Mills theories. Indeed, the field is determined by a field equation involving sources (the Einstein and Yang-Mills equations, respectively) and an integrability condition in the form of a Bianchi identity. This parallel is supported by theorems \ref{thm:genchar} and \ref{thm:scacurv}, which interpret the Yang-Mills fields as pseudo-Riemannian metrics and indeed put source-free Yang-Mills theories in a framework that closely resembles that of general relativity.

Since there is a vast theory of the quantisation of Yang-Mills fields, one might want to do the opposite. Would it be possible to formulate the general theory of relativity as a gauge theory on a principal bundle? If so, this could help lead to a quantum theory of gravity.

Using the equivalence in corollary \ref{metrequi}, we have an obvious candidate for this principal bundle: the bundle of Lorentz frames (a principal $O(1,3)$-bundle\footnote{Perhaps it would be better to start with the Einstein-Cartan theory of gravity and construct the bundle of spin frames, which would be a $SL(2,\C)$-principal bundle.}). Indeed, we can formulate general relativity on this bundle by interpreting tensors as equivariant maps from the bundle of Lorentz frames, using claim \ref{seccoo} (and noting that by section \ref{catconstrafcon} each tensor bundle is associated with the bundle of Lorentz frames). Moreover, by claim \ref{afconind} we can transfer the Levi-Civita connection to a unique principal connection on the bundle of Lorentz frames\footnote{In fact, this claim would let us transfer it to a unique principal connection on the bundle of all frames. However, it is easily seen that this restricts to a connection on the bundle of Lorentz frames, since the horizontal bundle is tangent to the bundle of Lorentz frames. (This property together with the demand that it has zero torsion even uniquely characterises the Levi-Civita connection. \cite{blegau})}. Of course, one can write down the Einstein equation in this setting to obtain a principal bundle formulation of general relativity.

However, for unification purposes, this is not entirely satisfactory. Three of the four fundamental interactions in our current picture of nature can (modulo Higgs mechanisms in the case of the weak and electroweak interactions) be described by a Yang-Mills equation, but the equation for gravity has an entirely different form. Moreover, gravity appears in an entirely different way in the gauge theory than the other interactions. Indeed, it is not only described by the Levi-Civita connection (as a kind of gravitational potential) and its curvature (as a sort of field strength), but it is also present as the metric on the base space of the principal bundle, governing the geometry of spacetime. It is clear that the fundamental forces of nature are still far from being on an equal footing.

\clearpage
\addcontentsline{toc}{chapter}{Index}
\printindex 
\clearpage
\noindent\\
\\
\vspace{15pt}\\ \Huge \textbf{Glossary of Symbols} \normalsize
\addcontentsline{toc}{chapter}{Glossary of Symbols}\\
\quad\\
\\
\quad\\
\begin{tabularx}{\textwidth}{llX}
$\mathrm{Ad}$ && Adjoint representation of a Lie group on its Lie algebra\\
$CC^G$ && Category of $G$-cocycles\\
$CFB$ && Category of coordinate fibre bundles\\
$CFB^\lambda$ && For an action $G\ra{\lambda}\mathrm{Diff}\, S$, the category of coordinate $(G,\lambda)$-fibre bundles\\
$C^\infty(P,S)^G$ && Space of $G$-equivariant mappings (in the sense that, for $\Phi\in C^\infty(P,S)^G$, $\Phi(p\cdot g)=g^{-1}\cdot \Phi(p)$) from a principal bundle $P\ra{\pis}M$ to a manifold $S$, on which we have a left $G$-action\\
$\mathscr{C}(P)$ && Space of principal connections on $P$\\
$D/dt$ && For a path $I\ra{c}M$ and a vector bundle $V\ra{\pi}M$ with a connection $P_V$, the covariant derivative associated with the pullback connection $c^*(P_V)$\\
$\delta^\omega$ && For a principal connection $\omega$, the covariant codifferential\\
$d^\omega$ && For a principal connection $\omega$, the covariant differential\\
$FB$ && Category of fibre bundles\\
$FB^\lambda$ && For an action $G\ra{\lambda}\mathrm{Diff}\, S$, the category of $(G,\lambda)$-fibre bundles\\
$F^\lambda$ && Generalised frame bundle functor: $FB^\lambda\ra{F^\lambda}PB^G$, for an effective $G$-action $\lambda$\\
$F_m$ && For a fibre bundle $F\ra{\pi}M$, the fibre $\pi^{-1}(\{m\})$ over $m$\\
$\g$ && Lie algebra of a Lie group $G$\\
$\Gamma(F\ra{\pi}M)$ & & The sheaf of sections of a fibre bundle $F\ra{\pi}M$\\
$\bar h_f$ & & In the case of a fibre bundle $F\ra{\pi}M$ with a connection, where $M$ has a pseudo-Riemannian metric $g$, the inner product on $HF_f$ given by $\bar h_f(X,Y)=g_{\pi(f)}(T_f\pi(X),T_f\pi(Y))$\\
$HF\ra{}F$ && Horizontal bundle of $F$, if $F\ra{\pi}M$ is a fibre bundle with a connection\\
$\bar h k_W$ && See page \pageref{vreselijkip}\\
$J^k(M,N)\ra{}M$ && Bundle of $k$-jets of maps from $M$ to $N$\\
$K$ && Connector of some connection\\
$J^kF\ra{}M$ && Bundle of $k$-jets of sections of a fibre bundle $F\ra{}M$\\
$k_W$ && For a vector space $W$ with a linear $G$-action, an invariant inner product\\
$\Bbbk$ && Denotes either $\R$ or $\C$\\
$\Bbbk^{p,q}$ && $\Bbbk^k$ with the standard $(p,q)$-inner product\\
$L$ && Left action of a Lie group on itself by left multiplication\\
$\Li$ && Lie derivative\\
$-[\lambda]$ && Associated bundle functor: $PB^G\ra{-[\lambda]}FB^\lambda$, for a left $G$-action $\lambda$ on some manifold $S$
\end{tabularx}
\newpage
\begin{tabularx}{\textwidth}{llX}
${\dau L_i}/{\dau\psi}$ && See page \pageref{moeilijkeafgeleide}\\
${\dau L_i}/{\dau(d^\omega\psi)}$ && See page \pageref{moeilijkeafgeleide}\\
$MVB^k_{\Bbbk,(p,q)}$ && Category of $k$-dimensional $\Bbbk$-vector bundles that are equipped with a $(p,q)$-metric and maps of vector bundles that are isometries\\
$\nabla$ & & Covariant derivative operator\\
$\Omega_\hor(P,W)^G$ && $G$-equivariant, horizontal $W$-valued differential forms on a principal bundle with an equivariant connection\\
$\Omega(M,W)$ && Sheaf of $W$-valued differential forms, for a vector space $W$ and a manifold $M$\\
$\Omega^\omega$ && Curvature $2$-form of a connection $\omega$\\
$PB$ & &  Category of principal bundles \\
$PB^G$ & & Subcategory of $PB$ with bundles with a fixed structure group $G$ and maps that are the identity on $G$\\
$PB_M$  & & Subcategory of $PB$ with bundles over a fixed base space $M$ and maps that are the identity on $M$ \\
$P_H$ && Projection onto horizontal bundle corresponding to a connection\\
$\pi_{\bullet\times\Bbbk^k}$ && Projection $T\Bbbk^k\cong \Bbbk^k\times\Bbbk^k\ra{\pi_{\bullet\times\Bbbk^k}}\Bbbk^k$; $(x,y)\mapsto y$, where the second $\Bbbk^k$ represents the tangent space\\
$P_V$ && Projection onto vertical bundle corresponding to a connection\\
$R$ && Right action of a Lie group on itself by right multiplication\\
$\rho$ && Often denotes the principal right action on a principal fibre bundle.\\
$\rho_g$ && For a group action $P\times G\ra{\rho}P$, for a manifold $P$ and a Lie group $G$, the map $P\ra{\rho_g}P$; $p\mapsto\rho(p,g)$\\
$\rho^p$ && For a group action $P\times G\ra{\rho}P$, for a manifold $P$ and a Lie group $G$, the map $G\ra{\rho^p}P$; $g\mapsto\rho(p,g)$\\
$\mathrm{Ric}$ && For a Levi-Civita connection, the Ricci tensor\\
$\mathcal{S}$ && For a Levi-Civita connection, the curvature scalar\\
$T_c$ && Parallel transport map along a path $c$\\
$\mathrm{tr}^{i,j}$ && Trace of a tensor on a (pseudo)-Riemannian manifold over input slot $i$ and $j$\\
$VB^k_\Bbbk$ && Category of $k$-dimensional $\Bbbk$-vector bundles and maps of vector bundles that are invertible on the fibre\\
$VF\ra{}F$ && Vertical bundle of $F$, if $F\ra{\pi}M$ is a fibre bundle\\
$X_\hor$ && Horizontal lift of a vector field $X\in\mathcal{X}(M)$ to a fibre bundle $F\ra{\pi}M$ on which we have specified a connection\\
$\mathcal{X}(M)$ && Sheaf of vector fields on a manifold $M$\\
$Z_M$ && For a Lie algebra action $\g \ra{} \mathcal{X}(M)$ on a manifold $M$, the image of $Z\in\g$ under this action
\end{tabularx}
\end{document}